\documentclass{jams-l}

\usepackage{amssymb}
\usepackage[makeroom]{cancel}
\usepackage{graphicx}
\usepackage[cmtip,all]{xy}
\usepackage[usenames,dvipsnames]{xcolor}
\usepackage{upgreek}
\usepackage{multirow}
\usepackage{hyperref}
\definecolor{linkcolour}{rgb}{0,0.2,0.6}
\hypersetup{colorlinks,breaklinks,urlcolor=linkcolour, linkcolor=linkcolour}
\usepackage{multicol}
\usepackage{colortbl}
\definecolor{bobby}{gray}{0}
\definecolor{lightblue}{rgb}{0.88,1,1}
\usepackage{arydshln}
\usepackage{tikz}
\usepackage{soul}
\usepackage{ textcomp }

\usepackage[bb=boondox]{mathalfa} 

\newtheorem{theorem}{Theorem}[section]

\newtheorem{proposition}[theorem]{Proposition}

\theoremstyle{definition}
\newtheorem{definition}[theorem]{Definition}
\newtheorem{convention}[theorem]{Convention}

\newtheorem{example}[theorem]{Example}

\newtheorem{generation}[theorem]{Generation}
\newtheorem{publication}[theorem]{Publication}
\newtheorem{encryption}[theorem]{Encryption}
\newtheorem{decryption}[theorem]{Decryption}
\newtheorem{application}[theorem]{Application}

\theoremstyle{remark}
\newtheorem{remark}[theorem]{Remark}

\numberwithin{equation}{section}

\newcommand{\llbracket}[0]{\textrm{\normalfont\textlbrackdbl}}
\newcommand{\rrbracket}[0]{\textrm{\normalfont\textrbrackdbl}}
\newcommand{\intbrackets}[1]{\llbracket#1\rrbracket}

\newcommand{\tuplebrk}[1]{(\!(#1)\!)}

\newcommand{\mapsfrom}{\mathrel{\reflectbox{\ensuremath{\mapsto}}}}

\usepackage{geometry}
\begin{document}

\large

\title{Constructing a fully homomorphic encryption scheme \\with the Yoneda Lemma}

\author{R\'{e}my Tuy\'{e}ras}
\email{rtuyeras@gmail.com}

\begin{abstract}
This paper redefines the foundations of asymmetric cryptography's homomorphic cryptosystems through the application of the Yoneda Lemma. It demonstrates that widely adopted systems, including ElGamal, RSA, Benaloh, Regev's LWE, and NTRUEncrypt, are directly derived from the principles of the Yoneda Lemma. This synthesis leads to the creation of a holistic homomorphic encryption framework, the Yoneda Encryption Scheme. Within this framework, encryption is modeled using the bijective maps of the Yoneda Lemma Isomorphism, with decryption following naturally from the properties of these maps. This unification suggests a conjecture for a unified model theory framework, offering a foundation for reasoning about both homomorphic and fully homomorphic encryption (FHE) schemes. As a practical demonstration, the paper introduces the FHE scheme ACES, which supports arbitrary finite sequences of encrypted multiplications and additions without relying on conventional bootstrapping techniques for ciphertext refreshment. This highlights the practical implications of the theoretical advancements and proposes a new approach for leveraging model theory and forcing techniques in cryptography, particularly in the design of FHE schemes.
\end{abstract}

\maketitle



\section{Introduction}

\subsection{Short presentation}
This work proposes the unification of various asymmetric homomorphic encryption schemes under a singular framework, establishing what we term the Yoneda Encryption Scheme. Beyond its role in consolidating disparate encryption methodologies, this scheme offers a versatile framework for comparing existing homomorphic encryption schemes and devising novel ones. This unified approach presents a promising avenue for constructing a comprehensive theory for fully homomorphic cryptography, addressing a critical gap in the field. Historically, the efficiency of fully homomorphic encryption schemes has been hindered by the complexity of their foundational theories, primarily relying on commutative group theory and attempting to reconcile properties from commutative ring theory (as noted in \cite{survey1,TFHE,survey2}).

To tackle the challenges associated with axiomatization, our theoretical framework draws inspiration from a category-theoretic variant of model theory known as limit sketch theory. We illustrate how assemblies of models within a given limit sketch can lead to the creation of novel encryption schemes. Importantly, these schemes offer cryptographers a guiding principle for designing cryptosystems that are not strictly reliant on commutative groups. By harnessing the expressive power of limit sketch theory, particularly in its domain of forcing techniques often articulated through reflective subcategories, we clarify how our overarching framework can incorporate these techniques to enforce specific properties within the underlying theory of a cryptosystem.

As a practical application, we introduce an unbounded fully homomorphic encryption scheme that operates within modules over rings of polynomials. This scheme represents an enhancement over prior FHE implementations as it exercises full control over noise accumulation during iterated operations. Throughout this paper, our goal is to establish the groundwork for integrating more sophisticated forcing techniques into cryptography, thereby unlocking novel avenues for designing efficient fully homomorphic encryption schemes.

\subsection{Background}
Homomorphic encryption (HE) is an encryption technique that converts a message $m$ into a ciphertext $c$ while maintaining specific arithmetic operations, notably multiplication and addition. This feature is particularly crucial for securing computations on untrusted servers. Fully homomorphic encryption (FHE) further extends this capability to encompass both addition and multiplication operations simultaneously. This advancement is significant as it enables the execution of logical operations such as $\mathsf{AND}$ and $\mathsf{OR}$ on encrypted data. Consequently, algorithms can be performed on untrusted remote servers while the data remains securely encrypted. The applications of such encryption methods are diverse, often finding utility in machine learning computations and collaborative work conducted on cloud platforms (refer to \cite{survey2,Chillotti2020NewCF} for a comprehensive overview).
\bigskip

At a more formal level, HE explores encryption schemes that exhibit \emph{malleability}, indicating the persistence of a given operation $\square$ throughout the encryption-decryption process $(\mathsf{Enc},\mathsf{Dec})$. This characteristic can be mathematically expressed for a generic operation $\square$ as follows (refer to \cite{survey1,survey2}):
\[
\mathsf{Dec}(\mathsf{Enc}(m_1) \square \mathsf{Enc}(m_2)) = m_1 \square m_2
\]

To provide historical context for homomorphic encryption, extensive surveys, as cited in \cite{survey2,survey1}, meticulously document its evolution over time. However, for the specific focus of this paper, our discussion will center on the fundamental background necessary to grasp the key concepts outlined in the current work.

HE schemes fall into three categories based on their proximity to the definition of a fully homomorphic encryption scheme: partially homomorphic (PHE), somewhat homomorphic (SWHE), or fully homomorphic (FHE) schemes. The pursuit of FHE schemes, considered the pinnacle of cryptographic achievement, reached its culmination with Gentry's groundbreaking work in 2009 \cite{Gen09}. Following this pivotal breakthrough, the field of FHE schemes underwent four distinct \emph{generations} of refinement \cite{DGHV10,BV11,GSW13,CKKS16}. This transformative evolution was essentially propelled by the introduction of a technique known as \emph{bootstrapping} by Gentry.
\bigskip

To define the concept of bootstrapping, let us consider an encryption scheme comprising an encryption algorithm $\mathsf{Enc}:K_0 \times M \to C$ and a decryption algorithm $\mathsf{Dec}:K_1 \times C \to M$. Recall that the encryption-decryption protocol yields the relations $\mathsf{Enc}(k^0,m) = c$ and $\mathsf{Dec}(k^1,c) = m$ for a public key $k^0$ and a private key $k^1$. In the context of Gentry's framework, we also desire an encryption algorithm $\mathsf{Enc}$ to satisfy the following equation for a given set $S$ of pairs $(\square_C, \square_E)$ of algebraic operations:
\[
\mathsf{Enc}(k^0,m_1) \square_C \mathsf{Enc}(k^0,m_2) = \mathsf{Enc}(k^0,m_1 \square_E m_2)
\]
However, in practice, for the encryption $\mathsf{Enc}$ to be secure, it is imperative to introduce a controlled amount of noise $r$ into the mappings of the encryption $\mathsf{Enc}$. As a result, utilizing operations on encrypted data usually results in the accumulation of more noise within the encryption process such that a long algebraic operation
\[
F(\mathsf{Enc}(k^0,m_1),\dots,\mathsf{Enc}(k^0,m_n))
\]
may not be of the form $\mathsf{Enc}(k^0,F(m_1,\dots,m_n))$. To achieve full homomorphic functionality, the accumulated noise must be eradicated through the decryption of data. However, decrypting this data directly on the remote server poses a significant security risk, exposing sensitive information to untrusted third parties. Gentry's ingenious insight was to recognize that the decryption algorithm $\mathsf{Dec}$ could be expressed as an arithmetic operation $\square_{\mathsf{Dec}}$, formed by combinations of operations from the set $S$.

This led to the concept that encrypted data could be shielded by an additional layer of encryption, provided by a key pair $(q^0, q^1)$. This pair would be utilized to encrypt the private key $k^1$ for secure transmission to the server. Consequently, a sequence of operations $F$ operating on encrypted data could be ``refreshed'' through an initial encryption using the key $q^0$. In practical terms, the server would only need to compute the following operation on encrypted information to obtain a freshly encrypted result:
\begin{align*}
&\mathsf{Enc}(q^0,k^1) \square_{\mathsf{Dec}} \mathsf{Enc}(q^0,F(\mathsf{Enc}(k^0,m_1),\dots,\mathsf{Enc}(k^0,m_n))) \\
&= \mathsf{Enc}(q^0, k^1 \square_{\mathsf{Dec}} F(\mathsf{Enc}(k^0,m_1),\dots,\mathsf{Enc}(k^0,m_n)))\\
&= \mathsf{Enc}(q^0, k^1 \square_{\mathsf{Dec}} \mathsf{Enc}(k^0,F(m_1,\dots,m_n)))\\
&= \mathsf{Enc}(q^0, F(m_1,\dots,m_n))
\end{align*}
Most existing FHE schemes incorporate bootstrapping techniques to ensure correctness over arbitrary-depth computations, often at significant computational cost. This paper introduces a novel framework that achieves FHE without relying on Gentry’s bootstrapping technique, as the refresh operation does not depend on the public key $q^0$. Furthermore, the refresh operation relies on a transformation of the secret key $k^1$, whose public characterization exhibits significantly greater randomness than $k^1$ itself. Our refresh operation also diverges from techniques found in more recent literature, such as those employed in integer-based FHE schemes with digit decomposition and deep circuits \cite{BGV12}, or in torus-based FHE schemes using lookup tables \cite{TFHE}. For example, the transformation of the secret key $k^1$ eliminates the need for deep circuits to process the refresh operation, and we avoid the need for lookup tables by introducing an affine decomposition technique on ciphertexts. While this decomposition method may initially appear similar in spirit to Gentry’s squashing technique \cite{Gen09}, it differs notably in that it avoids decomposing secret-key information, thereby ensuring both the security and efficiency of the technique.

Finally, the framework that supports our scheme is derived from the Yoneda Lemma \cite{MacLane}, offering a new perspective on the structural foundations of FHE and providing insight into the formal principles underlying bootstrapping techniques. In particular, the Yoneda Lemma establishes a profound connection between theory and models, paralleling Gentry's proposition that the operation $\mathsf{Dec}$ should be \emph{realizable} within the model used for encrypting messages. While we apply the Yoneda Lemma at a practical level, this paper posits that many concepts and techniques commonly used in homomorphic encryption are grounded in model theory, category theory, and, ultimately, forcing techniques.

\subsection{Motivation and roadmap}
In essence, the Yonedian formalism serves as a fundamental tool for distinguishing intrinsic elements of a cryptosystem from those introduced through external concepts. Concepts forced into the cryptosystem, which do not align with the underlying theory, often hinder the emergence of homomorphic properties. Once established, the Yonedian formalism becomes instrumental in deriving desired equations and relations that can be realized within the underlying theory. This mirrors Gentry's approach, where the decryption algorithm is implemented within the algebraic structure associated with the space of ciphertexts. Thus, by systematically deconstructing the process of imposing specific properties on a theory, we enable the deduction of appropriate homomorphic properties.

In the context of this paper, the Yonedian formalism significantly enhances our comprehension of homomorphic properties within polynomial rings. Through the lens of a category-theoretic formalism, we have identified a suitable theory for post-quantum homomorphic encryption schemes. Notably, this formalism emphasizes the importance of viewing rings of polynomials as $\mathbb{Z}[X]$-modules, elucidating the relationships between the module and its corresponding ring $\mathbb{Z}[X]$, as well as the action of the ``free'' commutative group on both the ring and the module. This clarification has led us to identify a binary operation (Definition \ref{def:boxtimes:channel}) based on $3$-tensors in $\mathbb{Z}$, from which a fully homomorphic encryption scheme naturally follows (section \ref{sec:FHE:from-Yoneda}).

The paper is structured into three main sections. In the initial section (\ref{sec:preparation}), our aim is to provide or review the necessary category-theoretical background essential for comprehending the paper. While we shall not delve into all formal calculations, our objective is to establish a comprehensive bridge between readers with diverse backgrounds, whether in category theory or cryptography. Within this first section, we revisit fundamental concepts such as limits, cones, universal cones, limit sketches, and various versions of the Yoneda Lemma.

The first version, articulated in Theorem \ref{theo:YonedaLemma}, reinstates the conventional statement applicable to categories of functors. The second version, presented in Theorem \ref{theo:YonedaLemma2}, slightly extends the first version by incorporating colimits, making it more applicable to a broader spectrum of cryptosystems. The third and final version, outlined in Theorem \ref{theo:YonedaLemma3}, generalizes the Yoneda Lemma to reflective subcategories. Although not a novel result, this version is less frequently mentioned in the existing literature.

Throughout section \ref{ssec:limit-sketch} and section \ref{ssec:modules}, we delve into various applications of the Yoneda Lemma across different categories, preparing the reader for the calculatory considerations employed in the subsequent sections.

Following this preliminary section, section \ref{sec:Cryptosystems-and-Yoneda} introduces the general formalism of the paper, referred to as the Yoneda Encryption Scheme (section \ref{ssec:Yoneda-encryption}). This encryption scheme serves as the foundation for describing and categorizing common cryptosystems, including ElGamal (section \ref{ssec:ElGamal}), RSA (section \ref{ssec:RSA}), Benaloh (section \ref{ssec:Benaloh}), NTRU (section \ref{ssec:NTRU}), and LWE (section \ref{ssec:LWE-based}). The intention is not only theoretical but also practical, providing readers with a hands-on exercise in identifying scenarios where each facet of the Yoneda Encryption Scheme can be effectively applied.

Lastly, in section \ref{sec:applications}, we leverage the insights gained from the diverse applications of the Yoneda Encryption Scheme (discussed in section \ref{sec:Cryptosystems-and-Yoneda}) to establish a theoretical foundation for a fully homomorphic encryption system. This system is named the Arithmetic Channel Encryption Scheme (ACES). We define this scheme using the Yoneda formalism in section \ref{sec:FHE:from-Yoneda}. Subsequently, in section \ref{ssec:ACES:homomorphic-properties}, we employ a slightly more general formalism to show that ACES defines a leveled fully homomorphic encryption scheme (refer to Theorem \ref{theo:Yoneda:to:leveled-FHE} and Proposition \ref{prop:FHE:decryption}). Then, in section \ref{ssec:proper-FHE}, we add a refresh operation to our arithmetic, which turns ACES into a proper fully homomorphic encryption scheme (see Theorem \ref{theo:Yoneda:to:proper-FHE}).

To conclude, our central contribution extends beyond providing a comprehensive theory for homomorphic encryption schemes. Notably, our key result asserts that ACES achieves non-leveled full homomorphism without relying on traditional bootstrapping techniques. For those interested in exploring the practical implications of our findings, we have developed a user-friendly Python package supporting this theoretical paper, accessible at \url{https://github.com/remytuyeras/aces}.

\section{Acknowledgment}
The author would like to thank Mark Schultz, Martti Karvonen, Sam Jaques, Matan Prasma for pertinent questions and feedback regarding earlier versions on this paper.

\section{Preparation}\label{sec:preparation}

\subsection{Conventions}
We assume the reader's familiarity with fundamental concepts of category theory, including categories, functors, and natural transformations, among other classical definitions. For a comprehensive introduction to these concepts, we recommend consulting \cite{MacLane}. This section serves not only to recap less obvious concepts but also to establish the notations and conventions consistently employed throughout the paper.

\begin{convention}[Homsets and small categories]
For every category $\mathcal{C}$, and every pair $(X,Y)$ of objects in $\mathcal{C}$, we will denote as $\mathcal{C}(X,Y)$ the set of arrows from $X$ to $Y$. Every set of arrows of the form $\mathcal{C}(X,Y)$ will be referred to as a \emph{homset}. We will say that $\mathcal{C}$ is \emph{small} if the class of objects of $\mathcal{C}$ is a proper set (as opposed to a class). For every category $\mathcal{C}$ and every object $X$, we will denote as $\mathsf{id}_X$ the identity arrow $X \to X$ in $\mathcal{C}$.
\end{convention}

\begin{convention}[Opposite category]
For every category $\mathcal{C}$, we will denote as $\mathcal{C}^{\mathsf{op}}$ the category whose objects are those of $\mathcal{C}$ and whose arrows from a given object $X$ to a given object $Y$ are the elements of the homsets $\mathcal{C}(Y,X)$. The composition of $\mathcal{C}^{\mathsf{op}}$ is the one directly inherited from $\mathcal{C}$.
\end{convention}

\begin{convention}[Opposite functors]
The opposite operation on categories extends to functors as follows: For every pair $(\mathcal{C},\mathcal{D})$ of categories and every functor $F:\mathcal{C} \to \mathcal{D}$, we denote as $F^{\mathsf{op}}:\mathcal{C}^{\mathsf{op}} \to \mathcal{D}^{\mathsf{op}}$ the obvious functor that sends an arrow $f:Y \to X$ in $\mathcal{C}^{\mathsf{op}}$ to the arrow $F(f):F(Y) \to F(X)$ in $\mathcal{D}^{\mathsf{op}}$.
\end{convention}

\begin{convention}[Sets]
We denote as $\mathbf{Set}$ the category of sets and functions. For every non-negative integer $n$, we will denote as $[n]$ the set of integers from $1$ to $n$. If $n$ is zero, then the set $[n]$ is the empty set.
\end{convention}

\begin{convention}[Functors]
For every pair $(\mathcal{C},\mathcal{D})$ of categories, we will denote as $[\mathcal{C},\mathcal{D}]$ the category whose objects are functors $\mathcal{C} \to \mathcal{D}$ and whose arrows are the natural transformations between functors $\mathcal{C} \to \mathcal{D}$.
\end{convention}

\begin{definition}[Constant functor]
Let $A$ be a small category. For every category $\mathcal{C}$, we will denote as $\Delta_A$ the functor $\mathcal{C} \to [A,\mathcal{C}]$ the obvious functor that sends an object $X$ in $\mathcal{C}$ to the constant function $A \to \mathbf{1} \to \mathcal{C}$ that maps every object $a$ in $A$ to the object $X$ in $\mathcal{C}$. For every object $X$ in $\mathcal{C}$, the functor $\Delta_A(X)$ sends arrows in $A$ to identities on $X$. For every arrow $f:X \to Y$ in $\mathcal{C}$, the natural map $\Delta_A(f)$ are given by copies of the arrow $f:X \to y$ for each object $a$ in $A$.
\end{definition}

\begin{definition}[Limits and colimits]\label{def:limits_colimits}
Let $A$ be a small category and $\mathcal{C}$ be a category. The category $\mathcal{C}$ will be said to have \emph{limits over $A$} if the functor $\Delta_A:\mathcal{C} \to [A,\mathcal{C}]$ is equipped with a right adjoint $\mathsf{lim}_A:[A,\mathcal{C}] \to \mathcal{C}$. This means that for every object $X$ in $\mathcal{C}$ and every functor $F:A \to \mathcal{C}$, the category $\mathcal{C}$ is equipped with a natural isomorphism as follows:
\begin{equation}\label{eq:limits_colimits:lim}
[A,\mathcal{C}](\Delta_A(X),F) \cong \mathcal{C}(X,\mathsf{lim}_A(F))
\end{equation}
Conversely, the category $\mathcal{C}$ will be said to have \emph{colimits over $A$} if the functor $\Delta_A:\mathcal{C} \to [A,\mathcal{C}]$ is equipped with a left adjoint $\mathsf{col}_A:[A,\mathcal{C}] \to \mathcal{C}$. This means that for every object $X$ in $\mathcal{C}$ and every functor $F:A \to \mathcal{C}$, the category $\mathcal{C}$ is equipped with a natural isomorphism as follows:
\begin{equation}\label{eq:limits_colimits:colim}
[A,\mathcal{C}](F,\Delta_A(X)) \cong \mathcal{C}(\mathsf{col}_A(F),X)
\end{equation}
\end{definition}

\begin{example}[Sets]\label{exa:set_co_limits}
The category $\mathbf{Sets}$ has limits and colimits over all small categories. As a result, we can show (see \cite{MacLane}) that for every small category $[T,\mathbf{Set}]$ has limits and colimits over all small categories (take these limits and colimits to be the pointwise limits and colimits).
\end{example}

It follows from a straightforward reformulation of the definition of limits and colimits in $\mathbf{Set}$ and that of natural transformations that the following isomorphism always exist.

\begin{proposition}\label{prop:limit_colimits_as_limits_in_set}
Let $A$ be a small category and $\mathcal{C}$ be a category. If $\mathcal{C}$ has limits over $A$, then there is a natural isomorphism (in $X$ and $F:A \to \mathcal{C}$), where the rightmost limit is defined on the functorial mapping $a \mapsto \mathcal{C}(X,F(a))$
\[
\mathcal{C}(X,\mathsf{lim}_A(F)) \cong \mathsf{lim}_A(\mathcal{C}(X,F(-)))
\]
Similarly, if $\mathcal{C}$ has colimits over $A$, then there is a natural isomorphism (in $X$ and $F:A \to \mathcal{C}$), where the rightmost limit is defined on the functorial mapping $a \mapsto \mathcal{C}(F(a),X)$
\[
\mathcal{C}(\mathsf{col}_A(F),X) \cong \mathsf{lim}_{A^{\mathsf{op}}}(\mathcal{C}(F^{\mathsf{op}}(-),X))
\]
\end{proposition}
\begin{proof}
The isomorphisms of the statememt follow from the isomorphisms given in Definition \ref{def:limits_colimits} and the straightforward correspondences given by the following isomorphisms.
\[
\begin{array}{lll}
[A,\mathcal{C}](\Delta_A(X),F) &\to& \mathsf{lim}_A(\mathcal{C}(X,F))\\
(f_a:X \to F(a))_a &\mapsto & (f_a:X \to F(a))_a
\end{array}
\quad
\textrm{and}
\quad
\begin{array}{lll}
[A,\mathcal{C}](F,\Delta_A(X)) &\to& \mathsf{lim}_A(\mathcal{C}(F,X))\\
(f_a: F(a) \to X)_a &\mapsto & (f_a:F(a) \to X)_a
\end{array}
\]
As can be seen through the previous mappings, an element of $[A,\mathcal{C}](\Delta_A(X),F)$ or $[A,\mathcal{C}](F,\Delta_A(X))$ can be interpreted as a tuple in the limit $\mathsf{lim}_A(\mathcal{C}(X,F))$ or $\mathsf{lim}_A(\mathcal{C}(F,X))$, respectively.
\end{proof}

\begin{definition}[Full and faithful]
Let $\mathcal{C}$ and $\mathcal{D}$ be two categories.
\begin{itemize}
\item[1)] A functor $F:\mathcal{C} \to \mathcal{D}$ is said to be \emph{full} if for every pair $(X,Y)$ of objects in $\mathcal{C}$, the functor $F$ induces a surjective map $\mathcal{C}(X,Y) \to \mathcal{D}(F(X),F(Y))$;
\item[2)] A functor $F:\mathcal{C} \to \mathcal{D}$ is said to be \emph{faithful} if for every pair $(X,Y)$ of objects in $\mathcal{C}$, the functor $F$ induces an injective map $\mathcal{C}(X,Y) \to \mathcal{D}(F(X),F(Y))$.
\end{itemize}
\end{definition}

\subsection{Yoneda lemma for functors} \label{ssec:Yoneda_lemma}
Consider a small category $T$. Recall that the homsets of $T$ induce a functor $T^{\mathsf{op}} \to [T,\mathbf{Set}]$ through the nested mapping rules $Y: a \mapsto (b \mapsto T(a,b))$ (see \cite{MacLane}). Intuitively, this functor serves as an embedding of the ``theory'' into the category $[T,\mathbf{Set}]$, which represents the set of models for the theory $T$. For each object $a$ in $\mathcal{C}$, the functor $Y(a):T \mapsto \mathbf{Set}$ is said to be \emph{representable}, as it is represented by the object $a$ in the small category $T$.

\begin{theorem}[Yoneda Lemma]\label{theo:YonedaLemma}
For every object $a$ in $T$ and every functor $F:T \mapsto \mathbf{Set}$, there is a function
\[
\varphi_{a,F}:[T,\mathbf{Set}](Y(a),F) \mapsto F(a),
\]
natural in $a$ and $F$, that maps a natural transformation $\theta:Y(a) \Rightarrow F$ to the element $\theta(\mathsf{id}_a) \in F(a)$, where $\mathsf{id}_a$ denotes the identity arrow in the homset $T(a,a)$. The function $\varphi_{a,F}$ is a bijection whose inverse is defined by the function
\[
F(a) \mapsto [T,\mathbf{Set}](Y(a),F)
\]
that maps an element $x \in F(a)$ to the natural transformation $Y(a) \Rightarrow F$ that sends, for every object $b$ in $T$, an arrow $f \in T(a,b)$ to the element $F(f)(x) \in F(b)$.
\end{theorem}
\begin{proof}
The statement given above can be proven through a straightforward verification. The reader can learn more about the Yoneda Lemma in \cite{MacLane}.
\end{proof}

By taking the functor $F$, in the statement of Theorem \ref{theo:YonedaLemma}, to be another representable functor, say $Y(b)$ for some object $b$ in $T$, we can show that the functor $Y:T^\mathsf{op} \to [T,\mathbf{Set}]$ is \emph{fully faithfulness} (see \cite{MacLane} for instance). It follows from this property that, for every small category $I$ and every functor $G:I \to [T,\mathbf{Set}]$, if the image $G(i)$ is a representable functor of the fome $Y(a_i)$, then the mapping $i \mapsto a_i$ induces a functor $H_G:I \to T^\mathsf{op}$ such that the following diagram commutes.
\[
\xymatrix{
I \ar@{-->}[r]^{H_G} \ar[rd]_{G} & T^\mathsf{op}\ar[d]^{Y}\\
&[T,\mathbf{Set}]
}
\]
Below, we use the universal property of colimits to generalize the Yoneda Lemma to colimits of representable functors. In this respect, let $I$ be a small category. Since $[T,\mathbf{Set}]$ has colimits over $I$ (Example \ref{exa:set_co_limits}), for every functor $H:I \to T^{\mathsf{op}}$, the colimit $\mathsf{col}_I(Y \circ H)$ of the composite functor $Y \circ H:I \to \mathcal{C}$ is well-defined.

\begin{theorem}[Yoneda Lemma]\label{theo:YonedaLemma2}
Take $I$ and $H:I \to T^{\mathsf{op}}$ as defined above. For every functor $F:T \mapsto \mathbf{Set}$, there is a bijection
\[
\varphi_{H,F}:[T,\mathbf{Set}](\mathsf{col}_I(Y \circ H),F) \mapsto \mathsf{lim}_{I^{\mathsf{op}}}(F \circ H^{\mathsf{op}}),
\]
natural in the variables $H$ and $F$ over the categories $[I,T^{\mathsf{op}}]$ and $[T,\mathbf{set}]$,
\end{theorem}
\begin{proof}
The proof results from the composition of the natural isomorphisms given by the Yoneda Lemma (Theorem \ref{theo:YonedaLemma}) and that given in Proposition \ref{prop:limit_colimits_as_limits_in_set}. Specifically, we obtain the natural isomorphism of the statement from the following sequence of isomorphisms:
\begin{align*}
[T,\mathbf{Set}](\mathsf{col}_I(Y \circ H),F) & \cong  \mathsf{lim}_{I^{\mathsf{op}}}([T,\mathbf{Set}](Y \circ H^{\mathsf{op}}(-),F))& (\textrm{Proposition \ref{prop:limit_colimits_as_limits_in_set}})\\
& \cong  \mathsf{lim}_{I^{\mathsf{op}}}(F \circ H^{\mathsf{op}}) & (\textrm{Yoneda Lemma})
\end{align*}
The composition of the two natural isomorphisms provides the natural isomorphism $\varphi_{H,F}$.
\end{proof}

\subsection{Yoneda lemma for models of theories}
This section extends section \ref{ssec:Yoneda_lemma} to reflective subcategories of functor categories. In practice, these subcategories will correspond to categories whose objects are limit-preserving functors.

\begin{definition}[Reflective subcategory]\label{def:reflective_subcategory}
Let $T$ be a small category. We will say that a category $\mathcal{L}$ is a reflective subcategory of $[T,\mathbf{Set}]$ if it is equipped with a full inclusion functor $U:\mathcal{L} \hookrightarrow [T,\mathbf{Set}]$ such that the functor $U$ has a left adjoint $L:[T,\mathbf{Set}] \to \mathcal{L}$.
\end{definition}

\begin{proposition}[Colimits]\label{prop:colimits_ref_sub}
Let $T$ be a small category and let $U:\mathcal{L} \hookrightarrow [T,\mathbf{Set}]$ be a reflective subcategory whose left adjoint is denoted as $L$. The category $\mathcal{L}$ has colimits, which are the images of the colimits of $[T,\mathbf{Set}]$ via the functor $L:[T,\mathbf{Set}] \to \mathcal{L}$.
\end{proposition}
\begin{proof}
The proof directly follows from the well-known property that left adjoints preserve colimits (see \cite{MacLane}). Specifically, the statement follows from the following isomorphisms, where $I$ is a small category, $X$ is an object in $\mathcal{L}$ and $F$ is a functor $I \to \mathcal{L}$.
\begin{align*}
[I,\mathcal{L}](F,\Delta_A(X)) &\cong [I,[T,\mathbf{Set}]](U \circ F,\Delta_A(U(X))) & (\textrm{full sbucategory})\\
&\cong [T,\mathbf{Set}](\mathsf{col}_I(U \circ F),U(X)) & (\textrm{colimits})\\
&\cong \mathcal{L}(L(\mathsf{col}_I(U \circ F)),X) & (\textrm{adjunction})
\end{align*}
We conclude from the characterization of colimits given in Definition \ref{def:limits_colimits}.
\end{proof}

\begin{theorem}[Yoneda Lemma]\label{theo:YonedaLemma3}
Let $I$ and $T$ be small categories and let $H:I \to T^{\mathsf{op}}$ be a functor. Let also $R:\mathcal{L} \hookrightarrow [T,\mathbf{Set}]$ denote a reflective subcategory with an adjoint $L:[T,\mathbf{Set}] \to \mathcal{L}$. For every functor $F:T \mapsto \mathbf{Set}$ in $\mathcal{L}$, there is a bijection
\[
\phi_{H,F}:\mathcal{L}(\mathsf{col}_I(L \circ Y \circ H),F) \mapsto \mathsf{lim}_{I^{\mathsf{op}}}(U(F ) \circ H^{\mathsf{op}}),
\]
natural in the variables $H$ and $F$ over the categories $[I,T^{\mathsf{op}}]$ and $\mathcal{L}$,
\end{theorem}
\begin{proof}
The proof results from the composition of the natural isomorphisms given by the version of the Yoneda Lemma stated in Theorem \ref{theo:YonedaLemma2} and that resulting from Definition \ref{def:reflective_subcategory}. Specifically, we obtain the natural isomorphism of the statement from the following sequence of isomorphisms:
\begin{align*}
\mathcal{L}(\mathsf{col}_I(L \circ Y \circ H),F) & \cong \mathcal{L}(L(\mathsf{col}_I(U \circ  L\circ Y\circ H)),F) & (\textrm{Proposition \ref{prop:colimits_ref_sub}}) \\
& \cong [T,\mathbf{Set}](\mathsf{col}_I(U \circ  L \circ Y\circ H),U(F)) & (\textrm{Definition \ref{def:reflective_subcategory}}) \\
& \cong [I,[T,\mathbf{Set}]](U \circ L \circ Y\circ H,\Delta_I(U(F))) & (\textrm{Isomorphism (\ref{eq:limits_colimits:colim})}) \\
& \cong [I,\mathcal{L}](L \circ Y\circ H,\Delta_I(F)) & (\textrm{full subcategory}) \\
& \cong [I,[T,\mathbf{Set}]](Y\circ H,\Delta_I(U(F))) & (\textrm{adjunction}) \\
& \cong [T,\mathbf{Set}](\mathsf{col}_I(Y\circ H),U(F))  & (\textrm{Isomorphism (\ref{eq:limits_colimits:colim})}) \\
& \cong  \mathsf{lim}_{I^{\mathsf{op}}}(U(F) \circ H^{\mathsf{op}}) & (\textrm{Theorem \ref{theo:YonedaLemma2}})
\end{align*}
The composition of this sequence of natural isomorphisms provides the natural isomorphism $\phi_{H,F}$.
\end{proof}

\subsection{Limit sketches}\label{ssec:limit-sketch}
This section defines a particular type of reflective subcategories whose structures will be used to construct cryptosystems or recover established ones. More precisely, these structures commonly reinstate categories of models pertaining to a predefined theory equipped with limits. Although our subsequent exposition leverages the general concept of limits for mathematical convenience and to minimize unnecessary hypotheses, it is noteworthy that the majority of our examples of reflective subcategories stem from highly specific limits, known as products. Nevertheless, the versatility of the formalism presented herein, which accommodates any limit, implies the potential for further adaptation and application beyond the examples presented.

\begin{definition}[Cones]\label{def:cones}
Let $A$ be a small category and $\mathcal{C}$ be a category. For every object $X$ in $\mathcal{C}$ and every functor $F:A \to \mathcal{C}$, we will refer to a natural transformation of the form $\Delta_A(X) \Rightarrow F$ as a \emph{cone} above $F$. Suppose that $\mathcal{C}$ has  limits over $A$. Then, we will say that the cone $\Delta_A(X) \Rightarrow F$ is a \emph{universal cone} (or a \emph{limit cone}) if its image through isomorphism (\ref{eq:limits_colimits:lim}) is an isomorphism of the form $X \to \mathsf{lim}_A(F)$ in $\mathcal{C}$.
\end{definition}

\begin{remark}[Universal cone]\label{rem:universal-cone:adjunction}
For every object $F$ in $[A,\mathcal{C}]$, isomorphism (\ref{eq:limits_colimits:lim}) allows us to easily construct a universal cone as the inverse image of the identity on $\mathsf{lim}_A(F)$ in $\mathcal{C}$, as shown below.
\[
\begin{array}{ccc}
[A,\mathcal{C}](\Delta_A(\mathsf{lim}_A(F)),F) &\cong &\mathcal{C}(\mathsf{lim}_A(F),\mathsf{lim}_A(F))
\\
(\mathsf{lim}_A(F) \to F(a))_a & \mapsfrom & \mathsf{id}_{\mathsf{lim}_A(F)}
\end{array}
\]
Each inverse image $\Delta_A(\mathsf{lim}_A(F))\Rightarrow F$ obtained in this way describes a component of the unit $\Delta_A \circ \mathsf{lim}_A\Rightarrow \mathsf{Id}$ for the adjunction $\Delta_A\vdash \mathsf{lim}_A$ referred to in Definition \ref{def:limits_colimits}.
\end{remark}

The following proposition can be used to verify that a given cone is universal.

\begin{proposition}[Universal cones]\label{prop:universal-cones}
Let $A$ be a small category and $\mathcal{C}$ be a category. For every object $X$ in $\mathcal{C}$ and every functor $F:A \to \mathcal{C}$, a cone $\alpha: \Delta_A(X) \Rightarrow F$ is universal in $\mathcal{C}$ if, and only if, for every cone $\beta: \Delta_A(Y) \Rightarrow F$ in $\mathcal{C}$, there exists a unique morphism $f:Y \to X$ such that the following diagram commutes in $[A,\mathcal{C}]$.
\[
\xymatrix{
\Delta_A(X)\ar@{<==}[d]_{\Delta_A(f)}\ar@{=>}[r]^{\alpha}&F&\\
\Delta_A(Y)\ar@{=>}[ru]_{\beta}&
}
\]
\end{proposition}
\begin{proof}
Suppose that $\alpha: \Delta_A(X) \Rightarrow F$ is universal and let $a:X \to \mathsf{lim}_A(F)$ denote the corresponding isomorphism through isomorphism (\ref{eq:limits_colimits:lim}) (Definition \ref{def:cones}). If we let $\eta_F:\Delta_A(\mathsf{lim}_A(F))\Rightarrow F$ denote the universal cone obtained using the identity on $\mathsf{lim}_A(F)$ (see Remark \ref{rem:universal-cone:adjunction}), then the naturality of isomorphism (\ref{eq:limits_colimits:lim}) tells us that $\alpha = \eta_F \circ \Delta_A(a)$. Let now $Y$ be an object in $\mathcal{C}$. Isomorphism (\ref{eq:limits_colimits:lim}) gives us the following sequence of isomorphisms.
\begin{align*}
[A,\mathcal{C}](Y,X)& \cong [A,\mathcal{C}](Y,\mathsf{lim}_A(F)) & (\textrm{use }a:X \to \mathsf{lim}_A(F))\\
& \cong [A,\mathcal{C}](\Delta_A(Y),F) & (\textrm{Isomorphism (\ref{eq:limits_colimits:lim})}\\
\end{align*}
This means that, for every cone $\beta: \Delta_A(Y) \Rightarrow F$, we have an arrow $f:X \to Y$ for which the naturality of isomorphism (\ref{eq:limits_colimits:lim}) implies that $\eta_F \circ \Delta_A(a \circ f) = \beta$. As a result, we have $\alpha \circ \Delta_A(f) = \beta$. The arrow $f$ is unique since the previous sequence of isomorphisms from $[A,\mathcal{C}](Y,X)$ to $[A,\mathcal{C}](\Delta_A(Y),F)$ states that $g \mapsto \alpha \circ \Delta_A(g)$ is a bijection.

Conversely, let us show that any cone $\alpha: \Delta_A(X) \Rightarrow F$ that satisfies the property stated in the statement is universal. First, because $\eta_F$ is universal (Remark \ref{rem:universal-cone:adjunction}), the first part of this proof shows that there exists a (unique) morphism $f:\mathsf{lim}_A(F) \to X$ such that $\alpha = \eta_F \circ \Delta_A(f)$. It also follows from the property satisfied by $\alpha$ that there exists a (unique) morphism $f':X \to \mathsf{lim}_A(F)$ such that $\eta_F = \alpha \circ \Delta_A(f')$. This gives us $\alpha =  \alpha \circ \Delta_A(f' \circ f)$. By assumption on $\alpha$, the arrow $f' \circ f$ is unique and is hence equal to the identity on $\mathsf{lim}_A(F)$. Similarly, since $\eta_F$ is universal, we can show that $f \circ f' = \mathsf{id}_X$ and hence $f$ is an isomorphism. Given that $\alpha = \eta_F \circ \Delta_A(f)$, it follows from Definition \ref{def:cones} and the definition of $\eta_F$ that $\alpha$ is universal.
\end{proof}

In general, a full subcategory $L \hookrightarrow [T,\mathbf{Set}]$ whose objects are functors $T \to \mathbf{Set}$ that sends certain chosen cones in $T$ to universal cones in $\mathbf{Set}$ are reflective subcategories (see \cite{freydkelly,kelly_1980,adamek_rosicky_1994,ElimQuot}).

\begin{definition}[Limit sketches]
We will use the term \emph{limit sketch} to refer to a small category $T$ equipped with a subset of its cones. We define a \emph{model} for a limit sketch $T$ as a functor $T \mapsto \mathbf{Set}$ that sends the chosen cones of $T$ to univercal cones in $\mathbf{Set}$. For every limit sketch $T$,  we will denote the category whose objects are models for $T$ and whose arrows are all the natural transformations between them as $\mathbf{Mod}(T)$. This category will be referred to as the category of models for the limit sketch (also called the \emph{theory}) $T$.
\end{definition}

\begin{remark}\label{rem:model_category}
For every limit sketch $T$, the inclusion $\mathbf{Mod}(T) \hookrightarrow [T,\mathbf{Set}]$ defines a reflective subcategory. As a result, Theorem \ref{theo:YonedaLemma3} holds.
\end{remark}

\begin{example}[Magmas]\label{exa:model1}
A \emph{magma} consists of a set $M$ and an operation $\star:M \times M \to M$. A morphism between two magmas $(M_1,\star_1)$ and $(M_2,\star_2)$ consists of a function $m:M_1 \to M_1$ such that the following diagram commutes.
\[
\xymatrix{
M_1 \times M_1\ar[d]_{\star_1}\ar[r]^m &M_2 \times M_2 \ar[d]^{\star_1}\\
M_1 \ar[r]_m& M_2
}
\]
Let us now denote as $T_{\mathsf{magma}}$ the small category generated by the following graph.
\[
\begin{array}{c}
\xymatrix{
d_2 \ar[r]^{f_1} \ar[d]_{f_2} & d_1 \\
d_1 &
}
\end{array}
\begin{array}{c}
\xymatrix{
d_2 \ar[r]^{t} & d_1
}
\end{array}
\]
If we equip $T_{\mathsf{magma}}$ with the cone $f_1,f_2:d_2 \rightrightarrows d_1$, then for every functor $M:T_{\mathsf{magma}} \to \mathbf{Set}$ in $\mathbf{Mod}(T_{\mathsf{magma}})$, the universality of Cartesian products gives the following correspondence:
\[
M(d_2) = M(d_1) \times M(d_1)
\]
As a result, the category of models $M$ in $\mathbf{Mod}(T_{\mathsf{magma}})$ corresponds to the category of magmas and their morphisms, such that the magma operation is provided by function $M(t):M(d_1) \times M(d_1) \to M(d_1)$.
\end{example}

\begin{remark}[Magma of binary trees]\label{rem:magma_binary_trees}
Recall that the Yoneda embedding provides a functor $Y(d_1):T_{\mathsf{magma}} \to \mathbf{Set}$. While this functor does not belong to the category $\mathbf{Mod}(T_{\mathsf{magma}})$, its image $L\circ Y(d_1)$ via the functor $L:[T_{\mathsf{magma}},\mathbf{Set}] \to \mathbf{Mod}(T_{\mathsf{magma}})$ does. The model $L\circ Y(d_1)$ can be recursively constructed from the images of the functor $Y(d_1)$ through a transfinite colimit (see \cite{ElimQuot}). Specifically, this transfinite construct \emph{forces} the presence of missing elements into the model to turn the functor $Y(d_1)$ into a model for the ``theory'' $T_{\mathsf{magma}}$.

To appreciate what the model $L\circ Y(d_1)$ looks like, let us identify the elements missing from the primary set of the functor $Y(d_1)$, specifically $Y(d_1)(d_1) = T_{\mathsf{magma}}(d_1,d_1)$. According to the definition of the limit sketch $T_{\mathsf{magma}}$, the sole arrow $d_1 \to d_1$ in $Y(d_1)(d_1)$ is the identity $\mathsf{id}_{d_1}$. As a result, constructing the model $L\circ Y(d_1)$ necessitates the formal multiplication of $\mathsf{id}_{d_1}$ with itself to be part of the primary set of the model. As outlined in \cite{ElimQuot}, this element is encoded by the tuple $(\mathsf{id}_{d_1},\mathsf{id}_{d_1},t)$, visually represented as a tree:
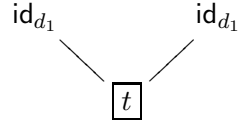
\[
\xymatrix@-10pt{
\mathsf{id}_{d_1}&&\mathsf{id}_{d_1}\\
&\fbox{$t$}\ar@{-}[ul]\ar@{-}[ur]&
}
\]
Iterating this formal process of generating multiplications with previously generated elements reveals that the model $L\circ Y(d_1)$ can be conceptualized as the magma of binary trees with no linear branches (e.g., no consecutive one-branch forkings):
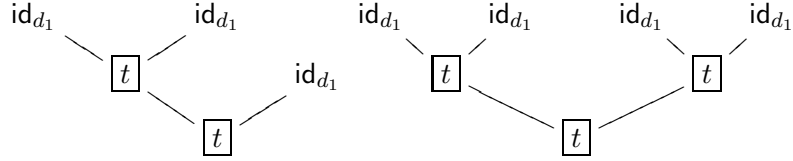
\[
\xymatrix@C-10pt@R-20pt{
\mathsf{id}_{d_1}&&\mathsf{id}_{d_1}&\\
&\fbox{$t$}\ar@{-}[ul]\ar@{-}[ur]&&\mathsf{id}_{d_1}\\
&&\fbox{$t$}\ar@{-}[ul]\ar@{-}[ur]&
}
\xymatrix@-20pt{
\mathsf{id}_{d_1}&&\mathsf{id}_{d_1}&&\mathsf{id}_{d_1}&&\mathsf{id}_{d_1}\\
&\fbox{$t$}\ar@{-}[ul]\ar@{-}[ur]&&&&\fbox{$t$}\ar@{-}[ul]\ar@{-}[ur]&\\
&&&\fbox{$t$}\ar@{-}[ull]\ar@{-}[urr]&&&
}
\]
This magma $L\circ Y(d_1)$, which we denote as $\mathbb{T}\mathsf{ree}$, can be used to capture elements in any magma $M$ by utilizing the Yoneda Lemma (Theorem \ref{theo:YonedaLemma3}), as expressed below:
\[
\mathbf{Mod}(T_{\mathsf{magma}})(\mathbb{T}\mathsf{ree},M) \cong M
\]
\end{remark}

\begin{example}[Semigroups]\label{exa:model2}
A \emph{semigroup} consists of a magma $(M,\star)$ whose operation is associative. This means that the magma operation must satisfies the following commutative diagram.
\[
\xymatrix@C+20pt{
M \times M \times M\ar[d]_{\mathsf{id}_M\times \star}\ar[r]^-{\star\times \mathsf{id}_M} &M  \times M \ar[d]^{\star}\\
M  \times M \ar[r]_-{\star}& M
}
\]
We can define the small category $T_{\mathsf{sgroup}}$ as the small category generated by the following graph structure.
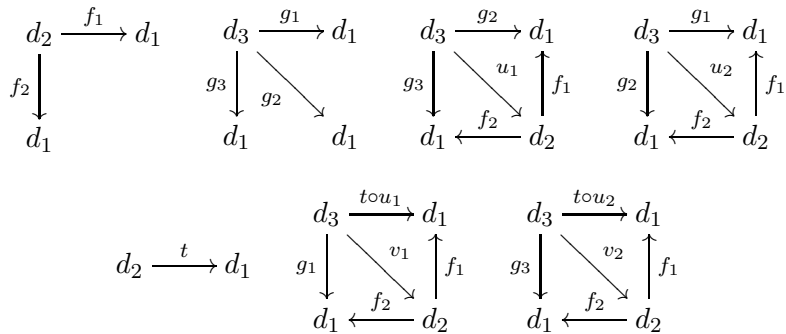
\[
\begin{array}{c}
\xymatrix{
d_2 \ar[r]^{f_1} \ar[d]_{f_2} & d_1 \\
d_1 &
}
\end{array}
\begin{array}{c}
\xymatrix{
d_3 \ar[r]^{g_1} \ar[d]_{g_3} \ar[rd]_{g_2} & d_1 \\
d_1 & d_1
}
\end{array}
\begin{array}{c}
\xymatrix{
d_3 \ar[r]^{g_2} \ar[d]_{g_3} \ar[rd]^{u_1} & d_1\\
d_1 & d_2 \ar[u]_{f_1} \ar[l]_{f_2}
}
\end{array}
\begin{array}{c}
\xymatrix{
d_3 \ar[r]^{g_1} \ar[d]_{g_2} \ar[rd]^{u_2} & d_1\\
d_1 & d_2 \ar[u]_{f_1} \ar[l]_{f_2}
}
\end{array}
\]
\[
\begin{array}{c}
\xymatrix{
d_2 \ar[r]^{t} & d_1
}
\end{array}
\begin{array}{c}
\xymatrix{
d_3 \ar[r]^{t \circ u_1} \ar[d]_{g_1} \ar[rd]^{v_1} & d_1\\
d_1 & d_2 \ar[u]_{f_1} \ar[l]_{f_2}
}
\end{array}
\begin{array}{c}
\xymatrix{
d_3 \ar[r]^{t \circ u_2} \ar[d]_{g_3} \ar[rd]^{v_2} & d_1\\
d_1 & d_2 \ar[u]_{f_1} \ar[l]_{f_2}
}
\end{array}
\]
\[
\xymatrix{
d_3\ar[d]_{v_1}\ar[r]^-{v_2} &d_2 \ar[d]^{t}\\
d_2 \ar[r]_-{t}& d_1
}
\]
If we equip $T_{\mathsf{sgroup}}$ with the cone $(f_1,f_2)$ and $(g_1,g_2,g_3)$, then for every functor $M:T_{\mathsf{sgroup}} \to \mathbf{Set}$ in $\mathbf{Mod}(T_{\mathsf{sgroup}})$, the universality of Cartesian products gives the following correspondences:
\[
M(d_2) = M(d_1) \times M(d_1) \quad\quad\textrm{and}\quad\quad M(d_3) = M(d_1) \times M(d_1) \times M(d_1)
\]
The universality of products also implies that we have the following equations:
\[
M(v_1) = \mathsf{id}_{M(d_1)} \times M(t)  \quad\quad\textrm{and}\quad\quad M(v_2) = M(t) \times \mathsf{id}_{M(d_1)}
\]
As a result, the category of models $M$ in $\mathbf{Mod}(T_{\mathsf{sgroup}})$ corresponds to the category of semigroups and their morphisms, such that the semigroup operation is provided by function $M(t):M(d_1) \times M(d_1) \to M(d_1)$.
\end{example}

\begin{remark}[Semigroup of positive integers]
In a manner akin to how we obtained a generative model for magmas in Remark \ref{rem:magma_binary_trees}, we can show that the semigroup encoded by the functor $L \circ Y(d_1):T_{\mathsf{sgroup}} \to \mathbf{Set}$ is isomorphic to the magma $\mathbb{T}\mathsf{ree}$ under an associative multiplication. This can be interpreted as the semigroup defined by the set $\mathbb{N}^{+}$ of positive integers.
\end{remark}

\begin{example}[Monoids]\label{exa:model3}
A \emph{monoid} is a semigroup $(M,\star)$ equipped with a neural element for the semigroup operation.  This means that it is equipped with a map $\mathbb{1}: \mathbf{1} \to M$ picking a distinguishing element in $M$ (where $\mathbf{1}$ is a singleton) such that the following diagrams commute.
\[
\xymatrix@C+20pt{
\mathbf{1} \times M \ar[rd]_{\mathsf{pr}_2}\ar[r]^{\mathbb{1} \times \mathsf{id}_M}&  M \times M \ar[d]^{\star} & \ar[l]_{\mathsf{id}_M \times \mathbb{1}} M \times \mathbf{1} \ar[ld]^{\mathsf{pr}_1}\\
&M&
}
\]
Because a singleton such as $\mathbf{1}$ is a limit over the empty set, it is possible to define a limit sketch $T_{\mathsf{monoid}}$ such that the category $\mathbf{Mod}(T_{\mathsf{monoid}})$ corresponds to the category of monoids and their morphisms.
\end{example}

\begin{remark}[Monoid of non-negative integers]
Taking $d_1$ to be the object of $T_{\mathsf{monoid}}$ representing the underlying set of the models in $\mathbf{Mod}(T_{\mathsf{monoid}})$, we can show that the monoid generated by the model $L \circ Y(d_1):T_{\mathsf{monoid}} \to \mathbf{Set}$ is the set $\mathbb{N}$ of non-negative integers.
\end{remark}

\begin{example}[Groups]\label{exa:model4}
A \emph{group} is a  monoid $(M,\star,\mathbb{1})$ equipped with an inverse operation $u:M \to M$ for the multiplication $\star$. This means that the following diagram must commute.
\[
\xymatrix@C+20pt{
M \times M \ar[d] \ar[r]^{u \times \mathsf{id}_M} & M \times M \ar[d]^{\star} & \ar[l]_{\mathsf{id}_M \times u} M \times M \ar[d]\\
\mathbf{1} \ar[r]_{\mathbb{1}} &M& \mathbf{1} \ar[l]^{\mathbb{1}}
}
\]
Similarly to what has been done so far with other examples of categories of models (Examples \ref{exa:model1}, \ref{exa:model2}, \ref{exa:model3}), we can show that there is a limit sketch $T_{\mathsf{group}}$ for which the category $\mathbf{Mod}(T_{\mathsf{group}})$ corresponds to the category of groups and their morphisms.
\end{example}

\begin{remark}[Group of integers]
Taking $d_1$ to be the object of $T_{\mathsf{group}}$ representing the underlying set of the models in $\mathbf{Mod}(T_{\mathsf{group}})$, we can show that the monoid generated by the model $L \circ Y(d_1):T_{\mathsf{group}} \to \mathbf{Set}$ is the set $\mathbb{Z}$ of all integers.
\end{remark}

\begin{example}[Commutative structures]\label{exa:model5}
The theories defined in \ref{exa:model1}, \ref{exa:model2}, \ref{exa:model3} and \ref{exa:model4} can be extended to commutative structures by adding structural maps that recover the usual symmetry axiom shown below, where $\gamma$ is the universal symmetry bijection.
\[
\xymatrix{
M \times M \ar[d]_{\star} \ar[r]^{\gamma}_{\cong} & M \times M \ar[d]^{\star}\\
M\ar@{=}[r]&M
}
\]
Below, we will denote the resulting limit sketch for commutative groups as $T_{\mathsf{cgroup}}$.
\end{example}

\begin{remark}[Group of integers]\label{rem:model5}
Taking $d_1$ to be the object of $T_{\mathsf{cgroup}}$ representing the underlying set of the models in $\mathbf{Mod}(T_{\mathsf{cgroup}})$, we can show that the commutative group generated by the model $L \circ Y(d_1)$ is the set $\mathbb{Z}$ of all integers.
\end{remark}

\begin{convention}[Notation]\label{conv:modulo-set}
For every positive integer $n$, we will denote as $\mathbb{Z}_n$ the commutative group whose elements are given by the set $[n-1]$ and whose group operation is given by the addition of integers modulo $n$, namely:
\[
\begin{array}{ccl}
\mathbb{Z}_n \times \mathbb{Z}_n & \to & \mathbb{Z}_n\\
(x,y) & \mapsto & x + y\,(\mathsf{mod}\,n)
\end{array}
\]
The inverse operation for this addition is the negation operation modulo $n$, namely $x \mapsto n-x$.
\end{convention}

\subsection{Modules}\label{ssec:modules}
A central application of the Yoneda Lemma in our context involves leveraging the limit sketch of modules over a given ring. The aim of this section is to examine the algebraic properties associated with these objects and formulate key propositions that will be instrumental in transforming cryptosystems over polynomials into Yoneda encryption schemes. To define module structures, we will build upon the concepts introduced in the previous section as well as the definition of rings (Definition \ref{def:rings}).

\begin{definition}[Rings]\label{def:rings}
A \emph{ring} is defined as a set $R$ equipped with:
\begin{itemize}
\item[1)] a commutative group structure $(R, \oplus, \mathbb{0})$;
\item[2)] and a monoid structure $(R, \otimes, \mathbb{1})$,
\end{itemize}
such that the two diagrams shown below commute, where $\delta: R \to R \times R$ denotes the obvious diagonal morphism and $\gamma: R \!\times\! R\! \times\! R\! \times\! R \to R \!\times\! R\! \times\! R\! \times\! R$ denotes the symmetry $(a,b,c,d) \mapsto (a,c,b,d)$:
\[
\xymatrix@C+25pt{
R \!\times\! R \!\times\! R \ar[r]^-{\gamma \circ (\delta \times \mathsf{id}_{(R \times R)})}\ar[d]^{\mathsf{id}_R \times \oplus} & *+!L(0.7){R \!\times\! R\! \times\! R\! \times\! R} \ar[r]^-{\otimes \times \otimes} & R \!\times\! R\ar[d]_{\oplus}\\
R \!\times\! R \ar[rr]_{\otimes} && R
}
\quad
\xymatrix@C+25pt{
R \!\times\! R \!\times\!R \ar[r]^-{\gamma \circ (\mathsf{id}_{(R \times R)} \times \delta)}\ar[d]^{\oplus \!\times\! \mathsf{id}_R} & *+!L(0.7){R \!\times\! R\! \times\! R\! \times\! R} \ar[r]^-{\otimes \times \otimes} & R \times R\ar[d]_{\oplus}\\
R \!\times\! R \ar[rr]_{\otimes} && R
}
\]
These diagrams enforce that the monoid multiplication $\otimes$ is distributive with respect to the group addition $\oplus$. This means that the equations $a \otimes (b+c) = a \otimes b + a \otimes c$ and $(b+c) \otimes a = b \otimes a + c \otimes a$ hold for every triple $(a,b,c)$ of elements in $R$.
\end{definition}

\begin{convention}[Limit sketches]
In much the same fashion as we derived limit sketches for theories characterized by commutative diagrams in section \ref{ssec:limit-sketch}, it is straightforward to show that rings are precisely described by models of a limit sketches $T_{\mathsf{ring}}$.
\end{convention}

\begin{remark}[Yoneda objects]
Let $d_1$ denote the object of $T_{\mathsf{ring}}$ that captures the underlying set of the models in $\mathsf{Mod}(T_{\mathsf{ring}})$. It follows from the distributivity axioms that the model $L \cdot Y(d_1):T_{\mathsf{ring}} \to \mathbf{Set}$ is the ring of integers given by $\mathbb{Z}[X]$.
\end{remark}

\begin{remark}[Morphisms]
A morphism of rings is defined as a natural transformation between models in $\mathsf{Mod}(T_{\mathsf{ring}})$. The Yoneda Lemma (Theorem \ref{theo:YonedaLemma3}) gives the isomorphism:
\[
\mathsf{Mod}(T_{\mathsf{ring}})(\mathbb{Z}[X],R) \cong R.
\]
At a fundamental level, this isomorphism means that a morphism $f$ of the form $\mathbb{Z}[X] \to R$ is characterized by an evaluation at a specific element $r \in R$. Indeed, since the ring $R$ contains all integer representatives, which are derived from the unit $\mathbb{1}$, the underlying function defining $f$ is given by the mapping $p(X) \mapsto p(r)$ where $p(X)$ is a polynomial over integers and $p(r)$ is its evaluation at $r \in R$.
\end{remark}

We now leave the realm of rings and mostly uses the fact that every ring $(R,\oplus,\otimes,\mathbb{0},\mathbb{1})$ provides an underlying commutative group  $(R,\oplus,\mathbb{0})$.

\begin{convention}[Powers]\label{conv:ring-as-group:powers}
For every non-negative integer $n$ and ring $(R,\oplus,\otimes,\mathbb{0},\mathbb{1})$, we will denote the $n$-fold Cartesian product of $R$ in $\mathbf{Mod}(T_{\mathsf{cgroup}})$ as $R^{(n)}$. For a chosen limit structure on $\mathbf{Mod}(T_{\mathsf{cgroup}})$, this commutative group is the limit $\mathsf{lim}_{[n]}(F)$ of the functor $F:[n] \to \mathbf{Mod}(T_{\mathsf{cgroup}})$ picking out $n$ copies of the commutative group $(R,\oplus,\mathbb{0})$. We can show that the underlying set for the resulting group structure is given by an $n$-fold Cartesian product of $R$ in $\mathbf{Set}$, and the associated addition is a componentwise extension of the addition $\oplus$.
\end{convention}

\begin{convention}[From rings to limit sketches]\label{conv:ring-to-sketch}
For every ring $(R,\oplus,\otimes,\mathbb{0},\mathbb{1})$, we will denote as $T_{R}$ the small whose objects are non-negative integers and whose morphisms $n \to m$ are given by group morphisms
\[
f = (f_1,\dots,f_m):R^{(n)} \to R^{(m)}
\]
such that, for every $k \in [m]$, the Cartesian component $f_k:R^{(n)} \to R$ satisfies the following equation in the ring $(R,\oplus,\otimes,\mathbb{0},\mathbb{1})$ for every$(r_1,\dots,r_n) \in R^{(n)}$:
\[
f_k(r_1,\dots,r_n) = f_k(\mathbb{1},\mathbb{0},\dots,\mathbb{0}) \otimes r_1 \oplus \dots \oplus f_k(\mathbb{0},\dots,\mathbb{0},\mathbb{1}) \otimes r_n
\]
It is straightforward to verify that this notion of morphisms defines a category structure for $T_R$. Indeed, the morphism $f$ defined above can be likened to a matrix product in $R$ and since matrix products in $R$ are associative, the composition of morphisms in $T_R$ is also associative. In other words, the morphisms of $T_R$ are linear maps of the form shown in (\ref{eq:matrix:T_R}), where $M$ is a $(m \times n)$- matrix in $R$ and where $e_i$ is the element $(\mathbb{0},\dots,\mathbb{0},\mathbb{1},\mathbb{0},\dots,\mathbb{0}) \in R^{(n)}$ whose $i$-th coefficient is $\mathbb{1}$.
\begin{equation}\label{eq:matrix:T_R}
f(r) = Mr = \Big(\bigoplus_{j=1}^n f_i(e_j) \otimes r_j\Big)_i\quad\quad M = (f_i(e_j))_{i,j})
\end{equation}
The category $T_R$ includes morphisms of the following form (where $s$ denotes any element in $R$):
\[
s:
\left(
\begin{array}{lll}
R^{(1)}&\to&R^{(1)}\\
r & \mapsto & s \otimes r
\end{array}
\right)
\quad\quad
\oplus:
\left(
\begin{array}{lll}
R^{(2)}&\to&R^{(1)}\\
(r,s) & \mapsto & r \oplus s
\end{array}
\right)
\quad\quad
\mathsf{pr}_i^n:
\left(
\begin{array}{lll}
R^{(n)}&\to&R^{(1)}\\
(r_1,\dots,r_n) & \mapsto & r_i
\end{array}
\right)
\]
We will assume that the small category $T_R$ is equipped with a limit sketch structure whose cones are the collections $(\mathsf{pr}_i^n:R^{(n)} \to R^{(1)})_{i = 1}^n$ for every non-negative integer $n$. Note that when $n=0$, we have $R^{(n)} = \{0\}$ and the resulting cone is given by the obvious inclusion $0:\{0\} \hookrightarrow R$.

According to the definition of $T_R$, we have an isomorphism $T_R(1,1) \cong R$. While this might imply that the functor $Y(1):T_R \to \mathbf{Set}$ represents the ring $R$, we will in fact see that this correspondence imbues the functor $Y(1)$ with a ``module structure'' over the ring $R$ (see Definition \ref{def:modules}). To expand on this, observe that the definition of $T_R$ yields the following isomorphisms for every non-negative integer $n$:
\[
T_R(1,n) \cong R^{(n)} \cong T_R(1,1) \times T_R(1,1) \times \dots \times T_R(1,1)
\]
At an elementary level,  this isomorphism asserts that the image of the cone $(\mathsf{pr}_i^n:R^{(n)} \to R^{(1)})_{i = 1}^n$ via the functor $k \mapsto T_R(1,k)$ is a universal cone in $\mathbf{Set}$. As a result, we conclude that the functor $Y(1):T_R \to \mathbf{Set}$ defines an object in $\mathbf{Mod}(T_R)$.
\end{convention}

\begin{definition}[Modules]\label{def:modules}
For every ring $(R,\oplus,\otimes,\mathbb{0},\mathbb{1})$, we define an \emph{$R$-module} as an object in $\mathbf{Mod}(T_R)$. It follows from Convention \ref{conv:ring-to-sketch} that, for every $R$-module $M:T_R \to \mathbf{Set}$, the arrows
\[
s:R^{(1)} \to R^{(1)} \quad\quad\quad\oplus:R^{(2)} \to R^{(1)}\quad\quad\quad\mathsf{pr}_0^1:R^{(0)} \to R^{(1)}
\]
in $T_R$ are sent via the functor $M$ to functions of the form
\[
s \odot :M(1) \to M(1)\quad\quad\quad \oplus_M:M(1) \times M(1) \to M(1) \quad\quad\quad \mathbb{0}_M:\mathbf{1} \to M(1)
\]
making the following diagram commute, where $\delta:M(1) \to M(1) \times M(1)$ the obvious diagonal morphism.
\[
\xymatrix@C+20pt{
 M(1) \times M(1) \ar[r]^{(r \odot) \times (s \odot)} &M(1) \times M(1) \ar[d]^{\oplus_M} \\
M(1) \ar[u] ^{\delta}\ar[r]_{(r \oplus s) \odot}& M(1)
}
\quad\quad\quad
\xymatrix@C+20pt{
 M(1) \times M(1)  \ar[r]^{(r \odot) \times (r \odot)}  \ar[d]^{\oplus_M}&M(1) \times M(1) \ar[d]^{\oplus_M} \\
M(1) \ar[r]_{(r \odot}& M(1)
}
\]
\[
\xymatrix@C+20pt{
M(1)\ar[r]^{r \odot} \ar[dr]_{(s \otimes r) \odot } &M(1) \ar[d]^{s \odot} \\
& M(1)
}
\quad\quad\quad
\xymatrix@C+20pt{
\mathbf{1} \ar[r]^{\mathbb{0}_M} \ar[dr]_{\mathbb{0}_M} & M(1) \ar[d]^{r \odot} \\
 & M(1)
}
\]
The four preceding diagrams encapsulate the four types of axioms typically stipulated for $R$-modules, as conventionally defined in the literature.
\end{definition}

\begin{remark}[Yoneda objects]\label{rem:Yoneda-object:modules}
The conclusions of Convention \ref{conv:ring-to-sketch} imply that, for every ring $(R,\oplus,\otimes,\mathbb{0},\mathbb{1})$, the functor $Y(1):T_R \to \mathbf{Set}$ recovers the $R$-module structure defined the commutative group $(R,\oplus,\mathbb{0})$. Specifically, for every element $s \in R$, the function $Y(1)(s):R \to R$ recovers the multiplication $r \mapsto s \otimes r$ and the function $Y(1)(\oplus):R \times R \to R$ recovers the addition $(r,s) \mapsto r \oplus s$.

Similarly, for every non-negative integer $n$, we can show that $Y(n):T_R \to \mathbf{Set}$ recovers the $R$-module defined by the power $R^{(n)}$. Interestingly, it turns out that the object $Y(n)$ is also the $n$-fold coproduct of the object $Y(1)$ in $\mathbf{Mod}(T_R)$. This statement is more formally verified through the Yoneda Lemma. Indeed, Theorem \ref{theo:YonedaLemma3} gives us the following isomorphism for every object $M$ in $\mathbf{Mod}(T_R)$:
\[
\mathbf{Mod}(T_R)(Y(n),M) \cong M(n)
\]
However, by the definition of a model in $\mathbf{Mod}(T_R)$, the set $M(n)$ is naturally in bijection with the $n$-fold product $M(1) \times M(1) \times \dots \times M(1)$. As a result, we have the following natural isomorphism.
\begin{equation}\label{iso:Yoneda:example:colimits}
\mathbf{Mod}(T_R)(Y(n),M) \cong M(1) \times M(1) \times \dots \times M(1)
\end{equation}
Here, some category-theoretic arguments could show that the object $Y(n)$ is the $n$-fold coproduct of the object $Y(1)$ in $\mathbf{Mod}(T_R)$, thereby providing another explanation for the previous isomorphism in the context of Theorem \ref{theo:YonedaLemma3}.
\end{remark}

\begin{remark}[Colimits in Yoneda Lemma]
Let $R$ be a ring. In Remark \ref{rem:Yoneda-object:modules}, we observed that the Yoneda functor $Y(n):T_R \to \mathbf{Set}$ corresponds to the $n$-fold coproduct of $Y(1):T_R \to \mathbf{Set}$. Consequently, it might seem like the use of colimits in Theorem \ref{theo:YonedaLemma3} is excessive, and its statement could be limited to Yoneda functors of the form $Y(k)$. However, this expectation should not be considered a general rule. For instance, for practical reasons, one might want to confine the limit sketch $T_R$ to the objects $0$, $1$, $2$ and $3$. While this restriction would be sufficient to recover the conventional definition of $R$-modules, it would lack the objects $4$, $5$, $\dots$, $n$, etc. As a result, it would be necessary to demonstrate that $Y(n)$ is a coproduct of $Y(1)$ to obtain isomorphism (\ref{iso:Yoneda:example:colimits}).
\end{remark}

\section{Cryptosystems and the Yoneda Lemma}\label{sec:Cryptosystems-and-Yoneda}

\subsection{Yoneda encryption}\label{ssec:Yoneda-encryption}
This section establishes a cryptosystem based on the Yoneda Lemma, as stated in Theorem \ref{theo:YonedaLemma3}. In the upcoming sections, we will explain how this overarching cryptosystem can reconstruct other cryptosystems. Below, we consider a small category $T$ and a reflective subcategory denoted by $U:\mathcal{L} \hookrightarrow [T,\mathbf{Set}]$ with a left adjoint $L:[T,\mathbf{Set}] \to \mathcal{L}$.
\smallskip

Before we proceed to define a Yoneda cryptosystem (Definition \ref{def:Yoneda_cryptosystem}), let us revisit the components of conventional cryptosystems. These systems typically involve three sets of elements:
\begin{itemize}
\item[-] a \emph{plaintext space} $M_1$ containing messages intended for transmission,
\item[-] a \emph{ciphertext space} $M_2$ encompassing encrypted messages used for secure communication,
\item[-] and a \emph{key space} $K$ housing both public and private keys.
\end{itemize}
Additionally, these cryptosystems feature two $K$-indexed collections of functions. One collection describes an encryption algorithm $\mathsf{E}_k: M_1 \to M_2$, responsible for transforming messages into encrypted form using a key $k$ drawn from the set $K$. The other is a decryption algorithm $\mathsf{D}_k: M_2 \to M_1$, which performs the inverse operation, decrypting the ciphertext back into the original message using some key $k$ drawn from the set $K$. In asymmetric cryptography, the keys used for encryption and decryption must differ.
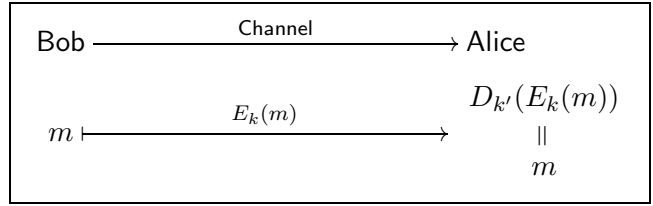
\[
\fbox{
\xymatrix@C+30pt@R-20pt{
\mathsf{Bob} \ar[rr]^-{\mathsf{Channel}} & & \mathsf{Alice}\\
m \ar@{|->}[rr]^-{E_k(m)} && *+!L(0.5){\begin{array}{c}D_{k'}(E_k(m)) \\ \rotatebox[origin=c]{90}{$=$}\\m\end{array}}
}
}
\]
To encompass a wide range of HE schemes, Yoneda cryptosystems exhibit a slightly greater level of generality than the conventional definition of cryptosystems provided earlier. For further clarification on this increased generality, refer to Remark \ref{rem:Yoneda_cryptosystem}.

\begin{definition}[Yoneda cryptosystems]\label{def:Yoneda_cryptosystem}
Let $(F,G)$ be a pair of objects in $\mathcal{L}$, let $(M_1,M_2)$ be a pair of sets, and let $e$ be an object in $T$. We say that the tuple $(F,G,M_1,M_2)$ defines a \emph{Yoneda cryptosystem} in $\mathcal{L}$ at the object $e$ if it is equipped with
\begin{itemize}
\item[a)] a set $R$ and an $R$-indexed collection $(\mathsf{E}_r)_{r \in R}$ of functions $\mathsf{E}_r:G(e) \times M_1 \to M_2$;
\item[b)] and a partial function $\mathsf{D}:G(e) \times M_2 \to M_1$,
\end{itemize}
such that every element $f \in F(e)$, there exists an element $f'\in F(e)$ for which the following equation holds for every element $r \in R$, every morphism $h \in \mathcal{L}(F,G)$ and every element $m \in M_1$.
\begin{equation}\label{eq:Yoneda-cryptosystem:D_h_E_H_m}
\mathsf{D}(h_e(f),\mathsf{E}_r(h_e(f'),m)) = m
\end{equation}
We shall refer to the resulting cryptosystem by the notation $\mathcal{Y}(F,G,M_1,M_2|R,\mathsf{E},\mathsf{D})$.
\end{definition}

\begin{convention}[Restricted Yoneda cryptosystem]\label{conv:restricted:Yoneda}
In practice, the use of Definition \ref{def:Yoneda_cryptosystem} does not require equation (\ref{eq:Yoneda-cryptosystem:D_h_E_H_m}) to hold for all morphisms $h \in \mathcal{L}(F,G)$. In fact, we will see two instances (in section \ref{ssec:NTRU} and section \ref{sec:FHE:from-Yoneda}) where we need to restrict equation (\ref{eq:Yoneda-cryptosystem:D_h_E_H_m}) to a subset $M \subseteq \mathcal{L}(F,G)$. For this reason, we will say that a Yoneda cryptosystem is \emph{restricted} along a subset $M \subseteq \mathcal{L}(F,G)$ if Definition \ref{def:Yoneda_cryptosystem} and, more specifically, its equation (\ref{eq:Yoneda-cryptosystem:D_h_E_H_m}) only holds for morphisms $h \in M$.
\end{convention}

\begin{remark}[Context]\label{rem:Yoneda_cryptosystem}
Let us clarify the notations introduced in Definition \ref{def:Yoneda_cryptosystem}. In this respect, let $\mathcal{Y}(F,G,M_1,M_2|R,\mathsf{E},\mathsf{D})$ be a Yoneda cryptosystem in $\mathcal{L}$ at an object $e$ in $T$, restricted along a subset $M \subseteq \mathcal{L}(F,G)$. Since $\mathcal{L}$ is a full subcategory of $[T, \mathbf{Set}]$, we can evaluate any object and morphism of $\mathcal{L}$ at the object $e$. Consequently, each morphism $h \in M$ induces a function as follows.
\[
h_e : F(e) \to G(e)
\]
This is means that, for every element $f$ in $F(e)$, we can construct a element $h_e(f)$ in $G(e)$. This construction is important role in the Yoneda cryptosystem as the key space of this cryptosystem is given by the set $G(e)$. Meanwhile, the set $F(e)$ serves as a parameter set that one can use to construct keys through morphisms of the form $h:F \Rightarrow G$ in $\mathcal{L}$.
As a result, elements like $h_e(f)$ will be used in the functions $\mathsf{E}_r : G(e) \times M_1 \to M_2$ to encode messages in $M_1$, and the function $\mathsf{D} : G(e) \times M_2 \to M_1$ will be used decrypt the resulting ciphertexts. For a given element $f \in F(e)$, the corresponding element $f' \in F(e)$ mentioned at the end of Definition \ref{def:Yoneda_cryptosystem} should be seen as a noise operation hiding important information contained in the element $f$. This noise operation can be seen as a function
\[
\alpha : \left( \begin{array}{lll} F(e) &\to &F(e)\\ f & \mapsto & f' \end{array}\right).
\]
that makes the diagram below commute for every $r \in R$ and every $h \in M$, where $\delta$ denotes the Cartesian diagonal function $F(e) \times F(e) \to F(e)$ in $\mathbf{Set}$, and $\mathsf{pr}_{2}$ denotes the Cartesian projection $F(e) \times M_1 \to M_1$:
\[
\xymatrix@C+40pt{ F(e) \times M_1 \ar[d]_{\mathsf{pr}_2}\ar[r]^-{\delta \times \mathsf{id}_{M_1}} & F(e) \times F(e) \times M_1 \ar[r]^-{h_e \times (h_e \circ \alpha) \times \mathsf{id}_{M_1}} & G(e) \times G(e) \times M_1 \ar[d]^{\mathsf{id}_{G(e)}\times \mathsf{E}_r} \\
M_1  && G(e) \times M_2 \ar[ll]^-{\mathsf{D}} }
\]
In the sequel, the function $\alpha$ will not be explicitly defined, as only the existence of an element $f'$ for a given $f$ is required. This justifies the introduction of Convention \ref{conv:reversors}.
\end{remark}

\begin{convention}[Reversors]\label{conv:reversors}
Let $\mathcal{Y}(F,G,M_1,M_2|R,\mathsf{E},\mathsf{D})$ be a Yoneda cryptosystem in $\mathcal{L}$ at an object $e$ in $T$, restricted along a subset $M \subseteq \mathcal{L}(F,G)$. For every element $f \in F(e)$, we define the subset $\mathcal{R}(f)$ of $F(e)$ as follows:
\[
\mathcal{R}(f) = \{f'\in F(e)~|~\forall r \in R,\,\forall h \in M,\,\forall m \in M_1:\mathsf{D}(h_e(f),\mathsf{E}_r(h_e(f'),m)) = m\}
\]
For every element $f \in F(e)$, an element in $\mathcal{R}(f)$ will be referred to as a \emph{reversor of $f$}.
\end{convention}

In the remainder of this section, we delve into a formal examination of the application of Yoneda cryptosystems, encompassing key generation, key publication, encryption, and decryption steps. Each of these steps will serve as a foundation for describing the other cryptosystems contained in this paper.
\bigskip

Let us start with the key generation step, which entails configuring parameters for the publication of the public key. This public key is intended for use by every sending party, commonly referred to as $\mathsf{Bob}$, to encrypt messages destined for transmission to the receiving party, usually referred to as $\mathsf{Alice}$.

\begin{generation}[Yoneda Encryption Scheme]\label{gen:YES}
Let $\mathcal{Y}(F,G,M_1,M_2|R,\mathsf{E},\mathsf{D})$ be a Yoneda cryptosystem in $\mathcal{L}$ at an object $e$ in $T$, restricted along a subset $M \subseteq \mathcal{L}(F,G)$. The key generation step for this cryptosystem is defined as follows.
$\mathsf{Alice}$ chooses:
\begin{itemize}
\item[-] a small category $I$ and a functor $H:I \to T^{\mathsf{op}}$;
\item[-] an element $f_0$, called the \emph{initializer}, in the limit $\mathsf{lim}_{I^{\mathsf{op}}}(F\circ H^{\mathsf{op}})$ computed in $\mathbf{Set}$;
\item[-] and an element $x$, called the \emph{private key}, in the set $(\mathsf{col}_I(L \circ Y \circ H))(e)$, where the colimit $\mathsf{col}_I(L \circ Y \circ H)$ is computed in $\mathcal{L}$ and evaluated at the object $e$ in $T$.
\end{itemize}
At this point, the tuple $(I,F,f_0,x)$ is solely known to $\mathsf{Alice}$. Nevertheless, the subsequent publication step will elucidate which components are intended to be considered public and private information.
\end{generation}

\begin{remark}[Key generation]\label{rem:key-generation}
The setup introduced in Generation \ref{gen:YES} is closely related to the statement of Theorem \ref{theo:YonedaLemma3}. Indeed, note that any element $f_0$ taken in the limit $\mathsf{lim}_{I^{\mathsf{op}}}(F\circ H^{\mathsf{op}})$ defines an element in the image of the bijection $\phi_{H,F}$ (see Theorem \ref{theo:YonedaLemma3}).
\[
\phi_{H,F}:\mathcal{L}(\mathsf{col}_I(L \circ Y \circ H),F) \mapsto \mathsf{lim}_{I^{\mathsf{op}}}(U(F ) \circ H^{\mathsf{op}}),
\]
This means that the inverse image $\phi^{-1}_{H,F}(f_0)$ of the initializer $f_0$ defines a morphism $\mathsf{col}_I(L \circ Y \circ H) \Rightarrow F$ in $\mathcal{L}$. We can then use this morphism to generate a function as follows.
\[
\phi^{-1}_{H,F}(f_0)_{e}:\mathsf{col}_I(L \circ Y \circ H)(e) \to F(e)
\]
We can then use this function to send the private key $x \in \mathsf{col}_I(L \circ Y \circ H)(e)$ to an element $\phi^{-1}_{H,F}(f_0)_{e}(x)$ in the set $F(e)$.
\end{remark}

\begin{publication}[Yoneda Encryption Scheme]\label{pub:YES}
Let $\mathcal{Y}(F,G,M_1,M_2|R,\mathsf{E},\mathsf{D})$ be a Yoneda cryptosystem in $\mathcal{L}$ at an object $e$ in $T$, restricted along a subset $M \subseteq \mathcal{L}(F,G)$. Suppose that we have a key generation $(H,f_0,x)$ as defined in Generation \ref{gen:YES}. The publication step for such a cryptosystem is as follows. $\mathsf{Alice}$ does the following:
\begin{itemize}

\item[-] she chooses a reversor $f' \in \mathcal{R}(f)$ for the element $f := \phi^{-1}_{H,F}(f_0)_{e}(x)$ (refer to Remark \ref{rem:key-generation});
\item[-] and she sends the tuple $(f_0,f')$ to $\mathsf{Bob}$. This tuple defines the \emph{public key} of the cryptosystem.
\end{itemize}
At this point, the tuple $(f_0,f')$ along with the cryptosystem data $(F,G,M_1,M_2,R,\mathsf{E},\mathsf{D},e)$ is deemed public and, therefore, known to all parties.
\end{publication}

The next step (Encryption \ref{enc:YES}) will explain how $\mathsf{Bob}$ can send a message to $\mathsf{Alice}$ securely. Note that, at this point, only the private key $x$ defined in Generation \ref{gen:YES} is considered to be unknown to $\mathsf{Bob}$ and any third party.

\begin{remark}[Key Publication]\label{rem:key-publication}
Let $\mathcal{Y}(F,G,M_1,M_2|R,\mathsf{E},\mathsf{D})$ be a Yoneda cryptosystem in $\mathcal{L}$ at an object $e$ in $T$, restricted along a subset $M \subseteq \mathcal{L}(F,G)$. Suppose that we have a key generation $(H,f_0,x)$ as defined in Generation \ref{gen:YES}. The details outlined in Publication \ref{pub:YES} indicate that the Yoneda Lemma isomorphism $\phi^{-1}_{H,F}$ is intended for use in tandem with the initializer $f_0$ and the private key $x$. This implies that the application of the Yoneda Lemma is relevant only to $\mathsf{Alice}$ and essentially establishes an encryption barrier between $\mathsf{Alice}$ and other parties.

Although $\mathsf{Bob}$ cannot directly utilize the Yoneda Lemma isomorphism $\phi^{-1}_{H,F}$, he can leverage its image space as a proxy. For instance, consider that the initializer $f_0$ is an element of the limit $\mathsf{lim}_{I^{\mathsf{op}}}(F \circ H^{\mathsf{op}})$. Now, observe that for every morphism $h: F \Rightarrow G$ in $\mathcal{L}$, we have a corresponding morphism
\[
h_H: F \circ H^{\mathsf{op}} \Rightarrow G \circ H^{\mathsf{op}}.
\]
This morphism enables the construction of a universal function
\[
\mathsf{lim}_{I^{\mathsf{op}}}(h_H)(f_0): \mathsf{lim}_{I^{\mathsf{op}}}(F \circ H^{\mathsf{op}}) \rightarrow \mathsf{lim}_{I^{\mathsf{op}}}(G \circ H^{\mathsf{op}}),
\]
allowing $\mathsf{Bob}$ to transform the element $f_0$ in specific ways. $\mathsf{Alice}$ can then simulate this transformation on an element of $F(e)$ using the naturality properties of the Yoneda Lemma.
\end{remark}

\begin{encryption}[Yoneda Encryption Scheme]\label{enc:YES}
Let $\mathcal{Y}(F,G,M_1,M_2|R,\mathsf{E},\mathsf{D})$ be a Yoneda cryptosystem in $\mathcal{L}$ at an object $e$ in $T$, restricted along a subset $M \subseteq \mathcal{L}(F,G)$. We will consider a private key $x$, as defined in Generation \ref{gen:YES}, and a public key $(f_0,f')$, as defined in Publication \ref{pub:YES}. The encryption step for such a cryptosystem unfolds as follows. $\mathsf{Bob}$ initiates the process by selecting:
\begin{itemize}
\item[-] an element $r \in R$, called the \emph{noise parameter};
\item[-] a morphism $h \in M$, called the \emph{garbling operation};
\item[-] and an element $m \in M_1$, identified as the \emph{message}.
\end{itemize}
Subsequently, $\mathsf{Bob}$ transmits the following information to $\mathsf{Alice}$:
\begin{itemize}
\item[-] the element $c_1 = \mathsf{lim}_{I^\mathsf{op}}(h_H)(f_0)$ in $\mathsf{lim}_{I^{\mathsf{op}}}(G \circ H^{\mathsf{op}})$ (refer to Remark \ref{rem:key-publication});
\item[-] and the element $c_2 = \mathsf{E}_r(h_{e}(f'),m)$ in $M_2$.
\end{itemize}
The tuple $(c_1, c_2)$ collectively constitutes the \emph{encryption} of the message $m$.
\end{encryption}

The decryption step (Decryption \ref{dec:YES}) hinges on the intrinsic naturality property of the Yoneda Lemma. In essence, this property empowers $\mathsf{Alice}$ to decipher the message encrypted by $\mathsf{Bob}$ without requiring any knowledge about the morphism $h: F \Rightarrow G$ selected by $\mathsf{Bob}$ in Encryption \ref{enc:YES}.

\begin{decryption}[Yoneda Encryption Scheme]\label{dec:YES}
Let $\mathcal{Y}(F,G,M_1,M_2|R,\mathsf{E},\mathsf{D})$ be a Yoneda cryptosystem in $\mathcal{L}$ at an object $e$ in $T$, restricted along a subset $M \subseteq \mathcal{L}(F,G)$. We will consider a private key $x$, as defined in Generation \ref{gen:YES}, and a public key $(f_0,f')$, as defined in Publication \ref{pub:YES}. For every encrypted message $(c_1,c_2)$ sent by $\mathsf{Bob}$, as defined in Encryption \ref{enc:YES}, the decryption step unfolds as follows. $\mathsf{Alice}$ computes:
\begin{itemize}
\item[-] the element $d = \phi_{H,G}^{-1}(c_1)_{e}(x)$ in $G(e)$;
\item[-] and the element $m' = \mathsf{D}(d,c_2)$ in $M_1$;
\end{itemize}
The element $m' \in M_1$ is well-defined and corresponds to the original message $m$ sent by $\mathsf{Bob}$, as shown in the sequence of equations below:
\begin{align*}
m' &= \mathsf{D}(d,c_2)&\\
& =  \mathsf{D}(\phi_{H,G}^{-1}(c_1)_{e}(x),c_2)&(d = \phi_{H,G}^{-1}(c_1)_{e}(x))\\
& =  \mathsf{D}(\phi_{H,G}^{-1}(\mathsf{lim}_{I^\mathsf{op}}(h_H)(f_0))_{e}(x),c_2)&(c_1 = \mathsf{lim}_{I^\mathsf{op}}(h_H)(f_0))\\
& =  \mathsf{D}((h \circ \phi_{H,F}^{-1}(f_0))_{e}(x),c_2)&(\textrm{Yoneda isomorphism naturality})\\
& =  \mathsf{D}(h_{e}(\phi_{H,F}^{-1}(f_0)_{e}(x)),c_2)&(\textrm{Natural transformations})\\
& =  \mathsf{D}(h_{e}(f),c_2)&(f = \phi_{H,F}^{-1}(f_0)_{e}(x))\\
& =  \mathsf{D}(h_{e}(f),\mathsf{E}_r(h_{e}(f'),m))&(c_2 = \mathsf{E}_r(h_{e}(f'),m))\\
&= m&(f' \in \mathcal{R}(f))
\end{align*}
\end{decryption}

This concludes the definition of the Yoneda Encryption Scheme. In the following sections, we will demonstrate how this scheme can recover more specific schemes. We will commence by exploring the ElGamal encryption scheme.

\subsection{ElGamal encryption}\label{ssec:ElGamal}
The ElGamal cryptosystem, introduced by Taher Elgamal in 1985, is based on the Diffie–Hellman key exchange protocol \cite{survey1,ElGamal_OG}. The mathematical foundation of this cryptosystem relies on the complexity of the discrete logarithm problem. Throughout this section, we demonstrate how to define a Yoneda cryptosystem that recovers the ElGamal cryptosystem.

\begin{generation}[ElGamal]\label{gen:ElGamal}
The key generation step for the ElGamal cryptosystem is defined as follows. $\mathsf{Alice}$ chooses:
\begin{itemize}
\item[-] a cyclic group $G$ of order $q$ with generator $g$;
\item[-] a (random) private key $x \in [q-1]$
\end{itemize}
Given that the discrete logarithm problem in $G$ extends to all cyclic subgroups of $G$, the element $g$ is not strictly required to be a generator of the overall group structure $G$ (refer to Remark \ref{rem:generator:ElGamal} for further clarification).
\end{generation}

\begin{remark}[Generator]\label{rem:generator:ElGamal}
As highlighted in Generation \ref{gen:ElGamal}, the element $g$ need not be a generator of the group $G$. This flexibility arises from the fact that solving a discrete logarithm problem of the form $g^x = f$ inherently implies that the element $y$ resides in the cyclic subgroup $\langle g \rangle := \{1, g, g^2, \dots, g^q\}$. Consequently, even if $g$ is not a generator, the cryptosystem data can effortlessly adapt to the cyclic subgroup $\langle g \rangle$ of $G$ along with its corresponding order. Given that the private key $x$ is intended for use as an exponent of $g$, it follows that the private key $x$ should either be selected after the choice of $g$ or preferably considered in the set $\mathbb{Z}$ of all integers, rather than in the modulo set $\mathbb{Z}_q = [q-1]$. As demonstrated in Application \ref{app:generation:ElGamal}, the Yoneda cryptosystem seamlessly resolves these ambiguities.
\end{remark}

Below, Application \ref{app:generation:ElGamal} clarifies how the steps outlined in Publication \ref{gen:ElGamal} can be recovered through the Yoneda cryptosystem's generation step. To achieve this, we need to define the reflective subcategory in which the Yoneda cryptosystem operates and articulate the overarching encryption-decryption protocol.

\begin{application}[ElGamal as Yoneda]\label{app:generation:ElGamal}
Let us describe the overall setup for the Yoneda cryptosystem encoding the ELGamal cryptosystem. First, the associated limit sketch $T$ must be the limit sketch $T_{\mathsf{cgroup}}$ defined for commutative groups (see Example \ref{exa:model5}). As a result, the associated reflective subcategory is given by the inclusion
\[
\mathbf{Mod}(T_{\mathsf{cgroup}}) \hookrightarrow [T_{\mathsf{cgroup}},\mathbf{Set}]
\]
and its left adjoint $L$ (see Remark \ref{rem:model_category}). The Yoneda cryptosystem is then defined at the object $d_1$ of $T_{\mathsf{cgroup}}$ which refers to the set of elements of the group structure. Specifically, the ElGamal cryptosystem is expressed as a Yoneda cryptosystem of the form
\[
\mathcal{Y}(G,G,G(d_1),G(d_1)|\mathbf{1},\mathsf{E},\mathsf{D})
\]
where $G$ is a cyclic group of order $q$ with generator $g$. The notation $G(d_1)$ refers to the underlying set of elements for the model $G$ in $\mathbf{Mod}(T_{\mathsf{cgroup}})$. The collection of encryption algorithms is indexed by a terminal set $\mathbf{1} = \{\ast\}$ and therefore consists of a single algorithm
\[
\mathsf{E}_{\ast}:G(d_1) \times G(d_1) \to G(d_1),
\]
which is given by the mutiplication operation $(a,b) \mapsto a\cdot b$ of $G$. The decryption algorithm
\[
\mathsf{D}:G(d_1) \times G(d_1) \to G(d_1)
\]
is given by the map $(a,b) \mapsto a^{-1}\cdot b$. With such functions, we can determine the form of the set $\mathcal{R}(f)$ for every $f \in G$, as shown below.
\begin{align*}
\mathcal{R}(f) & = \{f'~|~\forall h \in \mathcal{L}(G,G),\,\forall m \in G:\,h_{d_1}(f)^{-1} \cdot h_{d_1}(f') \cdot m = m \}&(\textrm{Convention \ref{conv:reversors}})\\
& = \{f' ~|~\forall h \in \mathcal{L}(G,G):\,h_{d_1}(f') = h_{d_1}(f) \} &(\textrm{Group structure})\\
& = \{f\}&(\textrm{implied by }h = \mathsf{id}_G)
\end{align*}
The fact that $\mathcal{R}(f)$ is non-empty for all $f \in G$ shows that the Yoneda cryptosystem is well-defined. Following Generation \ref{gen:YES}, the key generation step undertaken by $\mathsf{Alice}$ for this cryptosystem consists of:
\begin{itemize}
\item[-] a functor $H:\mathbf{1} \to T_{\mathsf{cgroup}}^{\mathsf{op}}$ picking out the element $d_1$ in $T_{\mathsf{cgroup}}$;
\item[-] an initializer $f_0$ in $\mathsf{lim}_{\mathbf{1}}(G \circ H^{\mathsf{op}}) = G(d_1)$;
\item[-] and a private key $x$ in $\mathsf{col}_{\mathbf{1}}(L \circ Y \circ H)(d_1) = L(Y(d_1))(d_1) = \mathbb{Z}$ (see Remark \ref{rem:model5}).
\end{itemize}
Here, the initializer $f_0$ denotes some element $g^{x_0}$ in $G$. As suggested in Remark \ref{rem:generator:ElGamal}, this element need not be a generator of $G$ (but it should at least be different form the neutral element). Also note that the private key $x$ chosen for the Yoneda cryptosystem belongs to $\mathbb{Z}$ instead of the set $[q-1]$, as required in Generation \ref{gen:ElGamal}. This distinction is inconsequential as $x$ will serve as an exponent for some element in $G$, and can therefore be considered modulo the order $q$ of $G$.
\end{application}

Let us examine how the ElGamal key publication aligns with that of the Yoneda cryptosystem (described in Application \ref{app:generation:ElGamal} above). First, we will present the ElGamal key publication step in Publication \ref{pub:ElGamal} and then draw parallels with Publication \ref{pub:YES} in Application \ref{app:publication:ElGamal}.

\begin{publication}[ElGamal]\label{pub:ElGamal}
The key publication step for the ElGamal cryptosystem unfolds as follows. $\mathsf{Alice}$ undertakes the following actions:
\begin{itemize}
\item[-] she computes the element $f = g^x$ in the group $G$;
\item[-] and she sends the public key $(G,q,g,f)$ to $\mathsf{Bob}$.
\end{itemize}
Since the tuple $(G,q,g)$ has already been determined in Generation \ref{gen:ElGamal}, $\mathsf{Alice}$ essentially generates and transmits only the element $f$.
\end{publication}

\begin{application}[ElGamal as Yoneda]\label{app:publication:ElGamal}
The key publication step for the Yoneda cryptosystem unfolds as follows. $\mathsf{Alice}$ does the following:
\begin{itemize}
\item[-] she computes the element $f = \phi_{H,G}^{-1}(f_0)_{d_1}(x)$ in $G$. Since $\phi_{H,G}^{-1}(f_0)$ corresponds to the group homomorphism $\mathbb{Z} \to G$ that picks out the element $f_0 \in G(d_1)$ at the element $1 \in \mathbb{Z}$, the element $f$ must be equal to the power $f_0^x$ in $G$.
\item[-] Since $\mathcal{R}(f) = \{f\}$, $\mathsf{Alice}$ must send the public key $(f_0,f)$ to $\mathsf{Bob}$.
\end{itemize}
As detailed in Publication \ref{pub:YES}, $\mathsf{Alice}$ discloses both the tuple $(f_0,f)$ and the accompanying cryptosystem data $(G,q)$. In essence, the entire tuple $(G,q,f_0,f)$ becomes public knowledge for all parties involved.
\end{application}

We now address the encryption step for the ElGamal cryptosystem.

\begin{encryption}[ElGamal]
The encryption step for the ElGamal cryptosystem is organized as follows. First, $\mathsf{Bob}$ chooses the following:
\begin{itemize}
\item[-] an (random) element $h \in [q-1]$;
\item[-] a message $m \in G$.
\end{itemize}
Then, $\mathsf{Bob}$ sends the following information to $\mathsf{Alice}$:
\begin{itemize}
\item[-] the element $c_1 = g^h$ in the group $G$;
\item[-] and the element $c_2 = f^h \cdot m$ in the group $G$.
\end{itemize}
The pair $(c_1,c_2)$ constitutes the encryption of the message $m$.
\end{encryption}

According to Application \ref{app:generation:ElGamal}, the set $R$ that indexes the encryption algorithms is trivial. Consequently, the step of Encryption \ref{enc:YES} where $\mathsf{Bob}$ typically selects an element from the set $R$ is omitted in the subsequent application. Furthermore, despite the group homomorphism chosen in Application \ref{app:encryption:ElGamal} being confined to a specific form, it is important to note that the Yoneda Encryption Scheme does not impose such restrictions on the chosen form.

\begin{application}[ElGamal as Yoneda]\label{app:encryption:ElGamal}
The encryption step for the Yoneda cryptosystem unfolds as follows. $\mathsf{Bob}$ initiates the process by selecting:
\begin{itemize}
\item[-] a morphism $\tilde{h}:G \Rightarrow G$ in $\mathcal{L}$ encoded by the mapping rule $u \mapsto u^h$ for some integer $h \in [q-1]$. This mapping rule does define a group endomorphism on $G$ because $G$ is commutative;
\item[-] and a message $m \in G$.
\end{itemize}
Then, $\mathsf{Bob}$ sends the following information to $\mathsf{Alice}$:
\begin{itemize}
\item[-] the element $c_1 = \mathsf{lim}_{\mathbf{1}}(\tilde{h}_H)(f_0) = f_0^h$ in $G(d_1)$;
\item[-] and the element $c_2 = \mathsf{E}_{\ast}(\tilde{h}_{d_1}(f),m) = f^h \cdot m$ in $G(d_1)$.
\end{itemize}
The pair $(c_1,c_2)$ constitutes the encryption of the message $m$.
\end{application}

We conclude this section with the decryption steps. The correctness of both the ElGamal cryptosystems and the Yoneda cryptosystems are provided by the calculations shown in Decryption \ref{dec:YES}.

\begin{decryption}[ElGamal]
The decryption step for the Yoneda cryptosystem unfolds as follows. $\mathsf{Alice}$ computes:
\begin{itemize}
\item[-] the element $d = c_1^x$ in $G$;
\item[-] the element $m' = d^{-1}\cdot c_2$ in $G$.
\end{itemize}
Since we have the relations $d^{-1} = f_0^{-hx}$ and $c_2 = f_0^{xh} \cdot m$, the identity $m' = m$ holds.
\end{decryption}

\begin{application}[ElGamal as Yoneda]\label{app:decryption:ElGamal}
The decryption step for the Yoneda cryptosystem unfolds as follows. $\mathsf{Alice}$ computes:
\begin{itemize}
\item[-] the element $d =\phi^{-1}_{H,G}(c_1)_{d_1}(x)$ in $G$. Since $\phi_{H,G}^{-1}(c_1)$ corresponds to the group homomorphism $\mathbb{Z} \to G$ that picks out the element $c_1 \in G(d_1)$ at the element $1 \in \mathbb{Z}$, the element $d$ must be equal to the power $c_1^x$ in $G(d_1)$.;
\item[-] the element $m' = \mathsf{D}(d,c_2) = d^{-1}\cdot c_2$ in $G(d_1)$.
\end{itemize}
The identity $m' = m$ follows from Decryption \ref{dec:YES}.
\end{application}

\begin{remark}[Cramer–Shoup cryptosystems]
The capability to recover the ElGamal cryptosystem with elements $g$ that are not necessarily generators has been implicitly utilized in the literature. For example, the Cramer–Shoup cryptosystem \cite{CramerShoup} is a modification of the ElGamal cryptosystem that leverages this flexibility along with a collision-resistant hash function to enhance the system's resistance to manipulation.
\end{remark}

\subsection{RSA encryption}\label{ssec:RSA}
The RSA encryption system, conceived in 1977 by Ron Rivest, Adi Shamir, and Leonard Adleman, is grounded in the arithmetic of finite modulo sets (Convention \ref{conv:modulo-set}) and their cyclic subgroups \cite{survey1}. The mathematical underpinning of this cryptographic scheme relies on the difficulty of two challenges: (1) determining the prime factorization of a given integer and (2) identifying the logarithmic base $m$ of a power $m^e$ modulo a specified integer $n$. In this section, we detail the formulation of a Yoneda cryptosystem that reconstructs the RSA encryption scheme (starting at Generation \ref{gen:RSA}).

\begin{convention}[Coprime]\label{conv:coprime}
We say that two integers $n$ and $m$ are \emph{coprime} if their greatest common  divisor is $1$. This means that there exists two integers $x$ and $y$ such that the relation $nx+my = 1$ holds.
\end{convention}

\begin{convention}\label{conv:modulo_n_intertible}
For every positive integer $n$, we will denote as $\mathbb{Z}_n^{\times}$ the subset of $\mathbb{Z}_n$ whose elements are coprime with $n$. It is straightforward to show that $\mathbb{Z}_n^{\times}$ defines a group for the multiplication of integers since the multiplication of two relations $nx_0+m_0y_0 = 1$ and $nx_1+m_1y_1 = 1$ gives the following relation.
\[
n(nx_1x_0+x_1m_0y_0+m_1y_1x_0) + m_1m_0(y_1y_0) = 1
\]
\end{convention}

\begin{definition}[Carmichael's totient function]\label{conv:Carmichael}
For every positive integer $n$, we define the \emph{Carmichael's totient number} $\lambda(n)$ as the order of the multiplicative group $\mathbb{Z}_n^{\times}$. In other words, $\lambda(n)$ is defined as the smallest positive integer $k$ such that, for every integer $x$ coprime to $n$, the equation $x^k \equiv 1 \,(\mathsf{mod}\,n)$ holds.
\end{definition}

Let us now present the RSA cryptosystem. As in previous instances, we will delineate each step of RSA alongside its counterpart in the Yoneda encryption scheme.

\begin{generation}[RSA]\label{gen:RSA}
The key generation step for the RSA cryptosystem is outlined as follows. $\mathsf{Alice}$ makes the following selections:
\begin{itemize}
\item[-] a positive integer $n$, which should be the product of two large prime numbers;
\item[-] and an element $f_0$ in the multiplicative group $\mathbb{Z}_{\lambda(n)}^{\times}$ whose inverse $f_0^{-1}$ in $\mathbb{Z}_{\lambda(n)}^{\times}$ is readily known to $\mathsf{Alice}$.
\end{itemize}
In the context of the RSA cryptosystem, the inverse $f_0^{-1}$ is kept private to other parties.
\end{generation}

\begin{application}[RSA as Yoneda]\label{app:generation:RSA}
Let us describe the Yoneda cryptosystem that recovers the RSA cryptosystem. To start with, we take the associated limit sketch $T$ to be the terminal category $\mathbf{1} = \{\ast\}$. As a result, the associated reflective subcategory is given by the identity
\[
\mathbf{Mod}(\mathbf{1}) = [\mathbf{1},\mathbf{Set}]
\]
where the category $[\mathbf{1},\mathbf{Set}]$ is isomorphic to $\mathbf{Set}$. The Yoneda cryptosystem is defined at the obvious object in $\mathbf{1}$.
The RSA cryptosystem is then defined as a Yoneda cryptosystem of the form
\[
\mathcal{Y}(\mathbb{Z}_{\lambda(n)}^{\times},\mathbb{Z}_{\lambda(n)}^{\times},\mathbb{Z}_n,\mathbb{Z}_n|\mathbf{1},\mathsf{E},\mathsf{D})
\]
where $n$ is a positive integer. The collection of encryption algorithms is indexed by a terminal set $\mathbf{1} = \{\ast\}$ and therefore consists of a single algorithm given by the following function.
\[
\mathsf{E}_{\ast}:
\left(
\begin{array}{ccc}
\mathbb{Z}_{\lambda(n)}^{\times} \times \mathbb{Z}_n & \to & \mathbb{Z}_n\\
(g,m) & \mapsto & m^g
\end{array}
\right)
\]
It is important to note that although the function $\mathsf{E}{\ast}$ is formally defined as a mapping \[
\mathbb{Z}_{\lambda(n)}^{\times} \times \mathbb{Z}_n \to \mathbb{Z}_n,
\]
the associated parameter $\lambda(n)$ exclusively pertains to $\mathsf{Alice}$. From an external viewpoint, while other parties are aware of the general structure of the function $\mathsf{E}{\ast}$, practically computing the parameter $\lambda(n)$ from $n$ poses a considerable challenge. Consequently, any sender party utilizes the encryption algorithm as a function of the form $\mathbb{N} \times \mathbb{Z}_n \rightarrow \mathbb{Z}_n$. Similarly, the practical knowledge required for computing the decryption algorithm is exclusively within the grasp of $\mathsf{Alice}$. This function is defined by the following mapping, where the term $g^{-1}$ represents the multiplicative inverse of the element $g$ in $\mathbb{Z}_{\lambda(n)}^{\times}$.
\[
\mathsf{D}:
\left(
\begin{array}{ccc}
\mathbb{Z}_{\lambda(n)}^{\times} \times \mathbb{Z}_n & \to & \mathbb{Z}_n\\
(g,c) & \mapsto & c^{g^{-1}}
\end{array}
\right)
\]
With such functions, we can determine the form of the set $\mathcal{R}(f)$ for every $f \in \mathbb{Z}_{\lambda(n)}^{\times}$, as shown below.
\begin{align*}
\mathcal{R}(f) & = \{f'~|~\forall h \in \mathbf{Set}(\mathbb{Z}_{\lambda(n)}^{\times},\mathbb{Z}_{\lambda(n)}^{\times}),\,\forall m \in \mathbb{Z}_n:\, \big(m^{h_{\ast}(f')}\big)^{h_{\ast}(f)^{-1}} = m \textrm{ in }\mathbb{Z}_n\}&(\textrm{Convention \ref{conv:reversors}})\\
& = \{f'~|~\forall h \in \mathbf{Set}(\mathbb{Z}_{\lambda(n)}^{\times},\mathbb{Z}_{\lambda(n)}^{\times}):\, h_{\ast}(f')\cdot h_{\ast}(f)^{-1} = 1 \textrm{ in }\mathbb{Z}_{\lambda(n)}^{\times}\}&(\textrm{implied by }m \in \mathbb{Z}_n^{\times})\\
& = \{f'~|~\forall h \in \mathbf{Set}(\mathbb{Z}_{\lambda(n)}^{\times},\mathbb{Z}_{\lambda(n)}^{\times}):\, h_{\ast}(f')=  h_{\ast}(f) \textrm{ in }\mathbb{Z}_{\lambda(n)}^{\times}\}&(\textrm{Group structure})\\
& = \{f\}&(\textrm{implied by }h = \mathsf{id})
\end{align*}
The fact that $\mathcal{R}(f)$ is non-empty for all $f \in \mathbb{Z}_{\lambda(n)}^{\times}$ shows that the Yoneda cryptosystem is well-defined. Following Generation \ref{gen:YES}, the key generation step undertaken by $\mathsf{Alice}$ for this cryptosystem consists of:
\begin{itemize}
\item[-] an identity functor $H:\mathbf{1} \to \mathbf{1}$;
\item[-] an initializer $f_0$ in $\mathsf{lim}_{\mathbf{1}}(\mathbb{Z}_{\lambda(n)}^{\times} \circ H^{\mathsf{op}}) = \mathbb{Z}_{\lambda(n)}^{\times}$ whose inverse $f_0^{-1}$ is readily known to $\mathsf{Alice}$;
\item[-] and a private key $x$ in the singleton $\mathsf{col}_{\mathbf{1}}(Y \circ H)(\ast) = Y(\ast)(\ast) = \{\mathsf{id}_{\ast}\}$.
\end{itemize}
Here, the private key $x$ must be taken to be the only element of the singleton $\mathbf{1}(\ast,\ast)$. The security of the cryptosystem remains unaffected since the confidential details lie in the computation of the inverse $f_0^{-1}$ and the definition of the decryption algorithm. Both of these components pose computational challenges, making them difficult to ascertain.
\end{application}

We now describe the key publication steps for both RSA and its Yoneda form.

\begin{publication}[RSA]\label{pub:RSA}
The key publication step for the RSA cryptosystem is as follows: $\mathsf{Alice}$ transmits the element $f_0 \in \mathbb{Z}_{\lambda(n)}^{\times}$ to $\mathsf{Bob}$ and reveals general information about her cryptosystem, such as the integer $n$ and the theoretical structure of the encryption-decryption protocol.
\end{publication}

\begin{application}[RSA as Yoneda]\label{app:publication:RSA}
The key publication step for the Yoneda cryptosystem unfolds as follows. $\mathsf{Alice}$ does the following:
\begin{itemize}
\item[-] she computes the element
\[
f = \phi_{H,\mathbb{Z}_{\lambda(n)}^{\times}}^{-1}(f_0)_{\ast}(x).
\]
This element corresponds to the image of the function $\mathbf{1} \to \mathbb{Z}_{\lambda(n)}^{\times}$ picking out the element $f_0$. As a result, the element $f$ is equal to the element $f_0$.
\item[-] Since $\mathcal{R}(f) = \{f\}$, $\mathsf{Alice}$ must send the public key $(f_0,f_0)$ to $\mathsf{Bob}$.
\end{itemize}
As outlined in Publication \ref{pub:YES}, $\mathsf{Alice}$ reveals both the key $f_0$ and the theoretical structure of the cryptosystem data, withholding only the computation details of the parameters intended to remain confidential.
\end{application}

We will now outline the encryption steps. Note that RSA encryption involves a single ciphertext, whereas Yoneda's encryption comprises a pair of ciphertexts, with the first one already known to $\mathsf{Alice}$.

\begin{encryption}[RSA]\label{enc:RSA}
The encryption step in the RSA cryptosystem unfolds as follows: $\mathsf{Bob}$ selects an element $m$ from $\mathbb{Z}_n$ as the message, and then transmits the element $c = m^{f_0} \in \mathbb{Z}_n$ to $\mathsf{Alice}$.
\end{encryption}

\begin{application}[RSA as Yoneda]\label{app:encryption:RSA}
The encryption step for the Yoneda cryptosystem unfolds as follows. $\mathsf{Bob}$ initiates the process by selecting:
\begin{itemize}
\item[-] the identity function $h:\mathbb{Z}_{\lambda(n)}^{\times} \to \mathbb{Z}_{\lambda(n)}^{\times}$ in $\mathbf{Set}$;
\item[-] and a message $m \in \mathbb{Z}_n$.
\end{itemize}
Then, $\mathsf{Bob}$ sends the following information to $\mathsf{Alice}$:
\begin{itemize}
\item[-] the element $c_1 = \mathsf{lim}_{\mathbf{1}}(h_H)(f_0) = f_0$ in $\mathbb{Z}_{\lambda(n)}^{\times}$;
\item[-] and the element $c_2 = \mathsf{E}_{\ast}(h_{\ast}(f),m) = m^{f_0}$ in $\mathbb{Z}_n$.
\end{itemize}
While the pair $(c_1, c_2)$ constitutes an encryption for the message $m$, its first component $c_1 = f_0$ is information already known to $\mathsf{Alice}$. Nevertheless, this data can be utilized by $\mathsf{Bob}$ as a certificate, ensuring that $\mathsf{Alice}$ is aware of the cryptosystem being employed by $\mathsf{Bob}$.
\end{application}

Given the computational difficulty of determining $m$ for a given power $m^{f_0}$ in $\mathbb{Z}_n$ (when $f_0$ is non-trivial), the ciphertexts employed in Encryption \ref{enc:RSA} and Application \ref{app:encryption:RSA} are deemed secure. We will now proceed to outline the decryption steps for these ciphertexts.

\begin{decryption}[RSA]
The decryption step for the RSA cryptosystem goes as follows: $\mathsf{Alice}$ calculates the value $m' = c^{f_0^{-1}}$ within $\mathbb{Z}_n$, where the inverse $f_0^{-1}$ is readily known to $\mathsf{Alice}$. Given that $m'$ is equivalent to the power $(m^{f_0})^{f_0^{-1}}$, we successfully retrieve the message $m$ in $\mathbb{Z}_n$.
\end{decryption}

\begin{application}[RSA as Yoneda]\label{app:decryption:RSA}
The decryption step for the Yoneda cryptosystem unfolds as follows. $\mathsf{Alice}$ computes:
\begin{itemize}
\item[-] the element
\[
d = \phi_{H,\mathbb{Z}_{\lambda(n)}^{\times}}^{-1}(c_1)_{\ast}(x).
\]
This element corresponds to the image of function $\mathbf{1} \to \mathbb{Z}_{\lambda(n)}^{\times}$ picking out the element $c_1$. As a result, the element $d$ is equal to the element $c_1 = f_0$.;
\item[-] and the element $m' = \mathsf{D}(d,c_2) = c_2^{d^{-1}} = (m^{f_0})^{f_0^{-1}}$ in $\mathbb{Z}_n$;
\end{itemize}
The identity $m' = m$ follows from Decryption \ref{dec:YES}.
\end{application}

\subsection{Benaloh encryption}\label{ssec:Benaloh}
In this section, we consider the Benaloh cryptosystem as a representative example for a well-known class of cryptosystems, including the Goldwasser-Micali cryptosystem, the Boneh-Goh-Nissim cryptosystem and the Paillier cryptosystem \cite{survey1,GM82,BNG,Paillier}. We shall leave the translation of these cryptosystems into Yoneda cryptosystems as an exercise to the reader. The exposition provided for the Benaloh cryptosystem below can serve as a reference for this purpose.
\smallskip

The Benaloh cryptosystem, introduced by Josh Benaloh in 1994, is grounded in the complexity of the Higher Residuosity Problem \cite{survey1}. However, the original correctness of the cryptosystem suffered from a problem, which was correct in 2011 by Fousse, Lafourcade, and Alnuaimi \cite{Benaloh_revisited}. Before showing how the Yoneda Encryption Scheme can effectively recover the Benaloh cryptosystem, we will establish a few essential notations. In particular, we will use the terminology introduced in Convention \ref{conv:coprime}, Convention \ref{conv:modulo_n_intertible} and Convention \ref{conv:Carmichael}.

\begin{convention}[Notations]\label{conv:discrete-log}
Let $n$ be a positive integer and $x$ be an element in $\mathbb{Z}_n^{\times}$. For every element $y$ in $\mathbb{Z}_n^{\times}$ represented as a power of $x$ in $\mathbb{Z}_n^{\times}$, we denote by $\mathsf{log}_{x}(y)$ the smallest exponent $m$ in $\mathbb{Z}_n$ for which the equation $y = x^m$ holds in $\mathbb{Z}_n^{\times}$. The computational steps for the function $\mathsf{log}_{x}$ are detailed in \cite{Benaloh_revisited}.
\end{convention}

The following results identify a situation in which the computation of discrete logarithms does not require to choose among a set of representatives.

\begin{proposition}[Correctness]\label{prop:Benaloh:correctness}
Let $n$ be a positive integers and $p$ be a prime factor of $n$. We denote by $\omega$ the generator of the cyclic group $\mathbb{Z}_p^{\times}$ of order $\lambda(p) = p-1$. Let $x$ be an element in $\mathbb{Z}_n^{\times}$ such that its representative in $\mathbb{Z}_p^{\times}$ is of the form $\omega^{k}$. If the exponent $k$ is coprime with a factor $v$ of $p-1$, then the following equivalence holds: for every $m \in \mathbb{Z}_v$, the equation $y = x^m$ holds in $\mathbb{Z}_n^{\times}$ if, and only, if $\mathsf{log}_{x}(y) = m$.
\end{proposition}
\begin{proof}
To show the direct implication, let us show that if the equation $y = x^m$ holds in $\mathbb{Z}_n^{\times}$ for some $m \in \mathbb{Z}_v$, then $m$ is the only exponent in $\mathbb{Z}_n$ for which the equation $y = x^m$ holds in $\mathbb{Z}_n^{\times}$. Indeed, if $x^{m_1} = x^{m_2}$ in $\mathbb{Z}_n^{\times}$, then $\omega^{k(m_1-m_2)} = 1$ in $\mathbb{Z}_n^{\times}$ and hence in $\mathbb{Z}_p^{\times}$. Since $\omega$ is a generator of $\mathbb{Z}_p^{\times}$, it follows from Convention \ref{conv:Carmichael} that we must have $k(m_1-m_2) = 0$ in $\mathbb{Z}_{\lambda(p)}$ where $\lambda(p) = p-1$. Since $v$ is a factor of $p-1$, we also have $k(m_1-m_2) = 0$ in $\mathbb{Z}_{v}$. By assumption, the integer $k$ is represented in $\mathbb{Z}_{v}^{\times}$, so we have $m_1 = m_2$ in $\mathbb{Z}_{v}$. This concludes the proof for the direct implication. The opposite implication follows from Convention \ref{conv:discrete-log}.
\end{proof}

\begin{convention}[Notations]\label{conv:Bset_Benaloh}
Consider a positive integer $v$ and two prime numbers $p$ and $q$ with the following conditions:
\begin{itemize}
\item[1)] the integer $v$ divides $p-1$;
\item[2)] both integers $(p-1)/v$ and $q-1$ are coprime with the integer $v$.
\end{itemize}
These assumptions imply that $v$ divides the integer $\phi = (p-1)(q-1)$, and $\phi/v$ is coprime with $v$. Additionally, note that if $p$ is a factor of an integer $n$, then every element in $\mathbb{Z}_{n}^{\times}$ is represented in $\mathbb{Z}_{p}^{\times}$. Now, if we let $n = pq$, we denote by $B^{p,q}_v$ the set of integers $g \in \mathbb{Z}_{n}^{\times}$ such that the representative of $g^{\phi/v}$ in $\mathbb{Z}_{p}^{\times}$ is of the form $\omega^k$, where $\omega$ is a generator of $\mathbb{Z}_{p}^{\times}$ and $k$ is coprime with $v$. As the integer $\phi/v$ is coprime with $v$, the set $B^{p,q}_v$ encompasses every element of $\mathbb{Z}_{n}^{\times}$ that is a generator in $\mathbb{Z}_{p}^{\times}$.
\end{convention}

\begin{remark}[Correctness]\label{rem:Correctness:Bset:benaloh}
Consider a positive integer $v$ and two prime numbers $p$ and $q$ such that the integer $v$ divides $p-1$ and both integers $(p-1)/v$ and $q-1$ are coprime with the integer $v$. It follows from Proposition \ref{prop:Benaloh:correctness} and Convention \ref{conv:Bset_Benaloh} that the following equation holds for every $f_0 \in B^{p,q}_v$ and every $m \in \mathbb{Z}_v$.
\[
\mathsf{log}_{f_0^{\phi/v}}(f_0^{m\phi/v}) = m
\]
This means that $\mathsf{Alice}$ can decrypt ciphertexts encoded as exponentials $f_0^{m\phi/v}$ by using the discrete logarithm.
\end{remark}

The requirements detailed in the key generation for the Benaloh cryptosystem (see Generation \ref{gen:Benaloh}) are primarily motivated by Proposition \ref{prop:Benaloh:correctness} and Remark \ref{rem:Correctness:Bset:benaloh}. This is because the Benaloh cryptosystem necessitates the computation of a discrete logarithm for decrypting its ciphertexts.

\begin{generation}[Benaloh]\label{gen:Benaloh}
The key generation step for the Benaloh cryptosystem goes as follows. $\mathsf{Alice}$ starts with the following selection:
\begin{itemize}
\item[-] a positive integer $v$ and two prime numbers $p$ and $q$ such that the integer $v$ divides $p-1$ and both integers $(p-1)/v$ and $q-1$ are coprime with the integer $v$;
\item[-] and an element $f_0 \in B^{p,q}_v$ where we take .
\end{itemize}
In the discussions that follow, we will use the notations $\phi = (p-1)(q-1)$ and $n=pq$.
\end{generation}

\begin{remark}[Carmichael's totient number]\label{rem:Benaloh:Carmichael}
Given two prime numbers $p$ and $q$ such that $n = pq$, it is well-known that Carmichael's totient number $\lambda(n)$ is equal to the least common multiple of $p-1$ and $q-1$. Letting $\phi = (p-1)(q-1)$, it follows that for every $r \in \mathbb{Z}_n^{\times}$, the power $r^{\phi}$ is equal to $1$ in $\mathbb{Z}_n^{\times}$.
\end{remark}

\begin{application}[Benaloh as Yoneda]\label{app:generation:Benaloh}
Let us describe the Yoneda cryptosystem that captures the Benaloh cryptosystem. First, we take the associated limit sketch $T$ to be the terminal category $\mathbf{1} = \{\ast\}$. As a result, the associated reflective subcategory is given by the following identity.
\[
\mathbf{Mod}(\mathbf{1}) = [\mathbf{1},\mathbf{Set}]
\]
The Yoneda cryptosystem is defined at the obvious object. The Benaloh cryptosystem is then defined as a Yoneda cryptosystem of the form
\[
\mathcal{Y}(B^{p,q}_v,B^{p,q}_v,\mathbb{Z}_{n}^{\times},\mathbb{Z}_{v}|\mathbb{Z}_n^{\times},\mathsf{E},\mathsf{D})
\]
where $p$ and $q$ are two primes, $n$ is the product $pq$, and $v$ is a positive integer dividing $p-1$ such that $(p-1)/v$ and $q-1$ are coprime with $v$. Let us also denote $\phi=(p-1)(q-1)$. The collection of encryption algorithms is indexed by the set $\mathbb{Z}_n^{\times}$ such that for every $r \in \mathbb{Z}_n^{\times}$ we have
an encryption function as follows.
\[
\mathsf{E}_r:
\left(
\begin{array}{ccc}
B^{p,q}_v \times \mathbb{Z}_{v} & \to & \mathbb{Z}_{n}^{\times}\\
(g,m) & \mapsto & g^m \cdot r^v
\end{array}
\right)
\]
The decryption algorithm is given by a partial function (see below). We can use Remark \ref{rem:Benaloh:Carmichael} and Remark \ref{rem:Correctness:Bset:benaloh} to show that this function is well-defined at pairs $(g,c)$ where $c$ is of the form $\mathsf{E}_r(g,m)$ for some elements $r \in \mathbb{Z}_n^{\times}$ and $m \in \mathbb{Z}_{v} $.
\[
\mathsf{D}:
\left(
\begin{array}{ccc}
B^{p,q}_v \times \mathbb{Z}_{n}^{\times} & \to & \mathbb{Z}_{v}\\
(g,c) & \mapsto & \mathsf{log}_{g^{\phi/v}}(c^{\phi/v})
\end{array}
\right)
\]
With such functions, we can determine the form of the set $\mathcal{R}(f)$ for every $f \in B^{p,q}_v$, as shown below.
\begin{align*}
\mathcal{R}(f) & = \{f'~|~\forall h \in \mathbf{Set}(B^{p,q}_v,B^{p,q}_v),\,\forall m \in \mathbb{Z}_v\,:
\mathsf{log}_{h(f)^{\phi/v}}((h(f')^m \cdot r^v)^{\phi/v}) = m\}&(\textrm{Convention \ref{conv:reversors}})\\
& = \{f'~|~\forall h \in \mathbf{Set}(B^{p,q}_v,B^{p,q}_v),\,\forall m \in \mathbb{Z}_v:\, h(f')^{m\phi/v}r^{\phi} = h(f)^{m\phi/v}\}&(\textrm{Proposition \ref{prop:Benaloh:correctness}})\\
& = \{f'~|~\forall h \in \mathbf{Set}(B^{p,q}_v,B^{p,q}_v),\,\forall m \in \mathbb{Z}_v:\, h(f')^{m\phi/v} = h(f)^{m\phi/v}\}&(\textrm{Remark \ref{rem:Benaloh:Carmichael}})\\
& \supseteq \{f\} &
\end{align*}
The fact that $\mathcal{R}(f)$ is non-empty for all $f \in \mathbb{Z}_{\lambda(n)}^{\times}$ shows that the Yoneda cryptosystem is well-defined. Following Generation \ref{gen:YES}, the key generation step undertaken by $\mathsf{Alice}$ for this cryptosystem consists of:
\begin{itemize}
\item[-] an identity functor $H:\mathbf{1} \to \mathbf{1}$;
\item[-] an initializer $f_0$ in $\mathsf{lim}_{\mathbf{1}}(B^{p,q}_v \circ H^{\mathsf{op}}) = B^{p,q}_v$;
\item[-] and a private key $x$ in the singleton $\mathsf{col}_{\mathbf{1}}(Y \circ H)(\ast) = Y(\ast)(\ast) = \{\mathsf{id}_{\ast}\}$.
\end{itemize}
Here, the private key $x$ is taken to be trivial, primarily due to the fact that most of the sensitive information to be kept secret resides in the definition of the decryption algorithm.
\end{application}

The publication step for the Benaloh cryptosystem is described below.

\begin{publication}[Benaloh]\label{pub:Benaloh}
The key publication step for the Benaloh cryptosystem is as follows: $\mathsf{Alice}$ sends the element $f_0$ to $\mathsf{Bob}$.
\end{publication}

\begin{application}[Benaloh as Yoneda]\label{app:publication:Benaloh}
The key publication step for the Yoneda cryptosystem unfolds as follows. $\mathsf{Alice}$ does the following:
\begin{itemize}
\item[-] she computes the element
\[
f = \phi_{H,B^{p,q}_v}^{-1}(f_0)_{\ast}(x).
\]
This element corresponds to the image of function $\mathbf{1} \to B^{p,q}_v$ picking out the element $f_0$. As a result, the element $f$ is equal to the element $f_0$.
\item[-] Since $\mathcal{R}(f) \supseteq \{f\}$, $\mathsf{Alice}$ can send the pair $(f_0,f_0)$ as a public key to $\mathsf{Bob}$.
\end{itemize}
As outlined in Publication \ref{pub:YES}, $\mathsf{Alice}$ reveals both the key $f_0$ and the theoretical structure of the cryptosystem data, withholding only the computation details of the parameters intended to remain confidential.
\end{application}

We will now outline the encryption steps. In contrast to other cryptosystems discussed in previous sections, the associated step for the Yoneda translation involves selecting a noise parameter (see Encryption \ref{enc:YES}). This noise parameter is a crucial component in ensuring the security of the Benaloh cryptosystem.

\begin{encryption}[Benaloh]\label{enc:Benaloh}
The encryption step in the Benaloh cryptosystem goes as follows: First, $\mathsf{Bob}$ chooses
\begin{itemize}
\item[-] a random element $r \in \mathbb{Z}_n^{\times}$;
\item[-] and a message $m \in \mathbb{Z}_v$;
\end{itemize}
and then sends the ciphertext $c = f_0^m \cdot r^v \in \mathbb{Z}_n^{\times}$ to $\mathsf{Alice}$.
\end{encryption}

\begin{application}[Benaloh as Yoneda]\label{app:encryption:Benaloh}
The encryption step for the Yoneda cryptosystem unfolds as follows. First, $\mathsf{Bob}$ selects:
\begin{itemize}
\item[-] a noise parameter $r \in \mathbb{Z}_n^{\times}$;
\item[-] the identity function $B^{p,q}_v \to B^{p,q}_v$.
\item[-] and a message $m \in \mathbb{Z}_v$;
\end{itemize}
Then, $\mathsf{Bob}$ sends the following information to $\mathsf{Alice}$:
\begin{itemize}
\item[-] the element $c_1 = \mathsf{lim}_{\mathbf{1}}(h_H)(f_0) = f_0$ in $B^{p,q}_v$;
\item[-] and the element $c_2 = \mathsf{E}_{r}(h_{\ast}(f),m) = f_0^m \cdot r^v$ in $\mathbb{Z}_n^{\times}$.
\end{itemize}
While the pair $(c_1, c_2)$ constitutes an encryption for the message $m$, its first component $c_1 = f_0$ is information already known to $\mathsf{Alice}$. As noted before, this data can be utilized by $\mathsf{Bob}$ as a certificate, ensuring that $\mathsf{Alice}$ is aware of the cryptosystem being employed by $\mathsf{Bob}$.
\end{application}

Finally, we address the decryption steps. Initially, this stage exhibited certain ambiguities, resulting in identical decryptions for distinct messages. However, these issues were subsequently rectified by Fousse, Lafourcade, and Alnuaimi, as documented in \cite{Benaloh_revisited}.

\begin{decryption}[Benaloh]
The decryption step for the Benaloh cryptosystem consists in making $\mathsf{Bob}$ find the smallest integer $m' \in \mathbb{Z}_v$ for which the equation $(f_0^{\phi/v})^{m'} = c^{\phi/v}\,(\mathsf{mod}\,n)$ holds.
\end{decryption}

\begin{application}[Benaloh as Yoneda]\label{app:decryption:Benaloh}
The decryption step for the Yoneda cryptosystem unfolds as follows. $\mathsf{Alice}$ computes:
\begin{itemize}
\item[-] the element
\[
d = \phi_{H,\mathbb{Z}_{\lambda(n)}^{\times}}^{-1}(c_1)_{\ast}(x).
\]
This element corresponds to the image of function $\mathbf{1} \to \mathbb{Z}_{\lambda(n)}^{\times}$ picking out the element $c_1$. As a result, the element $d$ is equal to the element $c_1 = f_0$.;
\item[-] and the element $m' = \mathsf{D}(d,c_2) = \mathsf{log}_{f_0^{\phi/v}}(c_2^{\phi/v})$ in $\mathbb{Z}_v$;
\end{itemize}
The identity $m' = m$ follows from Decryption \ref{dec:YES}.
\end{application}

\subsection{NTRU encryption}\label{ssec:NTRU}
The NTRU cryptosystem, initially developed by Hoffstein, Pipher, and Silverman in 1998 \cite{NTRU98}, has undergone several enhancements since its inception. Subsequent research, such as works by Steinfeld and Stehlé \cite{SteinfeldStehle} and L\'{o}pez-Alt \emph{et al.} \cite{LopezAltNTRU}, has contributed to enhancing the security and efficiency of NTRU, resulting into what is now called the NTRUEncrypt system. In this section, we illustrate how this framework translates into the structure derived from the Yoneda Lemma.

\begin{convention}[Polynomials]\label{conv:polynomials_1}
Given a polynomial $p(X)$ in $R[X]$, we use the notation $R[X]/p(X)$ to represent the ring of polynomial representatives in $R[X]$ modulo $p(X)$. In simpler terms, the set $R[X]/p(X)$ consists of polynomials $f(X)$ with degrees less than or equal to $p(X)$. The operations defined for this ring involve the addition and multiplication of polynomials, followed by the computation of their remainder for the Euclidean division by the polynomial $p(X)$.
\end{convention}

\begin{convention}[Notations]\label{conv:polynomials_2}
For every ring $R$ with unit $\mathbb{1}$ (Definition \ref{def:rings}) and every positive integer $N$, we will denote the quotient ring $R[X]/(X^N-\mathbb{1})$ of $R[X]$ by the polynomial $X^N-\mathbb{1}$ as $R[X]_N$.
\end{convention}

\begin{convention}[Inverses]
Let $p$ be a prime number and $N$ be a positive integer. We denote as $I_p^N$ the set of polynomials $g$ in $\mathbb{Z}[X]$ such that the mapping of $g$ to the ring $\mathbb{Z}_p[X]_N$ has an inverse for the multiplication, namely there exists a (unique) polynomial $g'$ in $\mathbb{Z}_p[X]_N$ for which the relation $g' \cdot g = 1$ holds in $\mathbb{Z}_p[X]_N$. For every $g \in I_p^N$, we will denote as $\mathsf{inv}_p(g)$ the multiplicative inverse of $g$ in $\mathbb{Z}_p[X]_N$.
\end{convention}

\begin{definition}[Keys]\label{def:NTRU:Q}
For every positive integer $N$ and every pair $(p,q)$ of prime numbers, we denote by $Q_{p,q}^N$ the set of pairs $(g_0,g_1)$ where $g_0$ and $g_1$ are two polynomials in $\mathbb{Z}_3[X]_N$ such that
\begin{itemize}
\item[1)] the polynomial $g_0$ is in $I_p^N$ and $I_q^N$,
\item[2)] and for every pair $(r,m) \in \mathbb{Z}_3[X]_N \times \mathbb{Z}_p[X]_N$, the polynomial expression
\[
g_1(X)\cdot r(X) \cdot p + g_0(X) \cdot m(X)
\]
computed in $\mathbb{Z}[X]_N$ has all its coefficients within the interval $[-q/2,q/2)$.
\end{itemize}
\end{definition}

\begin{remark}[Correctness]\label{rem:NTRU:correctness}
We will see that the NTRU cryptosystem relies on the computation of an expression of the following form for some pair $(g_0,g_1) \in Q_{p,q}^N$.
\[
g_0 \cdot (\mathsf{inv}_q(g_0) \cdot g_1\cdot r \cdot p + m)
\]
According to Definition \ref{def:NTRU:Q}, this expression provides a polynomial whose representative in $\mathbb{Z}_q[X]_N$ is given by the polynomial $g_1rp + g_0m$ computed in $\mathbb{Z}[X]_N$.
\end{remark}

\begin{generation}[NTRUEncrypt]
The key generation step for the NTRU cryptosystem goes as follows. $\mathsf{Alice}$ starts with the following selection:
\begin{itemize}
\item[-] a positive integer $N$;
\item[-] two prime numbers $p$ and $q$;
\item[-] a pair $(g_0,g_1) \in Q_{p,q}^N$.
\end{itemize}
The previous selection will later allow $\mathsf{Alice}$ to construct a product of the form $\mathsf{inv}_q(g_0) \cdot g_1 \cdot p$, which she will use as a public key.
\end{generation}

The following generation algorithm slightly differs from those explored in the previous sections, as it relies on a Yoneda encryption scheme restricted along a singleton (see Convention \ref{conv:restricted:Yoneda}).

\begin{application}[NTRUEncrypt as Yoneda]
Let us describe the overall setup for the Yoneda cryptosystem encoding the NTRUEncrypt systen. First, we take the associated limit sketch $T$ to be the terminal category $\mathbf{1} = \{\ast\}$. As a result, the associated reflective subcategory is given by the following identity.
\[
\mathbf{Mod}(\mathbf{1}) = [\mathbf{1},\mathbf{Set}]
\]
The Yoneda cryptosystem is defined at the obvious object. The NTRUEncrypt scheme is then defined as a Yoneda encryption scheme of the form
\[
\mathcal{Y}(\mathbb{Z}_q[X]_N,\mathbb{Z}_q[X]_N,\mathbb{Z}_p[X]_N,\mathbb{Z}_q[X]_N|\mathbb{Z}_3[X]_N,\mathsf{E},\mathsf{D})
\]
where $p$ and $q$ are two primes. We will also require the cryptosystem to be restricted along the  subset $M = \{\mathsf{id}:\mathbb{Z}_q[X]_N \to \mathbb{Z}_q[X]_N\}$ of $\mathbf{Set}(\mathbb{Z}_q[X]_N,\mathbb{Z}_q[X]_N)$ containing the identity function. The collection of encryption algorithms is indexed by the set $\mathbb{Z}_3[X]_N$ such that for every polynomial $r \in \mathbb{Z}_3[X]_N$ we have an encryption function as follows.
\[
\mathsf{E}_r:
\left(
\begin{array}{ccc}
\mathbb{Z}_q[X]_N \times \mathbb{Z}_p[X]_N & \to & \mathbb{Z}_q[X]_N\\
(g,m) & \mapsto & g \cdot r + m
\end{array}
\right)
\]
The decryption algorithm is given by the following function, which is defined for a fixed pair $(g_0,g_1) \in Q_{p,q}^N$. The modulo operation by $q$ used in the decryption below is assumed to output the only representative in $\mathbb{Z}_q[X]_N$ whose coefficients are in the interval $[-q/2,q/2)$.
\[
\mathsf{D}:
\left(
\begin{array}{ccc}
\mathbb{Z}_q[X]_N \times \mathbb{Z}_q[X]_N & \to & \mathbb{Z}_p[X]_N\\
(g,c) & \mapsto & \mathsf{inv}_p(g_0) \cdot \Big( g_0 \cdot c \,(\mathsf{mod}\,q) \Big)
\end{array}
\right)
\]
With such functions, we can determine the form of the set $\mathcal{R}(f)$  for every element $f \in \mathbb{Z}_q[X]_N$, as shown below.
\begin{align*}
\mathcal{R}(f) & = \left\{f'~\left|~\def\arraystretch{1.6}\begin{array}{l}\forall r \in \mathbb{Z}_3[X]_N,\,\forall h \in M,\,\forall m \in \mathbb{Z}_p[X]_N\,:\\
\mathsf{inv}_p(g_0) \Big( g_0 \big(h(f') r + m\big) \,(\mathsf{mod}\,q) \Big) = m\end{array}\right.\right\}&(\textrm{Convention \ref{conv:reversors}})\\
& = \left\{f'~\left|~\def\arraystretch{1.6}\begin{array}{l}\forall r \in \mathbb{Z}_3[X]_N,\,\forall m \in \mathbb{Z}_p[X]_N,\,\exists k \in \mathbb{Z}[X]\,:\\
g_0 f' r + g_0m \,(\mathsf{mod}\,q) = g_0m-kp \textrm{ in }\mathbb{Z}[X]_N\end{array}\right.\right\}&(\textrm{Because }M=\{\mathsf{id}\})\\
& \supseteq\{\mathsf{inv}_q(g_0) \cdot g_1 \cdot p \textrm{ computed in }\mathbb{Z}_q[X]_N \}&(\textrm{Remark \ref{rem:NTRU:correctness}})
\end{align*}
The fact that $\mathcal{R}(f)$ is non-empty for all $f \in \mathbb{Z}_q[X]_N$ shows that the Yoneda cryptosystem is well-defined. Following Generation \ref{gen:YES}, the key generation step undertaken by $\mathsf{Alice}$ for this cryptosystem consists of:
\begin{itemize}
\item[-] an identity functor $H:\mathbf{1} \to \mathbf{1}$;
\item[-] an initializer $f_0$ in $\mathsf{lim}_{\mathbf{1}}(\mathbb{Z}_q[X]_N \circ H^{\mathsf{op}}) = \mathbb{Z}_q[X]_N$;
\item[-] and a private key $x$ in the singleton $\mathsf{col}_{\mathbf{1}}(Y \circ H)(\ast) = Y(\ast)(\ast) = \{\mathsf{id}_{\ast}\}$.
\end{itemize}
Here, the private key $x$ is taken to be trivial, primarily due to the fact that most of the sensitive information to be kept secret resides in the definition of the decryption algorithm. Also note that the initializer always satisfies the inclusion $\mathcal{R}(f_0) \supseteq \{\mathsf{inv}_q(g_0) \cdot g_1 \cdot p\}$ where the pair $(g_0,g_1)\in Q_{p,q}^N$ is fixed by the decryption algorithm above.
\end{application}

Let us now describe the publication steps. Interestingly, the NTRU cryptosystem is the first cryptosystem presented so far whose associated reversors do not seem to depend on the initializer. This feature is probably one of the main characteristics of NTRUEncrypt compared to other HE schemes.

\begin{publication}[NTRUEncrypt]
The key publication step for the NTRU cryptosystem is as follows: $\mathsf{Alice}$ sends the compute the polynomial $f' = \mathsf{inv}_q(g_0) \cdot g_1 \cdot p$ in $\mathbb{Z}_q[X]_N$ and send it to $\mathsf{Bob}$.
\end{publication}

\begin{application}[NTRUEncrypt as Yoneda]
The key publication step for the Yoneda cryptosystem unfolds as follows. $\mathsf{Alice}$ does the following:
\begin{itemize}
\item[-] she computes the element
\[
f = \phi_{H,\mathbb{Z}_q[X]_N}^{-1}(f_0)_{\ast}(x).
\]
This element corresponds to the image of function $\mathbf{1} \to \mathbb{Z}_q[X]_N$ picking out the element $f_0$. As a result, the element $f$ is equal to the element $f_0$.
\item[-] Since $\mathcal{R}(f) \supseteq \{\mathsf{inv}_q(g_0) \cdot g_1 \cdot p\}$, $\mathsf{Alice}$ can take $f'  = \mathsf{inv}_q(g_0) \cdot g_1 \cdot p$ and send the pair $(f_0,f')$ as a public key to $\mathsf{Bob}$.
\end{itemize}
Note that the element $f_0$ sent by $\mathsf{Alice}$ has no influence on the encryption-decryption protocole and should just be seen as a formal certificate, which will later be resent by $\mathsf{Bob}$ as a confirmation to $\mathsf{Alice}$.
\end{application}

We now verify that the NTRU cryptosystem and its Yoneda translation have corresponding encryption protocols. As mentioned above, the Yoneda formalism adds a formal certificate that can be used between $\mathsf{Alice}$ and a sending party to agree on a parameters to use for the communication.

\begin{encryption}[NTRUEncrypt]
The encryption step in the NTRU cryptosystem goes as follows: First, $\mathsf{Bob}$ selects
\begin{itemize}
\item[-] a random element $r \in \mathbb{Z}_3[X]_N$;
\item[-] and a message $m \in \mathbb{Z}_p[X]_N$;
\end{itemize}
and then sends the ciphertext $c = f' \cdot r + m$, computed in $\mathbb{Z}_q[X]_N$, to $\mathsf{Alice}$.
\end{encryption}

\begin{application}[NTRUEncrypt as Yoneda]
The encryption step for the Yoneda cryptosystem unfolds as follows. First, $\mathsf{Bob}$ selects:
\begin{itemize}
\item[-] a noise parameter $r \in \mathbb{Z}_3[X]_N$;
\item[-] the identity function $h:\mathbb{Z}_q[X]_N \to \mathbb{Z}_q[X]_N$.
\item[-] and a message $m \in \mathbb{Z}_p[X]_N$;
\end{itemize}
Then, $\mathsf{Bob}$ sends the following information to $\mathsf{Alice}$:
\begin{itemize}
\item[-] the element $c_1 = \mathsf{lim}_{\mathbf{1}}(h_H)(f_0) = f_0$ in $\mathbb{Z}_q[X]_N$;
\item[-] and the element $c_2 = \mathsf{E}_{r}(h_{\ast}(f'),m) = f' \cdot r + m$ in $\mathbb{Z}_q[X]_N$.
\end{itemize}
While the pair $(c_1, c_2)$ constitutes an encryption for the message $m$, its first component $c_1 = f_0$ is information already known to $\mathsf{Alice}$. As noted before, this data can be utilized by $\mathsf{Bob}$ as a certificate, ensuring that $\mathsf{Alice}$ is aware of the cryptosystem being employed by $\mathsf{Bob}$.
\end{application}

Now, we proceed to outline the decryption step for the NTRU cryptosystem. A vital prerequisite for the success of this step is to guarantee the non-emptiness of the set $Q_{p,q}^N$ introduced in Definition \ref{def:NTRU:Q}. This condition relies on the careful selection of parameters $p$, $q$, and $N$.

\begin{decryption}[NTRUEncrypt]
The decryption step for the NTRU cryptosystem requires $\mathsf{Bob}$ to:
\begin{itemize}
\item[-] compute the representative $m_0$ of $g_0 \cdot c$ in $\mathbb{Z}_q[X]_N$ that has all its coefficients in $[-q/2,q/2)$
;
\item[-] and then to compute the element $m'=\mathsf{inv}_p(g_0) \cdot m_0$ in $\mathbb{Z}_p[X]_N$.
\end{itemize}
It follows from Definition \ref{def:NTRU:Q} and the fact that $(g_0,g_1) \in Q_{p.q}^N$ that the element $m'$ is equal to the message $m$ in $\mathbb{Z}_p[X]_N$.
\end{decryption}

\begin{application}[NTRUEncrypt as Yoneda]
The decryption step for the Yoneda cryptosystem unfolds as follows. $\mathsf{Alice}$ computes:
\begin{itemize}
\item[-] the element
\[
d = \phi_{H,\mathbb{Z}_q[X]_N}^{-1}(c_1)_{\ast}(x).
\]
This element corresponds to the image of function $\mathbf{1} \to \mathbb{Z}_q[X]_N$ picking out the element $c_1$. As a result, the element $d$ is equal to the element $c_1 = f_0$.;
\item[-] and the element $m' = \mathsf{D}(d,c_2) = \mathsf{inv}_p(g_0) \Big( g_0 (f' r + m) \,(\mathsf{mod}\,q) \Big) $ in $\mathbb{Z}_p[X]_N$;
\end{itemize}
The identity $m' = m$ follows from Decryption \ref{dec:YES}.
\end{application}

\subsection{LWE-based cryptosystems}\label{ssec:LWE-based}
Learning with Error (LWE) serves as the fundamental basis for most, if not all, post-quantum cryptography and stands as the cornerstone for fully homomorphic encryption schemes. The inherent complexity of LWE is grounded in the difficulty of lattice-based problems. Oded Regev introduced the LWE problem in 2009, presenting a cryptosystem reliant on the problem's hardness on the torus \cite{Regev09}. Regev's cryptosystem has since sparked inspiration for a diverse array of similar encryption schemes \cite{BV11}. Since its inception, the LWE problem has undergone various formulations, particularly expanding to rings of polynomials, leading to the RLWE problem \cite{RLWE_def,RLWE_more,BGV12,FV12,TFHE,DGHV10,CKKS16}. The general version of LWE is articulated below, utilizing some terminology introduced in Convention \ref{conv:polynomials_1}.

\begin{convention}[Quotient rings as modules]\label{conv:ZX-module:quotient-rings}
For every ring $R$ and every polynomial $u$ in $R$, we will denote the quotient ring $R[X]/u(X)$ as $R[X]_u$. Unless specified otherwise, a ring of the form $R[X]_u$ will be seen as a $\mathbb{Z}[X]$-module whose action $\mathbb{Z}[X] \times R[X]_u \to R[X]_u$ is given by the multiplication of polynomials up to the canonical morphism $\mathbb{Z}[X] \hookrightarrow R[X] \to R[X]_u$.
\end{convention}

In this section, we use the notation introduced in Convention \ref{conv:ring-as-group:powers}.

\begin{definition}[GLWE]
Consider four positive integers $p$, $q$, $n$, $N$. For a given polynomial $u(X)$ in $\mathbb{Z}[X]$, the \emph{General LWE problem} entails finding $\mathbf{s} \in  \mathbb{Z}_q[X]_u^{(n)}$ given the following data:
\begin{itemize}
\item[1)] a polynomial $m \in \mathbb{Z}_p[X]_u$;
\item[2)] for every $i \in [N]$, a vector $\mathbf{a}_i \in  \mathbb{Z}_q[X]_u^{(n)}$;
\item[3)] for every $i \in [N]$, a vector $b_i \in  \mathbb{Z}_q[X]_u$ such that
\[
b_i = \langle \mathbf{a}_i, \mathbf{s} \rangle + e_i
\]
where $e_i$ is chosen in $\mathbb{Z}_q[X]_u$ according to some distribution $\chi$, and $\langle \mathbf{a}_i, \mathbf{s} \rangle$ denotes the scalar product in $\mathbb{Z}_q[X]_u^{(n)}$, namely the sum $\sum_{j=1}^n a_{i,j}s_{j}$ where $\mathbf{a}_i = (a_{i,1},\dots,a_{i,n})$ and $\mathbf{s} = (s_{1},\dots,s_{n})$.
\end{itemize}
For additional details on the definition of $\chi$, refer to \cite{Regev09} and other cited references such as \cite{RLWE_def,BGV12}.
\end{definition}

The Yoneda encryption scheme takes on two forms within FWE-based systems. If the cryptosystem discloses the data $\mathbf{a}_i$ used in the statement of LWE for its public key, the associated limit sketch $T$ under the Yoneda translation is non-trivial. On the other hand, if this data is not disclosed, the limit sketch is forced to be a terminal category. For instance, cryptosystems outlined in \cite{Regev09,BGV12,FV12,CKKS16} are instances of Yoneda encryption schemes with non-terminal limit sketches. In contrast, those elucidated in \cite{BV11,TFHE} are described by Yoneda encryption schemes featuring terminal limit sketches.
\smallskip

In this section, we exclusively explore schemes with non-terminal limit sketches. Readers can apply the insights developed throughout this article to understand other cryptosystems with terminal limit sketches \cite{BV11,TFHE,survey1}. Below, we offer a synthesis of the cryptosystems presented in \cite{Regev09,BGV12,FV12,CKKS16}, consolidating them into a single Yoneda encryption scheme. It is important to note that our approach aims to provide a comprehensive overview of the various features of the cryptosystems discussed in those works. However, a more specific refinement of the Yoneda framework may be necessary for a thorough recovery of the individual cryptosystems detailed therein.
\smallskip

We will now provide a broad overview of LWE-based cryptosystems and their transformation into Yoneda cryptosystems. Our focus will be on prominent schemes, including BGV, FV, and CKKS \cite{BGV12,FV12,CKKS16}, alongside the original cryptosystem introduced by Regev in \cite{Regev09}.

\begin{generation}[LWE-based]
The key generation step for the FWE-based cryptosystem goes as follows. $\mathsf{Alice}$ starts with the following selection:
\begin{itemize}
\item[-] four positive integers $p$, $q$, $n$, $N$;
\item[-] a polynomial $u(X)$ in $\mathbb{Z}[X]$;
\item[-] a vector $\mathbf{s}  = (s_{1},\dots,s_{n}) \in  \mathbb{Z}_q[X]_u^{(n)}$;
\end{itemize}
\end{generation}


\begin{convention}\label{conv:floor:ceil:notations}
For every real number $x$, we define the following integers:
\begin{itemize}
\item[-] $\lfloor x \rfloor$ as the greatest integer less than or equal to $x$ (floored value of $x$);
\item[-] $\lceil x \rceil$ as the smallest integer greater than or equal to $x$ (ceiling value of $x$);
\item[-] $\lfloor x \rceil$ as the integer closest to $x$ (rounded value of $x$).
\end{itemize}
When the rounding function $x \mapsto \lfloor x \rceil$ is applied to a polynomial, it is applied independently to each of its coefficients.
\end{convention}

\begin{application}[LWE-based as Yoneda]\label{app:LWE:setup:Yoneda}
Let us describe the overall setup for the Yoneda cryptosystem encoding the LWE-based cryptosystems described in \cite{Regev09,BGV12,FV12,CKKS16}. First, the associated limit sketch $T$ must be the limit sketch $T_{\mathbb{Z}[X]}$ defined for modules over the ring $\mathbb{Z}[X]$ (Convention \ref{conv:ring-to-sketch}). As a result, the associated reflective subcategory is given by the inclusion
\[
\mathbf{Mod}(T_{\mathbb{Z}[X]}) \hookrightarrow [T_{\mathbb{Z}[X]},\mathbf{Set}]
\]
and its left adjoint $L$. The Yoneda cryptosystems will all be defined at the object $1$ of $T_{\mathbb{Z}[X]}$ (see Convention \ref{conv:ring-to-sketch}). Specifically, the LWE-based cryptosystems that we investigate here are all expressed as Yoneda cryptosystems of the form
\[
\mathcal{Y}(\mathbb{Z}_q[X]_u^{(N)},\mathbb{Z}_q[X]_u,\mathbb{Z}_p[X]_u,\mathbb{Z}_q[X]_u|\mathbb{Z}_q[X]_u,\mathsf{E},\mathsf{D})
\]
where $u$ is some polynomial in $\mathbb{Z}[X]$ and where $p$, $q$ and $N$ are positive integers. It is straightforward to verify that $\mathbb{Z}_q[X]_u^{(N)}$ and $\mathbb{Z}_q[X]_u$ define $\mathbb{Z}[X]$-modules under the usual polynomial multiplication (see Convention \ref{conv:ZX-module:quotient-rings}). Given the variety of LWE-based cryptosystems, we assume, for convenience, that their associated Yoneda cryptosystems are restricted along a suitable set $M$ of morphisms.

The collection of encryption algorithms is indexed by the set $\mathbb{Z}_q[X]_u$ such that for every polynomial $r \in \mathbb{Z}_q[X]_u$ we have an encryption function of the following form, where $\delta_0$ and $\delta_1$ are some integers known to all parties.
\[
\mathsf{E}_r:
\left(
\begin{array}{ccc}
\mathbb{Z}_q[X]_u \times \mathbb{Z}_p[X]_u& \to & \mathbb{Z}_q[X]_u\\
(g,m) & \mapsto & \delta_0 \cdot m + \delta_1 \cdot r + g
\end{array}
\right)
\]
The decryption algorithm can vary slightly across different cryptosystems, yet it consistently takes the form shown below. Specifically, the decryption function incorporates a comparison operation $\mathsf{C}$ that gauges the relationship between its input and some given reference. The general form of $\mathsf{D}$ is as follows.
\[
\mathsf{D}:
\left(
\begin{array}{ccc}
\mathbb{Z}_q[X]_u \times \mathbb{Z}_q[X]_u & \to & \mathbb{Z}_p[X]_u\\
(g,c) & \mapsto & \mathsf{C}(c - g)
\end{array}
\right)
\]
For example:
\begin{itemize}
\item[1)] In Regev's cryptosystem \cite{Regev09}, we have:
\[
\delta_0 = 1, \quad \delta_1 = 0, \quad\textrm{and}\quad u(X) = X-1,
\]
such that $\mathbb{Z}_q[X]_u \cong \mathbb{Z}_q$. The function $\mathsf{C}:\mathbb{Z}_q \to \mathbb{Z}_p$ returns $0$ if its input is closer to $0$ than to $\lfloor p/2 \rfloor$, and $1$ otherwise.

\item[2)] In the BGV cryptosystem \cite{BGV12}, we have:
\[
\delta_0 = 1, \quad \delta_1 = 0, \quad\textrm{and}\quad u(X) \text{ is of the form } X^d - 1.
\]
The function $\mathsf{C}:\mathbb{Z}_q[X]_u \to \mathbb{Z}_p[X]_u$ is defined by the mapping $t \mapsto t \,(\mathsf{mod}\,p)$.

\item[3)] In the FV cryptosystem \cite{FV12}, we have:
\[
\delta_0 = \lfloor q/p \rfloor, \quad \delta_1 = 1, \quad\textrm{and}\quad u \text{ is a cyclotomic polynomial.}
\]
The function $\mathsf{C}:\mathbb{Z}_q[X]_u \to \mathbb{Z}_p[X]_u$ is defined by the mapping $t \mapsto \lfloor (p t)/q \rceil \,(\mathsf{mod}\,p)$, where the product $p \cdot t$ is computed in $\mathbb{Z}[X]_u$, and $(p t)/q$ is the quotient of the Euclidean division of $p t$ by $q$.

\item[4)] In the CKKS cryptosystem \cite{CKKS16}, we have:
\[
\delta_0 = 1, \quad \delta_1 = 1, \quad u = X^d + 1, \quad\textrm{and}\quad p = q.
\]
The function $\mathsf{C}:\mathbb{Z}_q[X]_u \to \mathbb{Z}_p[X]_u$ is defined to return $m$ for inputs of the form $t = m + e$, where $e$ is sampled from a bounded distribution. Although constructing $\mathsf{C}$ explicitly is infeasible, it represents the underlying goal of CKKS. For the purposes of this discussion, we will assume its existence, as doing so provides a framework for explaining how CKKS aligns with the Yoneda perspective.
\end{itemize}
With such functions, we can determine the form of the set $\mathcal{R}(f)$ (Convention \ref{conv:reversors}) for every $f \in \mathbb{Z}_q[X]_u^{(N)}$, as shown below.
\[
\mathcal{R}(f)  = \{f'~|~\forall h \in \mathcal{L}(\mathbb{Z}_q[X]_u^{(N)},\mathbb{Z}_q[X]_u),\,\forall m \in \mathbb{Z}_p[X]_u:\, \mathsf{C}(\delta_0 \cdot m + \delta_1 \cdot r + h(f') - h(f)) = m \}
\]
We can verify that the cryptosystems mentioned above satisfy the equation $\mathsf{C}(\delta_0 \cdot m + \delta_1 \cdot r) = m$ in $\mathbb{Z}_p[X]_u$ for elements $r$ that are small enough. This implies that, for suitably sampled elements $r$, the element $f$ itself belongs to $\mathcal{R}(f)$. However, for these cryptosystems to be secure, the set $\mathcal{R}(f)$ must include a wider range of elements, particularly elements of the form $f + e$, where the summand $e$ can be sampled from a random distribution on $\mathbb{Z}_q[X]_u^{(N)}$. For simplicity, we will assume that, for every $f \in \mathbb{Z}_q[X]_u$, the set $\mathcal{R}(f)$ includes such elements for an appropriately chosen set $M$. Following Generation \ref{gen:YES}, the key generation step undertaken by $\mathsf{Alice}$ for an LWE-based cryptosystem consists of
\begin{itemize}
\item[-] a non-empty discrete category $I = [n]$ containing $n$ objects (and their identities) and a functor
\[
H:\left(
\def\arraystretch{1.2}
\begin{array}{lll}
I & \to & T_{\mathbb{Z}[X]}^{\mathsf{op}}\\
k & \mapsto & 1
\end{array}
\right)
\]
picking out the object $1$ in the limit sketch $T_{\mathbb{Z}[X]}$;
\item[-] an initializer $f_0$ taken in the following set:
\[
\mathsf{lim}_{I^{\mathsf{op}}}(\mathbb{Z}_q[X]_u^{(N)} \circ H^{\mathsf{op}}) = \mathbb{Z}_q[X]_u^{(N\times n)}
\]
\item[-] and a private key $x$ taken in the following $\mathbb{Z}[X]$-module:
\[
\mathsf{col}_{I}(L \circ Y \circ H)(1) =\Big(\coprod^{n} Y(1)\Big)(1) \cong Y(n)(1) \cong \mathbb{Z}[X]^{(n)}
\]
\end{itemize}
Later, we will use the notations $f_0 = (f_{0,i,j})_{i,j}$ and $x = (x_j)_j$ to refer to the components of the elements $f_0$ and $x$ in $\mathbb{Z}_q[X]_u$ and $\mathbb{Z}[X]$, respectively. For convenience, we will also denote $f_{0,i} : = (f_{0,i,j})_{i,j}$ for every fixed index $i \in [N]$.
\end{application}

The publication steps outlined below involve a parameter $\delta_2$, which is explicitly described in Application \ref{app:LWE:publication:Yoneda}. In general, we can assume that $\delta_2$ is a carefully chosen integer that can be eliminated during the decryption process.

\begin{publication}[LWE-based]
The key publication step for LWE-based cryptosystems is defined as follows. $\mathsf{Alice}$ does the following:
\begin{itemize}
\item[-] she chooses $N$ random polynomials $\mathbf{a}_1,\dots,\mathbf{a}_N$ in $\mathbb{Z}_q[X]_u^{(n)}$;
\item[-] she chooses $N$ polynomials $e_1,\dots,e_N$ in $\mathbb{Z}_q[X]_u$ according to some distribution $\chi$;
\item[-] and she sends the following tuple to $\mathsf{Bob}$, where $\delta_2$ is an integer known to all parties:
\[
\Big((\mathbf{a}_k)_{k \in [N]},(\langle \mathbf{a}_k,\mathbf{s}\rangle + \delta_2 \cdot e_k)_{k \in [N]}\Big)
\]
\end{itemize}
From now on, we will use the notations $\mathbf{a}_k = (a_{k,j})_{j}$ to refer to the components $a_{k,j}$ of the element $\mathbf{a}_k$ in $\mathbb{Z}_q[X]_u$ for every $k \in [N]$.
\end{publication}

\begin{application}[LWE-based as Yoneda]\label{app:LWE:publication:Yoneda}
The key publication step for the Yoneda cryptosystem unfolds as follows. $\mathsf{Alice}$ does the following:
\begin{itemize}
\item[-] she computes the element
\[
f = \phi_{H,\mathbb{Z}_q[X]_u^{(N)}}^{-1}(f_0)_{1}(x)
\]
This element corresponds to the image of the element $x \in \mathbb{Z}[X]^{(n)}$ via the morphism
\[
\phi_{H,\mathbb{Z}_q[X]_u^{(N)}}^{-1}(f_0):\mathbb{Z}[X]^{(n)} \to \mathbb{Z}_q[X]_u^{(N)}
\]
that sends every element $y$ in $\mathbb{Z}[X]^{(n)}$ to the matrix product of $f_0$ with $y$. As a result, the element $f$ is equal to the matrix product $f_0x$.
\item[-] In Regev's cryptosystem, as well as in FV, CKKS, and BGV, one can sample elements $e \in \mathbb{Z}_q[X]_u^{(N)}$ from a given random distribution such that the set $\mathcal{R}(f)$ contains the element $f + \delta_2 \cdot e$, where $\delta_2$ is an integer known to all parties. As a result, $\mathsf{Alice}$ transmits a pair of the form $(f_0, f + \delta_2 \cdot e)$ to $\mathsf{Bob}$.;
\end{itemize}
For convenience, we will denote the pair sent by $\mathsf{Alice}$ to $\mathsf{Bob}$ as a tuple $(f_0,f')$.
\end{application}

An interesting characteristic of LWE-based cryptosystems is the presence of garbling operations (refer to Definition \ref{enc:YES}) with distinct sources and domains. In practical terms, this translates to the encryption algorithm concealing information through matrix products. The encryption step described below introduces two parameters $\delta_0$ and $\delta_1$, which correspond to those described in Application \ref{app:LWE:setup:Yoneda} for each LWE-based cryptosystem discussed there.

\begin{encryption}[LWE-based]
The encryption step for LWE-based cryptosystems is defined as follows. First, $\mathsf{Bob}$ selects:
\begin{itemize}
\item[-] an element $r$ in $\mathbb{Z}_q[X]_u$;
\item[-] an element $\mathbf{b} = (b_1, \dots, b_N)$ in $\mathbb{Z}_q[X]_u^{(N)}$;
\item[-] an element $m$ in $\mathbb{Z}_p[X]_u$.
\end{itemize}
Then, $\mathsf{Bob}$ sends the following information to $\mathsf{Alice}$, where $\delta_0$ and $\delta_1$ are parameters known to all parties:
\begin{itemize}
\item[-] the element $c_1 = \sum_{k=1}^N b_k \cdot \mathbf{a}_k$ in $\mathbb{Z}_q[X]_u^{(n)}$;
\item[-] the element $c_2 = \delta_0 \cdot m + \delta_1 \cdot r + \sum_{k=1}^N b_k \cdot (\langle \mathbf{a}_k, \mathbf{s} \rangle + \delta_2 \cdot e_k)$ in $\mathbb{Z}_q[X]_u$.
\end{itemize}

For example:
\begin{itemize}
\item[1)] In Regev's cryptosystem \cite{Regev09}, the integer $\delta_2$ is set to $1$, and the vector $\mathbf{b}$ belongs to $\{0,1\}^N$.
\item[2)] In the BGV cryptosystem \cite{BGV12}, the integer $\delta_2$ is set to $p$, and the vector $\mathbf{b}$ belongs to $\{0,1\}^N$.
\item[3)] In both the FV cryptosystem \cite{FV12} and the CKKS cryptosystem \cite{CKKS16}, the integer $\delta_2$ is set to $1$, and the pair $(N, n)$ is chosen as $(2, 1)$, where $a_{1,1}$ is the constant polynomial $1$. Consequently, the ciphertext components are given by:
\[
\left\{
\def\arraystretch{1.2}
\begin{array}{lll}
c_1 & = b_1 + b_2 \cdot a_{2,1} &\textrm{in }\mathbb{Z}_q[X]_u, \\
c_2 & = \delta_0 \cdot m + \underbrace{(r + b_1 \cdot (\mathbf{s} + e_1))}_{=b_1'} + b_2 \cdot (a_{2,1} \mathbf{s} + e_2) &\textrm{in }\mathbb{Z}_q[X]_u.
\end{array}
\right.
\]
In the FV cryptosystem, the three elements $b_1$, $b_2$, and $\mathbf{s}$ are sampled from a single distribution $\chi$. The FV scheme is recovered if the parameters $r$ and $e_1$ are sampled such that the distribution of the summand $b_1' = r + b_1 \cdot (\mathbf{s} + e_1)$ matches that of $\chi$.
In the CKKS cryptosystem, three distributions $\chi_1$, $\chi_2$, and $\chi_3$ are defined such that $b_1$ is sampled from $\chi_1$, $b_2$ is sampled from $\chi_2$, and $\mathbf{s}$ is sampled from $\chi_3$. The CKKS scheme is recovered if the parameters $r$ and $e_1$ are sampled such that the distribution of the summand $b_1'$ matches that of $\chi_1$.
\end{itemize}
For more detailed information on the encryption processes corresponding to the previous encryption schemes, we refer the reader to the respective references. Notably, in Regev's cryptosystem, the parameter $p$ is set to $2$, while in the CKKS scheme, the modulus $q$ exhibits a graded structure, which reflects the leveled design of the cryptographic scheme.
\end{encryption}

As previously mentioned, the encryption step in Yoneda encryption schemes employs a non-trivial garbling operation that calculates a scalar product of polynomials.

\begin{application}[LWE-based as Yoneda]
The encryption step for the Yoneda cryptosystem unfolds as follows. First, $\mathsf{Bob}$ selects:
\begin{itemize}
\item[-] a noise parameter $r \in \mathbb{Z}_q[X]_u$;
\item[-] a morphism of the following form in $\mathbf{Mod}(T_{\mathbb{Z}[X]})$.
\[
h:\mathbb{Z}_q[X]_u^{(N)} \to \mathbb{Z}_q[X]_u
\]
Such a morphism can be described as a linear map $y \mapsto b^Ty$ where $b$ is a fixed element in $\mathbb{Z}_q[X]_u^{(N)}$. We will use the notation $b = (b_k)_{k}$ to refer to the components $b_k$ of $b$ in $\mathbb{Z}_q[X]_u$;
\item[-] and a message $m \in \mathbb{Z}_p[X]_u$;
\end{itemize}
Then, $\mathsf{Bob}$ sends the following information to $\mathsf{Alice}$:
\begin{itemize}
\item[-] the element $c_1 = \mathsf{lim}_{I^{\mathsf{op}}}(h_H)(f_0) = f_0^Tb$ in $\mathbb{Z}_q[X]^{(n)}$ where $f_0^Tb$ denotes the matrix product of the transpose of $f_0$ with $b$;
\item[-] and the element $c_2 = \mathsf{E}_{r}(h_{1}(f'),m) = \delta_0 \cdot m + \delta_1 \cdot r + b^T(f_0x+ \delta_2 \cdot e)$ in $\mathbb{Z}_q[X]_N$.
\end{itemize}
\end{application}

To conclude this section, we omit the decryption steps for the previously discussed LWE-based cryptosystems and focus directly on decryption within the Yoneda encryption framework. Typically, the decryption protocols for LWE-based cryptosystems align with those derived from their corresponding Yoneda encryption schemes, employing the decryption function $\mathsf{D} : \mathbb{Z}_q[X]_u \times \mathbb{Z}_q[X]_u \to \mathbb{Z}_p[X]_u$.

\begin{application}[LWE-based as Yoneda]
The decryption step for the Yoneda cryptosystem unfolds as follows. $\mathsf{Alice}$ computes:
\begin{itemize}
\item[-] the element
\[
d = \phi_{H,\mathbb{Z}_q[X]_u}^{-1}(c_1)_{1}(x).
\]
This element corresponds to the image of the element $x \in \mathbb{Z}[X]^{(n)}$ via the morphism
\[
\phi_{H,\mathbb{Z}_q[X]_u}^{-1}(f_0):\mathbb{Z}[X]^{(n)} \to \mathbb{Z}_q[X]_u
\]
that sends every element $y$ in $\mathbb{Z}[X]^{(n)}$ to the scalar product of $c_1$ with $y$. As a result, the element $f$ is equal to the matrix product $c_1^Tx$.
\item[-] and the element $m' = \mathsf{D}(d,c_2) = \mathsf{C}(c_2 - c_1^Tx)$ in $\mathbb{Z}_p[X]_N$.
\end{itemize}
In the general case, the following identity holds:
\[
m' = \mathsf{C}\Big(\delta_0 \cdot m + \delta_1 \cdot r + \delta_2 \cdot b^T e \Big)
\]
Theoretically, the equality $m' = m$ is guaranteed by the calculations of Decryption \ref{dec:YES}.
\end{application}

\section{Applications}\label{sec:applications}

This section is divided into three key components. First, we formulate an encryption scheme employing the Yoneda Encryption Scheme. Then, we demonstrate the fulfillment of the homomorphism property by this scheme, thereby establishing its classification as a leveled FHE scheme. Lastly, we introduce a refreshing operation designed to transform this leveled FHE scheme into an unbounded FHE scheme.

\subsection{From Yoneda Lemma to leveled FHE schemes}\label{sec:FHE:from-Yoneda}
This section builds on the concepts introduced earlier, demonstrating how the Yoneda Lemma can be applied to construct a cryptosystem from foundational principles. We apply these insights to develop a novel cryptosystem, the \emph{Arithmetic Channel Encryption Scheme} (ACES). Notably, this construction does not rely on traditional bootstrapping techniques, such as those that use digit decomposition with deep circuit evaluation. Instead, it employs a simpler and more efficient noise management technique based on decomposition methods in an affine-like space.
\smallskip

The content presented from Definition \ref{def:maps-for-modulo-operations} to Example \ref{exa:maps-for-modulo-operations} is crucial for understanding the homomorphic correctness and security of ACES. This material forms the foundation for much of the subsequent technical content.

\begin{definition}[Maps for modulo operations]\label{def:maps-for-modulo-operations}
Let $p$ be a positive integer. We will denote by $\iota_p$ the inclusion $\mathbb{Z}_p \to \mathbb{Z}$ that sends every integer $x \in \{0,1,\dots,p\}$ to its corresponding element in $\mathbb{Z}$ and we will denote as $\pi_p$ the surjection $\mathbb{Z} \to \mathbb{Z}_p$ that sends every integer $x \in \mathbb{Z}$ to its value modulo $p$ in $\{0,1,\dots,p-1\}$, namely $x\,(\mathsf{mod}\,p)$.
\end{definition}

\begin{remark}[Ring homomorphism]\label{rem:maps-for-modulo-operations}
Let $p$ be a positive integer. The surjection $\pi_p: \mathbb{Z} \to \mathbb{Z}_p$ is a ring homomorphism: it preserves the units ($\pi_p(1) = 1$ and $\pi_p(0) = 0$) as well as the ring operations ($\pi_p(x+y) = \pi_p(x)+\pi_p(y)$, $\pi_p(x-y) = \pi_p(x)-\pi_p(y)$, and $\pi_p(x\cdot y) = \pi_p(x)\cdot\pi_p(y)$).
In addition, for every $x \in \{0,1,\dots,p-1\}$, we have $\pi_p(x) = x$. This means that the equation $\pi_p \circ \iota_p = \mathsf{id}_{\mathbb{Z}_p}$ holds.
\end{remark}

\begin{proposition}[Maps for modulo operations]\label{prop:maps-for-modulo-operations}
Let $p$ be a positive integer. The inclusion $\iota_p:\mathbb{Z}_p \to \mathbb{Z}$ preserves the units, namely $\iota_p(1) = 1$ and $\iota_p(0) = 0$. However, the map $\iota_p$ does not preserve ring operations. Instead, for every closed formula $F(x_1,\dots,x_n)$ using the addition, multiplication and subtraction on the elements $x_1,\dots,x_n$, we have the following implication:
\[
F(\iota_p(x_1),\dots,\iota_p(x_n)) \in \{0,1,\dots,p-1\} \quad\Rightarrow\quad \iota_p(F(x_1,\dots,x_n)) = F(\iota_p(x_1),\dots,\iota_p(x_n))
\]
\end{proposition}
\begin{proof}
First, it follows from Remark \ref{rem:maps-for-modulo-operations} that the following equations hold:
\[
\pi_p(F(\iota_p(x_1),\dots,\iota_p(x_n))) = F(\pi_p(\iota_p(x_1)),\dots,\pi_p(\iota_p(x_n))) = F(x_1,\dots,x_n) = \pi_p(\iota_p(F(x_1,\dots,x_n)))
\]
Then, it follows from the previous equation and the homomorphic properties of $\pi_p$ that we have the following identity:
\[
\pi_p(F(\iota_p(x_1),\dots,\iota_p(x_n)) - \iota_p(F(x_1,\dots,x_n))) = 0
\]
By definition of $\pi_p$, we deduce that there exists an integer $k$ for which the following equation holds.
\[
F(\iota_p(x_1),\dots,\iota_p(x_n)) - \iota_p(F(x_1,\dots,x_n)) = kp
\]
Since $F(\iota_p(x_1),\dots,\iota_p(x_n)) \in \{0,1,\dots,p-1\}$ (by assumption) and $\iota_p(F(x_1,\dots,x_n)) \in \{0,1,\dots,p-1\}$ (by definition), we deduce that the only value possible for $k$ is $0$. This shows the statement.
\end{proof}

\begin{example}[Modulo operations]\label{exa:maps-for-modulo-operations}
Let $p = 5$. The following equalities and inequalities illustrate the statement of Proposition \ref{prop:maps-for-modulo-operations}:
\begin{itemize}
\item[1)] $\iota_5(4) \cdot \iota_5(3) = 12 \neq \iota_5(4\cdot 3) = \iota_5(12) = 2$
\item[2)] $\iota_5(4) \cdot \iota_5(3) - \iota_5(10) = 2 = \iota_5(2) = \iota_5(4\cdot3-10)$
\end{itemize}
In general, for every closed formula $F(x_1,\dots,x_n)$ using additions, multiplications and subtractions, the difference $F(\iota_5(x_1),\dots,\iota_5(x_n)) - \iota_5(F(x_1,\dots,x_n))$ is a multiple of $p=5$.
\end{example}

The following definition introduces the concept of arithmetic channels. The construction of these objects will be central to generating an instance of ACES.

\begin{definition}[Arithmetic channels]\label{def:arith-channel}
An \emph{arithmetic channel} consists of a tuple $\mathsf{C} = (p,q,\omega,u)$ where:
\begin{itemize}
\item[1)] $p$, $q$ and $\omega$ are positive integers such that $p< q$;
\item[2)] $u$ is a polynomial in $\mathbb{Z}[X]$ such that $u(\omega) \equiv 0 \,(\mathsf{mod}\,q)$;
\end{itemize}
For such a structure $\mathsf{C}$, we will denote as $\intbrackets{\mathsf{C}}$ the function $\mathbb{Z}_q[X]_u \to \mathbb{Z}_q$ that sends every polynomial $v$ in $\mathbb{Z}_q[X]_u$ to the evaluation $v(\omega)$ in $\mathbb{Z}_q$. Specifically, the operation $\intbrackets{\mathsf{C}}$ is computed as follows: first, the representative $v$ is embedded in $\mathbb{N}[X]$ and then its evaluation $v(\omega)$ in $\mathbb{N}$ is computed modulo $q$.
\end{definition}

The homomorphic properties associated with ACES find their basis in the result outlined in Proposition \ref{prop:channel-homomorphism} below. This result exploits the interplay between formal polynomials and polynomial functions, a connection that we leverage through quotiented polynomial rings.

\begin{proposition}\label{prop:channel-homomorphism}
Let $\mathsf{C} = (p,q,\omega,u)$ be an arithmetic channel. For every pair $(v_1,v_2)$ of polynomials in $\mathbb{Z}_q[X]_u$, the following equations hold in the ring $\mathbb{Z}_q$.
\[
\intbrackets{\mathsf{C}}(v_1 \cdot v_2) = \intbrackets{\mathsf{C}}(v_1) \cdot \intbrackets{\mathsf{C}}(v_2)
\quad\quad\quad\quad
\intbrackets{\mathsf{C}}(v_1 + v_2) = \intbrackets{\mathsf{C}}(v_1) + \intbrackets{\mathsf{C}}(v_2)
\]
Since $\intbrackets{\mathsf{C}}(1) = 1$ and $\intbrackets{\mathsf{C}}(0) = 0$, it follows that $\intbrackets{\mathsf{C}}$ defines a ring morphism $\mathbb{Z}_q[X]_u \to \mathbb{Z}_q$.
\end{proposition}
\begin{proof}
Suppose that the symbol $\square$ denotes either an addition or a multiplication in $\mathbb{Z}[X]$. This means that the operation $\square$ commutes with the evaluation of polynomials in $\mathbb{Z}$. Let $r \in \mathbb{N}[X]$ denote the representative of $v_1 \square v_2$ in $\mathbb{Z}_q[X]_u$. This means that there exist polynomials $s$ and $t$ such that the equation
\[
r(X) = v_1(X)  \square v_2(X) + u(X)s(X) + qt(X)
\]
holds in $\mathbb{Z}[X]$. This gives us the following equations:
\begin{align*}
\intbrackets{\mathsf{C}}(v_1 \square v_2) & = \intbrackets{\mathsf{C}}(r)\\
& = r(\omega)\\
& = v_1(\omega) \square v_2(\omega) + u(\omega)s(\omega) + qt(\omega)\\
& = \intbrackets{\mathsf{C}}(v_1) \square \intbrackets{\mathsf{C}}(v_2) + u(\omega)s(\omega) + qt(\omega)
\end{align*}
Since $u(\omega)$ can be factored by $q$ (see Definition \ref{def:arith-channel}), the previous sequence of equations shows that the relation $\intbrackets{\mathsf{C}}(v_1 \square v_2) = \intbrackets{\mathsf{C}}(v_1) \square \intbrackets{\mathsf{C}}(v_2)$ holds in $\mathbb{Z}_q$. This equation proves the statement.
\end{proof}

The morphism defined in Proposition \ref{prop:channel-homomorphism} will enable us to further explore the relationship between formal polynomials and polynomial functions in Proposition \ref{prop:chi:properties} and Definition \ref{def:vanishing-ideal}. In particular, the latter draws on a well-established connection between ideals and their zero locus.

\begin{convention}[Leveled ideal]\label{def:chi}
Let $p$ be a positive integer. For every non-negative integer $k$, we will denote as $\chi(k)$ the subset $\{0,p,\dots,kp\}$ of $\mathbb{N}$ consisting of all multiples of $p$, from $0$ to $kp$.
\end{convention}

\begin{proposition}\label{prop:chi:properties}
Let $p$ be a non-negative integer. The following properties hold:
\begin{itemize}
\item[1)] if $a \in \chi(k_a)$ and $b \in \chi(k_b)$, then the addition $a+b$ in $\mathbb{Z}$ is also in $\chi(k_a+k_b)$;
\item[1)] if $a \in \{0,1,\dots,k_ap\}$ and $b \in \chi(k_b)$, then the multiplication $a \cdot b$ in $\mathbb{Z}$ is also in $\chi(k_ak_bp)$;
\end{itemize}
\end{proposition}
\begin{proof}
Directly follows from Convention \ref{def:chi}.
\end{proof}

\begin{definition}[Vanishing ideal]\label{def:vanishing-ideal}
For every arithmetic channel $\mathsf{C} = (p,q,\omega,u)$ and every non-negative integer $k$, we define the \emph{vanishing ideal} $\mathcal{I}_k(\mathsf{C})$ of $\mathsf{C}$ as the inverse image of the set $\chi(k)$ via the composition of the ring morphism $\intbrackets{\mathsf{C}}:\mathbb{Z}_q[X]_u \to \mathbb{Z}_q$ with the inclusion $\iota_q:\mathbb{Z}_q \to \mathbb{Z}$.
This means that the following equation holds:
\[
\mathcal{I}_k(\mathsf{C}) = \{e \in \mathbb{Z}_q[X]_u~|~\iota_q(\intbrackets{\mathsf{C}}(e)) \in \chi(k)\}
\]
\end{definition}

\begin{example}[Vanishing element]\label{rem:arith-channel:vanishing-element}
Let $\mathsf{C} = (p,q,\omega,u)$ be an arithmetic channel and suppose that $\omega$ is invertible in $\mathbb{Z}_q$. Let us also denote as $d$ the degree of the polynomial $u$ in $\mathbb{Z}[X]$. For every non-negative integer $k$, the set $\mathcal{I}_k(\mathsf{C})$ contains the element
\[
e(X) = \Big(\big(p\ell\omega^{-s} - \sum_{j=0,j \neq s}^{d-1} a_{j}\omega^{j-s} \big)~(\mathsf{mod}~q)\Big)X^{s} +  \sum_{j=0,j \neq s}^{d-1} a_{j} X^j
\]
for every integer $\ell \in \{0,1,\dots,k\}$, every integer $s \in [d-1]$ and every coefficient $a_{j} \in \mathbb{Z}_q$. In particular, the previous formula can be used to randomly generate elements from $\mathcal{I}_k(\mathsf{C})$.
\end{example}

\begin{proposition}[Leveled ideal]\label{prop:level-ideal:vanishing-ideal}
Let $\mathsf{C} = (p,q,\omega,u)$ be an arithmetic channel. The collection $(\mathcal{I}_k(\mathsf{C}))_k$ satisfies the following properties for every pair $(k_1,k_2)$ of non-negative integers:
\begin{itemize}
\item[1)] for every $e_1 \in \mathcal{I}_{k_1}(\mathsf{C})$ and every $e_2 \in \mathcal{I}_{k_2}(\mathsf{C})$, the following implication holds.
\[
k_1+k_2 < q/p \quad\quad \Rightarrow \quad\quad e_1 + e_2\in \mathcal{I}_{k_1+k_2}(\mathsf{C})
\]
\item[2)] for every $e_1 \in \mathbb{Z}_q[X]_u$ such that $\intbrackets{\mathsf{C}}(e_1) \in \{0,1,\dots,k_1p\}$ and every $e_2 \in \mathcal{I}_{k_2}(\mathsf{C})$, the following implication holds.
\[
k_1k_2p < q/p \quad\quad \Rightarrow \quad\quad e_1 \cdot e_2 \in \mathcal{I}_{k_1k_2p}(\mathsf{C})
\]
\end{itemize}
\end{proposition}
\begin{proof}
The proof uses Proposition \ref{prop:channel-homomorphism} and Definition \ref{def:vanishing-ideal}. Let us show item 1). If we have $e_1 \in \mathcal{I}_{k_1}(\mathsf{C})$ and $e_2 \in \mathcal{I}_{k_2}(\mathsf{C})$ such that $k_1+k_2 < q/p$, then the following addition, computed in $\mathbb{Z}$, is less than $q$.
\[
\iota_q(\intbrackets{\mathsf{C}}(e_1)) + \iota_q(\intbrackets{\mathsf{C}}(e_2))
\]
First, Proposition \ref{prop:maps-for-modulo-operations} implies that the identity $\iota_q(\intbrackets{\mathsf{C}}(e_1)) + \iota_q(\intbrackets{\mathsf{C}}(e_2)) = \iota_q((\intbrackets{\mathsf{C}}(e_1) + \intbrackets{\mathsf{C}}(e_2))$ holds. Then, Proposition \ref{prop:channel-homomorphism} implies that the following equation holds.
\[
\iota_q(\intbrackets{\mathsf{C}}(e_1)) + \iota_q(\intbrackets{\mathsf{C}}(e_2)) = \iota_q(\intbrackets{\mathsf{C}}(e_1 + e_2))
\]
Proposition \ref{prop:chi:properties} and Definition \ref{def:vanishing-ideal} then give the relation $\iota_q(\intbrackets{\mathsf{C}}(e_1 + e_2)) \in \chi(k_1+k_2)$. Similarly, let us show item 2). If we have $e_1 \in \mathbb{Z}_q[X]_u$ such that $\intbrackets{\mathsf{C}}(e_1) \in \{0,1,\dots,k_1p\}$ and $e \in \mathcal{I}_{k_2}(\mathsf{C})$ such that $k_1k_2 p< q/p$, then the following multiplication, computed in $\mathbb{Z}$, is less than $q$.
\[
\iota_q(\intbrackets{\mathsf{C}}(e_1)) \cdot \iota_q(\intbrackets{\mathsf{C}}(e_2))
\]
First, Proposition \ref{prop:maps-for-modulo-operations} implies that the identity $\iota_q(\intbrackets{\mathsf{C}}(e_1)) \cdot \iota_q(\intbrackets{\mathsf{C}}(e_2)) = \iota_q((\intbrackets{\mathsf{C}}(e_1) \cdot \intbrackets{\mathsf{C}}(e_2))$ holds. Then, Proposition \ref{prop:channel-homomorphism} implies that the following equation holds.
\[
\iota_q(\intbrackets{\mathsf{C}}(e_1)) \cdot \iota_q(\intbrackets{\mathsf{C}}(e_2)) = \iota_q(\intbrackets{\mathsf{C}}(e_1 \cdot e_2))
\]
Proposition \ref{prop:chi:properties} and Definition \ref{def:vanishing-ideal} then give the relation $\iota_q(\intbrackets{\mathsf{C}}(e_1 + e_2)) \in \chi(k_1k_2q)$.
\end{proof}

\begin{definition}[Error]\label{def:arith-channel:E:errors}
For every arithmetic channel $\mathsf{C} = (p,q,\omega,u)$, we denote as $\mathcal{E}(\mathsf{C})$ the set of functions $r:\mathbb{Z}_q \to \mathbb{Z}_q[X]_u$ such that for every integer $m$ in $\mathbb{Z}_q$, the equation $\intbrackets{\mathsf{C}}(r(m)) = m$ holds in $\mathbb{Z}_q$.
\end{definition}

\begin{example}[Error]\label{rem:arith-channel:E:errors}
Let $\mathsf{C} = (p,q,\omega,u)$ be an arithmetic channel and suppose that $\omega$ is invertible in $\mathbb{Z}_q$. Let us also denote as $d$ the degree of the polynomial $u$ in $\mathbb{Z}[X]$. The set $\mathcal{E}(\mathsf{C})$ contains the function
\[
r:m \mapsto \Big(\big(m\omega^{-s_{m}} - \sum_{j=0,j \neq s_m}^{d-1} a_{m,j}\omega^{j-s_m} \big)~(\mathsf{mod}~q)\Big)X^{s_{m}} +  \sum_{j=0,j \neq s_m}^{d-1} a_{m,j} X^j
\]
for every integer $s_{m} \in [d-1]$ and every coefficient $a_{m,j} \in \mathbb{Z}_q$. While, in theory, the previous formula should give us a way to randomly generate elements from $\mathcal{E}(\mathsf{C})$, in practice, we will only need to generate a single image $r(m)$ for some hypothetical polynomial $r$. This idea is further emphasized in the proof of Proposition \ref{prop:homomorphic-properties:E:errors} below.
\end{example}

Readers familiar with the subject may recognize the upcoming result as an affine adaptation of the Zariski correspondence between points and ideals in algebraic geometry. As demonstrated in Proposition \ref{prop:homomorphic-properties:E:errors} and later emphasized in Proposition \ref{prop:FHE:decryption}, a significant portion of our subsequent developments involves leveraging the affine nature of polynomials. This inherent affineness guarantees that the encrypted outputs of ACES are homomorphically linked to the messages that they protect.

\begin{proposition}\label{prop:homomorphic-properties:E:errors}
Let $\mathsf{C} = (p,q,\omega,u)$ be an arithmetic channel. For every pair $(r_1,r_2)$ of elements in $\mathcal{E}(\mathsf{C})$ and every pair $(m_1,m_2)$ of elements in $\mathbb{Z}_q$, the following holds:
\begin{itemize}
\item[1)] there exists $r_3 \in \mathcal{E}(\mathsf{C})$ such that $r_1(m_1) \cdot r_2(m_2) = r_3(m_1 \cdot m_2)$ in $\mathbb{Z}_q[X]_u$;
\item[1)] there exists $r_4 \in \mathcal{E}(\mathsf{C})$ such that $r_1(m_1) +r_2(m_2) = r_4(m_1 + m_2)$ in $\mathbb{Z}_q[X]_u$;
\end{itemize}
\end{proposition}
\begin{proof}
It is straightforward to construct two polynomials $Q_1(X)$ and $Q_2(X)$ in $\mathbb{Z}_q[X]_u$ for which the following equations hold.
\[
r_1(m_1) = m_1 + Q_1(X)\quad\quad\quad\quad r_2(m_2) = m_2 + Q_2(X)
\]
Since $r_1$ and $r_2$ are elements of $\mathcal{E}(\mathsf{C})$, it follows that $Q_1(\omega) = 0$ and $Q_2(\omega) = 0$. Now, if we take $r_3(m) = m + m_1Q_2(X) + m_2Q_1(X) + Q_1(X)Q_2(X)$ and $r_4(m) = m + Q_2(X) + Q_1(X)$ for every every element $m$ in $\mathbb{Z}_q$, then we can show that $r_3(m)(\omega) = m$ and $r_4(m)(\omega) = m$. This means that $r_3$ and $r_4$ belong to $\mathcal{E}(\mathsf{C})$. In addition, we can check that the following equations hold:
\[
\def\arraystretch{1.2}
\begin{array}{lll}
r_1(m_1) \cdot r_2(m_2) &= m_1m_2 + m_1Q_2(X) + m_2Q_1(X) + Q_1(X)Q_2(X) &= r_3(m_1m_2)\\
r_1(m_1) + r_2(m_2) &= m_1+m_2 + Q_2(X) + Q_1(X) &= r_3(m_1+m_2)\\
\end{array}
\]
This finishes the proof of the statement.
\end{proof}

In the remainder of this section, we extends the notations of Convention \ref{conv:ring-as-group:powers} to the set $\mathcal{I}_k(\mathsf{C})$. Specifically, for every arithmetic channel $\mathcal{C} = (p,q,\omega,u)$ and every non-negative integer $k$, we will denote as $\mathcal{I}_k(\mathsf{C})^{(N)}$ the subset of $\mathbb{Z}_q[X]_u^{(N)}$ whose elements $(f_1,\dots,f_N) \in \mathbb{Z}_q[X]_u^{(N)}$ are such that their components $f_i$ are in the set $\mathcal{I}_k(\mathsf{C})$ for every $i \in \{1,2,\dots,N\}$.

\begin{generation}[ACES]\label{gen:FHE:Yoneda}
Let us describe the overall setup for ACES. First, the associated limit sketch $T$ must be the limit sketch $T_{\mathbb{Z}[X]}$ defined for modules over the ring $\mathbb{Z}[X]$ (Convention \ref{conv:ring-to-sketch}). As a result, the associated reflective subcategory is given by the inclusion
\[
\mathbf{Mod}(T_{\mathbb{Z}[X]}) \hookrightarrow [T_{\mathbb{Z}[X]},\mathbf{Set}]
\]
and its left adjoint $L$. The cryptosystem is defined at the object $1$ of $T_{\mathbb{Z}[X]}$ (see Convention \ref{conv:ring-to-sketch}) and is expressed as a Yoneda cryptosystem of the form
\[
\mathcal{Y}(\mathbb{Z}_q[X]_u^{(N)},\mathbb{Z}_q[X]_u,\mathbb{Z}_p,\mathbb{Z}_q[X]_u|\mathcal{E}(\mathsf{C}),\mathsf{E},\mathsf{D})
\]
where $N$ is a positive integer and $(p,q,\omega,u)$ defines an arithmetic channel $\mathsf{C}$. We also require that the inequality $q \geq p^2N+1$ holds (see below). In addition, we restrict the cryptosystem along the following subset of morphisms (see Convention \ref{conv:restricted:Yoneda}).
\[
M = \left\{h \in \mathcal{L}(\mathbb{Z}_q[X]_u^{(N)},\mathbb{Z}_q[X]_u)~\left|~\begin{array}{l}\textrm{ where }h:y \mapsto b^Ty\\\textrm{ with }b = (b_j)_{j \in [N]} \in \mathbb{Z}_q[X]_u^{(N)}\end{array}\textrm{ and }\forall i \in [N]\,:\intbrackets{\mathsf{C}}(b_i) \in \{0,1,\dots,p\}\right.\right\}
\]
Limiting the garbling morphisms to $M$ is here needed to ensure that the decryption step of ACES is well-defined. Meanwhile, the inequality $q \geq p^2N+1$ ensures that ACES is secure. Further discussion on this topic is provided in Remark \ref{rem:ACES:for_more_security}. For now, we continue to describe the main components of the cryptosystem.

The collection of encryption algorithms for ACES is indexed by the set $\mathcal{E}(\mathsf{C})$ such that for every function $r \in \mathcal{E}(\mathsf{C})$ we have an encryption function of the following form.
\[
\mathsf{E}_r:
\left(
\begin{array}{ccc}
\mathbb{Z}_q[X]_u \times \mathbb{Z}_p& \to & \mathbb{Z}_q[X]_u\\
(g,m) & \mapsto & r(m) + g
\end{array}
\right)
\]
The decryption algorithm is given by the following function, where we consider the integer $\intbrackets{\mathsf{C}}(c - g) \in \{0,1,\dots,q\}$ modulo $p$.
\[
\mathsf{D}:
\left(
\begin{array}{ccc}
\mathbb{Z}_q[X]_u \times \mathbb{Z}_q[X]_u & \to & \mathbb{Z}_p\\
(g,c) & \mapsto & \pi_p(\iota_q(\intbrackets{\mathsf{C}}(c - g)))
\end{array}
\right)
\]
With such functions, we can determine the form of the set $\mathcal{R}(f)$ (Convention \ref{conv:reversors}) for every $f \in \mathbb{Z}_q[X]_u^{(N)}$, as shown below
\begin{align*}
\mathcal{R}(f) &= \{f'~|~\forall h \in M,\,\forall m \in \mathbb{Z}_p:\, \pi_p(\iota_q(\intbrackets{\mathsf{C}}(r(m) + h(f') - h(f)))) =  m  \}\\
&= \{f'~|~\forall h \in M,\,\forall m \in \mathbb{Z}_p:\, \pi_p(\iota_q(\intbrackets{\mathsf{C}}(r(m) + h(f') - h(f)))) =  \pi_p(\iota_q(m))  \}
\end{align*}
In fact, Definition \ref{def:arith-channel:E:errors} allows us to rewrite the specification of $\mathcal{R}(f)$ as follows.
\begin{align*}
\mathcal{R}(f) &= \{f'~|~\forall h \in M,\,\forall m \in \mathbb{Z}_p:\, \pi_p(\iota_q(\intbrackets{\mathsf{C}}(r(m) + h(f') - h(f)))) =  \pi_p(\iota_q(\intbrackets{\mathsf{C}}(r(m)))  \}\\
&= \{f'~|~\forall h \in M,\,\forall m \in \mathbb{Z}_p:\, \pi_p\Big(\iota_q(\intbrackets{\mathsf{C}}(r(m) + h(f') - h(f))) - \iota_q(\intbrackets{\mathsf{C}}(r(m))\Big) = 0 \}
\end{align*}
To determine whether $\mathcal{R}(f)$ contains more than just the element $f$ itself, we must identify non-trivial elements $f'$ for which the formula in the specification of $\mathcal{R}(f)$ holds. Intuitively, we aim to merge the two images of $\iota_q$ and utilize the fact that $\pi_q \circ \iota_q$ is the identity. First, in order to use the statement of Proposition \ref{prop:maps-for-modulo-operations}, let us assume that $f'$ satisfies the following relation.
\begin{equation}\label{eq:noisy-morphism:ACES-R_f}
\iota_q(\intbrackets{\mathsf{C}}(r(m) + h(f') - h(f))) - \iota_q(\intbrackets{\mathsf{C}}(r(m)) \in \{0,1,\dots,q-1\}
\end{equation}
This assumption is reasonable since $f$ itself serves as a candidate for $f'$, making the difference in (\ref{eq:noisy-morphism:ACES-R_f}) zero. Assuming that relation (\ref{eq:noisy-morphism:ACES-R_f}) holds, we can apply Proposition \ref{prop:maps-for-modulo-operations} to the images of $\iota_q$. Additionally, since $\intbrackets{\mathsf{C}}$ is a ring homomorphism, we can rewrite relation (\ref{eq:noisy-morphism:ACES-R_f}) as follows:
\begin{equation}\label{eq:noisy-morphism:ACES-R_f:simple}
\iota_q(\intbrackets{\mathsf{C}}(h(f') - h(f))) \in \{0,1,\dots,q-1\}
\end{equation}
Given that condition (\ref{eq:noisy-morphism:ACES-R_f}) holds, we can also apply Proposition \ref{prop:maps-for-modulo-operations} to the formula characterizing the set $\mathcal{R}(f)$. This leads to the following identities:
\begin{align*}
0 & = \pi_p\Big(\iota_q(\intbrackets{\mathsf{C}}(r(m) + h(f') - h(f))) - \iota_q(\intbrackets{\mathsf{C}}(r(m))\Big) &\\
& = \pi_p\Big(\iota_q(\intbrackets{\mathsf{C}}(h(f') - h(f)))\Big) &(\textrm{Proposition \ref{prop:maps-for-modulo-operations}})
\end{align*}
By Definition \ref{def:maps-for-modulo-operations}, this gives us the relation:
\begin{equation}\label{eq:noisy-morphism:ACES:f_diff_mnod_p}
\iota_q\big(\intbrackets{\mathsf{C}}\big(h(f') - h(f))\big) \equiv  0 \, (\mathsf{mod}\,p)
\end{equation}
Combining relation (\ref{eq:noisy-morphism:ACES:f_diff_mnod_p}) with relation (\ref{eq:noisy-morphism:ACES-R_f:simple}), we obtain the following condition:
\[
\iota_q\big(\intbrackets{\mathsf{C}}(h(f') - h(f))\big) \in \left\{0,p,\dots,p  \cdot \left\lfloor \frac{q-1}{p} \right\rfloor \right\}
\]
By Definition \ref{def:chi} and the fact that $h$ is a homomorphism of $\mathbb{Z}[X]$-modules, this condition is equivalent to requiring that all elements $f'$ satisfy:
\[
\iota_q\big(\intbrackets{\mathsf{C}}(h(f') - h(f))\big) \in \chi(\lfloor (q-1)/p \rfloor)
\]
By Definition \ref{def:vanishing-ideal}, we then have the following inclusions for $\mathcal{R}(f)$:
\begin{align*}
\mathcal{R}(f) & \supseteq \{f'~|~\forall h \in M:\, h(f' - f) \in \mathcal{I}_{\lfloor (q-1)/p \rfloor}(\mathsf{C}) \} & (\textrm{Definition \ref{def:vanishing-ideal}})\\
& \supseteq \{f' = (f'_1,\dots,f_N)~|~\,f'_i - f_i \in \mathcal{I}_{\scriptstyle \left\lfloor\frac{\lfloor (q-1)/p \rfloor}{pN}\right\rfloor }(\mathsf{C})\} &(\textrm{Def. of $M$ and Prop. \ref{prop:level-ideal:vanishing-ideal}})& \\
& \supseteq \{f' = (f'_1,\dots,f_N)~|~\,f'_i - f_i \in \mathcal{I}_{1}(\mathsf{C})\} &(\textrm{Since $q \geq p^2N+1$})& \\
& = f + \mathcal{I}_{1}(\mathsf{C})^{(N)}  &
\end{align*}
This shows that the set $\mathcal{R}(f)$ is non-trivial. Furthermore, the fact that $\mathcal{R}(f)$ is non-empty for all $f \in \mathbb{Z}_q[X]_u^{(N)}$ confirms that the Yoneda cryptosystem is well-defined when restricted along the subset $M$. Following Generation \ref{gen:YES}, the key generation step undertaken by $\mathsf{Alice}$ for this cryptosystem consists of:
\begin{itemize}
\item[-] a non-empty discrete category $I = [n]$ containing $n$ objects (and their identities) and a functor
\[
H:\left(
\def\arraystretch{1.2}
\begin{array}{lll}
I & \to & T_{\mathbb{Z}[X]}^{\mathsf{op}}\\
k & \mapsto & 1
\end{array}
\right)
\]
picking out the object $1$ in the limit sketch $T_{\mathbb{Z}[X]}$;
\item[-] an initializer $f_0$ taken in the following set:
\[
\mathsf{lim}_{I^{\mathsf{op}}}(\mathbb{Z}_q[X]_u^{(N)} \circ H^{\mathsf{op}}) = \mathbb{Z}_q[X]_u^{(N\times n)}
\]
\item[-] and a private key $x$ taken in the following $\mathbb{Z}[X]$-module:
\[
\mathsf{col}_{I}(L \circ Y \circ H)(1) =\Big(\coprod^{n} Y(1)\Big)(1) \cong Y(n)(1) \cong \mathbb{Z}[X]^{(n)}
\]
\end{itemize}
Later, we will use the notations $f_0 = (f_{0,i,j})_{i,j}$ and $x = (x_j)_j$ to refer to the components of the elements $f_0$ and $x$ in $\mathbb{Z}_q[X]_u$ and $\mathbb{Z}[X]$, respectively. For convenience, we will also denote $f_{0,i} : = (f_{0,i,j})_{i,j}$ for every fixed index $i \in [N]$.
\end{generation}

\begin{remark}[Correctness and security]\label{rem:ACES:for_more_security} Let us discuss the form of the elements in the set $\mathcal{R}(f)$, as defined in Generation \ref{gen:FHE:Yoneda}. Stepping back from the specific perspective of Generation \ref{gen:FHE:Yoneda}, we essentially used the definition of the set $M$ to show that the inclusion $\mathcal{R}(f) \supseteq f + \mathcal{I}_{k}(\mathsf{C})^{(N)}$ holds for every $f \in \mathbb{Z}_q[X]_u^{(N)}$, where we have the following formula:
\[
k = \left\lfloor\frac{\lfloor (q-1)/p \rfloor}{pN}\right\rfloor.
\]
This inclusion implies that if we take an element $f' = f + e$ with $e \in \mathcal{I}_{k}(\mathsf{C})^{(N)}$, then the decryption equation $\mathsf{D}(h(f),\mathsf{E}_{R}(h(f'),m)) = m$ holds for every $h \in M$. While this ensures correctness for ACES, a basic level of security is guaranteed only if $k \neq 0$. In fact, the larger $k$ is, the more secure ACES becomes.

This is where the inequality $q \geq p^2N+1$ serves its purpose, as it ensures that $k \geq 1$. However, a more general assumption can be made by requiring
\[
q \geq k_0p^2N+1
\]
to hold for some positive integer $k_0$, thereby ensuring that $k \geq k_0$. This generalized inequality can be used to further enhance security.
\end{remark}

The following publication step embeds the public key in the free-ring $\mathbb{Z}[X]$. However, in practice, this key would be more efficiently encoded by its representative modulo the integer $q$ and the polynomial $u$.

\begin{publication}[ACES]\label{pub:FHE:Yoneda}
The key publication step for ACES goes as follows. $\mathsf{Alice}$ does the following:
\begin{itemize}
\item[-] she computes the element
\[
f = \phi_{H,\mathbb{Z}_q[X]_u^{(N)}}^{-1}(f_0)_{1}(x)
\]
This element corresponds to the image of the element $x \in \mathbb{Z}[X]^{(n)}$ via the morphism
\[
\phi_{H,\mathbb{Z}_q[X]_u^{(N)}}^{-1}(f_0):\mathbb{Z}[X]^{(n)} \to \mathbb{Z}_q[X]_u^{(N)}
\]
that sends every element $y$ in $\mathbb{Z}[X]^{(n)}$ to the matrix product of $f_0$ with $y$. As a result, the element $f$ is equal to the matrix product $f_0x$.
\item[-] Since $\mathcal{R}(f) \supseteq f + \mathcal{I}_{1}(\mathsf{C})^{(N)}$, $\mathsf{Alice}$ can send a pair of the form $(f_0,f+ e)$, where $e \in \mathcal{I}_{1}(\mathsf{C})^{(N)}$, to $\mathsf{Bob}$;
\end{itemize}
For convenience, we will denote the pair sent by $\mathsf{Alice}$ to $\mathsf{Bob}$ as a tuple $(f_0,f')$.
\end{publication}

We now outline the encryption process for ACES. A distinctive feature of the encryption step described below, compared to the encryption schemes detailed thus far, is the careful specification of the set \(M\) of morphisms, which is crucial for managing the noise introduced into the public keys.

\begin{encryption}[ACES]\label{enc:FHE:Yoneda}
The encryption step for ACES is defined as follows. First, $\mathsf{Bob}$ selects:
\begin{itemize}
\item[-] a noise parameter $r \in \mathcal{E}(\mathsf{C})$;
\item[-] a morphism $h \in M$ of the following type.
\[
h:\mathbb{Z}_q[X]_u^{(N)} \to \mathbb{Z}_q[X]_u
\]
Such a morphism is described as a linear map $y \mapsto b^Ty$ where $b$ is a fixed element in $\mathbb{Z}_q[X]_u^{(N)}$ and, if we denote $b = (b_j)_{j}$, every component $b_i$ in $\mathbb{Z}_q[X]_u$ satisfies the relation $\intbrackets{\mathsf{C}}(b_i) \in \{0,1,\dots,p\}$;
\item[-] and a message $m \in \mathbb{Z}_p$;
\end{itemize}
Then, $\mathsf{Bob}$ sends the following information to $\mathsf{Alice}$:
\begin{itemize}
\item[-] the element $c_1 = \mathsf{lim}_{I^{\mathsf{op}}}(h_H)(f_0) = f_0^Tb$ in $\mathbb{Z}_q[X]^{(n)}$ where $f_0^Tb$ denotes the matrix product of the transpose of $f_0$ with $b$;
\item[-] and the element $c_2 = \mathsf{E}_{r}(h_{1}(f'),m) = r(m) + b^T(f_0x+ e)$ in $\mathbb{Z}_q[X]_N$.
\end{itemize}
\end{encryption}

Finally, we conclude the decryption step, whose correctness is ensured by the Yoneda Encryption Scheme.

\begin{decryption}[ACES]
The decryption step for ACES unfolds as follows. $\mathsf{Alice}$ computes:
\begin{itemize}
\item[-] the element
\[
d = \phi_{H,\mathbb{Z}_q[X]_u}^{-1}(c_1)_{1}(x).
\]
This element corresponds to the image of the element $x \in \mathbb{Z}[X]^{(n)}$ via the morphism
\[
\phi_{H,\mathbb{Z}_q[X]_u}^{-1}(f_0):\mathbb{Z}[X]^{(n)} \to \mathbb{Z}_q[X]_u
\]
that sends every element $y$ in $\mathbb{Z}[X]^{(n)}$ to the scalar product of $c_1$ with $y$. As a result, the element $d$ is equal to the matrix product $c_1^Tx$.
\item[-] and the element $m' = \mathsf{D}(d,c_2) = \pi_p(\iota_q(\intbrackets{\mathsf{C}}(c_2 - c_1^Tx)))$ in $\mathbb{Z}_p$.
\end{itemize}
The equality $m' = m$ is guaranteed by the calculations of Decryption \ref{dec:YES}.
\end{decryption}

\subsection{From polynomials to integers}\label{ssec:from-polynomials-to-integers}
As suggested by the form of the encryption and decryption functions used in Generation \ref{gen:FHE:Yoneda}, ACES relies on a dual lattice-based formalism, involving both integers and polynomials. While this interaction can also be formally captured by Proposition \ref{prop:maps-for-modulo-operations}, this section aims to clarify this duality by discussing canonical examples on which the reader can base their intuition.

\begin{remark}[Double modulo constructions]\label{rem:double-modulo}
The rationale behind employing a double modulo operation in most LWE-based cryptosystems can be explained as follows. Let $p$ and $q$ be two coprime integers such that $p < q$, and let $m$ be a message in $\mathbb{Z}_p$. First, note that the equation $(m \,(\mathsf{mod}\,q) ) \,(\mathsf{mod}\,p) = m$ holds. This equation can be generalized for every $g \in [q-1-m]$ as follows:
\[
\big( m+g \,(\mathsf{mod}\,q) \big) \,(\mathsf{mod}\,p) = (m+g)\, (\mathsf{mod}\,p)
\]
However, we cannot extend the previous formula to integers $g \in [q-m, +\infty]$. Specifically, for every positive integer $k$ and every integer $g \in [q-m, qk-m-1]$, we have the following equation for some positive integer $k_0 \leq k-1$.
\begin{align*}
\big( m+g \,(\mathsf{mod}\,q) \big) \,(\mathsf{mod}\,p) & = (m+g-qk_0) \,(\mathsf{mod}\,p) &\\
& = \big(m' +g\big ) \,(\mathsf{mod}\,p) & & m' = (m- qk_0)\,(\mathsf{mod}\,p)
\end{align*}
Here, we can see that passing $m+g$ through the two successive modulo operations conceals the message $m$ with a noisy component $qk_0 \,(\mathsf{mod}\,p)$. The presence of this term complicates guessing the actual value of $m$ as the range for $g$ becomes wider. This is precisely the strategy used by LWE-based cryptosystems for hiding messages.
\end{remark}

In Encryption \ref{enc:FHE:Yoneda}, we utilize two components, denoted as $r$ and $e$, to obscure the message during encryption. This construction aligns with the strategies employed in cryptographic systems such as BGV, FV, and CKKS. However, a notable distinction between these cryptosystems and ACES is that, in ACES, both noisy components $r$ and $e$ only affect the part of the ciphertext containing the message. In contrast, in BGV, FV, and CKKS, the corresponding noisy components are distributed across the two parts of the ciphertext.

\begin{remark}[Evaluation as an attack]
Let $\mathsf{C} = (p,q,\omega,u)$ be an arithmetic channel such that $p$ and $q$ are coprime. Since the messages of ACES belong to the set $\mathbb{Z}_p$, it seems tempting to attack the ACES encryption by evaluating ciphertexts through the morphism $\intbrackets{\mathsf{C}}$, followed by modding the resulting image by $p$. However, by doing so, we would also emphasize the phenomenon described in Remark \ref{rem:double-modulo}. Indeed, for every polynomial $g \in \mathbb{Z}_q[X]_u$, we can expect that the evaluation $g(\omega)$ in $\mathbb{N}$ is on average greater than $q$ (provided that the degree $d$ of $u$ is greater than $1$).

This means that, for every message $m \in \mathbb{Z}_p$ and element $r \in \mathcal{E}(\mathsf{C})$, if we sum the polynomial $r(m)$ with a polynomial $g \in \mathbb{Z}_q[X]_u$, then the noise component discussed in Remark \ref{rem:arith-channel:E:errors} will most likely occur when evaluating the expression $r(m) + g$ through the function $x \mapsto \pi_p(\iota_q(\intbrackets{\mathsf{C}}(x)))$.
\begin{align*}
\pi_p(\iota_q(\intbrackets{\mathsf{C}}\big(r(m) + g\big))) &= \big(r(m)(\omega)+g(\omega) \,(\mathsf{mod}\,q)\big)\,(\mathsf{mod}\,p) &\\
& = (m+g(\omega) - qk_0) \,(\mathsf{mod}\,p)) &\\
& = \big((m- qk_0)\,(\mathsf{mod}\,p)+g(\omega) \big) \,(\mathsf{mod}\,p) &\\
& = \big(r(m')(\omega)+g(\omega) \big) \,(\mathsf{mod}\,p) & m' = (m- qk_0)\,(\mathsf{mod}\,p)
\end{align*}
Here, the noisy component is difficult to guess because the encoding of the polynomial $r$ is only known to the sender of the message $m$. Indeed, note that the component $r$ is essential in replicating the phenomenon of Remark \ref{rem:double-modulo}. If the polynomial $r(m)$ was trivial, namely equal to $m$, then the coefficients of the polynomial $r(m) + g$ would be the same as those of $g$, except for the $0$-th coefficient, which is very likely to be less than $q$ as $p$ is small. This means that for a trivial $r(m)$, an attacker could potentially determine the value of the noisy component $qk_0 \, (\mathsf{mod} \, p)$ by evaluating $r(m') + g$ at $\omega$. As a result, the attack could be reduced to finding a representative of $g$ modulo $p$.
\end{remark}

To conclude, one of the main uses of polynomials in ACES is their ability to generate a broader range of noise, while also providing us with the opportunity to introduce more randomness by expressing integers as decompositions into polynomials. As discussed in Example \ref{rem:arith-channel:vanishing-element} and Example \ref{rem:arith-channel:E:errors}, this relationship allows us to shape the integer values in a specific manner while making it difficult to retrieve them when presented in a polynomial expression. We will employ a similar procedure in section \ref{ssec:ACES:homomorphic-properties} to further conceal the secret.

\subsection{Homomorphic properties}\label{ssec:ACES:homomorphic-properties}
This section addresses the homomorphic properties of the cryptosystem defined in section \ref{sec:FHE:from-Yoneda}. Specifically, it exhibits homomorphic properties for addition and multiplication with full traceability of the noise level.

To permit these homomorphic properties, an additional piece of information needs to be disclosed during the publication step of the cryptosystem. Importantly, this extra information is not anticipated to compromise the security of the cryptosystem.

\begin{definition}[Repartition]
Let $q$ and $n$ be positive integers such that $q$ has $n_0$ distinct prime factors. We will refer to any function of the form $[n] \to \{0,1,2,\dots,n_0\}$ as an \emph{$n$-repartition of $q$}. Let now $q_1,q_2,\dots,q_{n_0}$ denote the sequence of prime factors of $q$, ordered in increasing magnitude. We shall also denote $q_0 = 1$. For every $n$-repartition $\sigma:[n] \to \{0\} \cup [n_0]$ of $q$ and every pair $(i,j)$ in $[n]$, we use the following notation:
\[
\sigma[q]_{i,j} =
\left\{
\begin{array}{ll}
q/q_{\sigma(i)}q_{\sigma(j)} & \sigma(i) \neq \sigma(j)\\
q/q_{\sigma(i)} & \sigma(i) = \sigma(j)\\
\end{array}
\right.
\]
\end{definition}

\begin{definition}\label{def:tensor:channel}
Let $\mathsf{C} = (p,q,\omega,u)$ be an arithmetic channel, let $\sigma$ be an $n$-repartition of $q$, and let $x = (x_1,\dots,x_n)$ be an element in the $\mathbb{Z}[X]$-module $\mathbb{Z}[X]^{(n)}$. We denote as $\mathcal{H}(x|\mathsf{C},\sigma)$ the set of 3-tensors $\lambda = (\lambda_{i,j}^k)_{i,j,k \in [n]}$ in $\mathbb{Z}_q$ such that, for every pair $(i,j)$ of elements in $[n]$, if we denote
\begin{equation}\label{eq:tensor:channel}
e_{i,j} = x_i \cdot x_j - \sum_{k = 1}^n \lambda_{i,j}^k \cdot x_k
\end{equation}
then the element $\intbrackets{\mathsf{C}}(e_{i,j})$ is a multiple of  $\sigma[q]_{i,j}$ in $\mathbb{Z}_q$.
\end{definition}

As suggested by Generation \ref{gen:lambda:tensor}, the set defined in Definition \ref{def:tensor:channel} originates from the observation that $\mathbb{Z}[X]$-modules possess three degrees of freedom: the dimension of the ring of integer, which acts of the formal monomial $X$, the dimension given by polynomials in $\mathbb{Z}[X]$, and the dimension offered by the module structure\footnote{This leveled structure is implicitely suggested by our construction in Convention \ref{conv:ring-to-sketch}, where we (1) start with a ring structure $R$, (2) then send it to a commutative group structure, and (3) induce its associated modules using functors $T_R \to \mathbf{Set}$}. This distinguishes it from the cryptosystems explored in section \ref{sec:Cryptosystems-and-Yoneda}, which were typically dependent on what could be considered as two degrees of freedom, often encoded by generators and their exponents (e.g., ElGamal, RSA, Benaloh) or encoded by rings and their ideals (NTRU and LWE-based cryptosystems). In essence, the reader may perceive this freedom in dimensionality as the factor that renders ACES fully homomorphic.

\begin{generation}[Homomorphism]\label{gen:lambda:tensor}
Let $\mathsf{C} = (p,q,\omega,u)$ be an arithmetic channel, let $\sigma$ be an $n$-repartition of $q$, and let $x = (x_1,\dots,x_n)$ be an element in the $\mathbb{Z}[X]$-module $\mathbb{Z}[X]^{(n)}$. To define an element in $\mathcal{H}(x|\mathsf{C},\sigma)$, we can consider the following function of type $\mathbb{Z}_q^{(n)} \to \mathbb{Z}_q$ for every pair $(i,j)$ of elements in $[n]$:
\[
E_{i,j}(Y_1,Y_2,\dots,Y_n) = \intbrackets{\mathsf{C}}(x_i) \cdot \intbrackets{\mathsf{C}}(x_j) - \sum_{k = 1}^n Y_k \cdot \intbrackets{\mathsf{C}}(x_k)
\]
Then, because $\sigma[q]_{i,j}$ is a factor of $q$, every element $\lambda = (\lambda_{i,j}^k)_{i,j,k \in [n]}$ in $\mathcal{H}(x|\mathsf{C},\sigma)$ is characterized by the following relation:
\[
E_{i,j}(\lambda_{i,j}^1,\dots,\lambda_{i,j}^n) \equiv 0 \,(\mathsf{mod}\,\sigma[q]_{i,j})
\]
Note that the following tuples already provide solutions $(\lambda_{i,j}^k)_{i,j,k \in [n]}$ for the equation $E_{i,j}(Y_1,\dots,Y_n) = 0$.
\[
(0,\dots,0,\mathop{\intbrackets{\mathsf{C}}(x_i)}\limits_{j\textrm{-th}},0,\dots,0) \quad\quad (0,\dots,0,\mathop{\intbrackets{\mathsf{C}}(x_j)}\limits_{i\textrm{-th}},0,\dots,0)
\]
However, we want to steer away from such solutions as they reveal information that can be used to decrypt ciphertexts generated by ACES. More generally, we can sample $Y_1,\dots,Y_{n-1}$ randomly as a tuple $(\lambda_{i,j}^1,\dots,\lambda_{i,j}^{n-1})$ and solve the following equation for some random integer $\ell_{i,j} \in \mathbb{Z}_q$:
\[
Y_{n} \cdot \intbrackets{\mathsf{C}}(x_n)   = \intbrackets{\mathsf{C}}(x_i) \cdot \intbrackets{\mathsf{C}}(x_j) - \sum_{k = 1}^{n-1} \lambda_{i,j}^k \cdot \intbrackets{\mathsf{C}}(x_k) - \ell_{i,j} \cdot \sigma[q]_{i,j}
\]
The previous equation is always solvable when $\intbrackets{\mathsf{C}}(x_n)$ is invertible in $\mathbb{Z}_q$. Finally, note that for a given element $\lambda = (\lambda_{i,j}^k)_{i,j,k \in [n]}$ in $\mathcal{H}(x|\mathsf{C},\sigma)$, an attacker who discovers the $n$-repartition $\sigma$ can attempt to find the images $(\intbrackets{\mathsf{C}}(x_k))_{k \in [n]}$ by solving the following system of polynomial equations (defined for every pair $(i,j)$ in $[n]$).
\[
X_i \cdot X_j - \sum_{k = 1}^n \lambda_{i,j}^k \cdot X_k - L_{i,j} \cdot \sigma[q]_{i,j} = 0
\]
This system could potentially be solved by using Gr\"{o}bner basis reduction techniques. However, by imposing that $\lambda_{i,j}^k = \lambda_{j,i}^k$ and $L_{i,j} = L_{j,i}$, we prevent this scenario from occurring as the equation for $(i,j)$ cannot be reduced by the equation for $(j,i)$, as they are exactly the same.
\end{generation}

The set of 3-tensors defined in Generation \ref{gen:lambda:tensor} are used to define the operation presented in Definition \ref{def:boxtimes:channel}. The shape of the 3-tensors used for this operation will specifically allow us to obtain the identity stated in Proposition \ref{prop:boxtimes_lambda:product}.

\begin{definition}\label{def:boxtimes:channel}
Let $\mathsf{C} = (p,q,\omega,u)$ be an arithmetic channel, let $\sigma$ be an $n$-repartition of $q$, and let $x = (x_1,\dots,x_n)$ be an element in the $\mathbb{Z}[X]$-module $\mathbb{Z}[X]^{(n)}$. For every element $\lambda = (\lambda_{i,j}^k)$ in $\mathcal{H}(x|\mathsf{C},\sigma)$ and every pair $(v_1,v_2)$ of elements in the $\mathbb{Z}[X]$-module $\mathbb{Z}_q[X]_u^{(n)}$, we denote as $v_1 \boxtimes_{\lambda} v_2$ the element of $\mathbb{Z}_q[X]_u^{(n)}$ encoded by the following tuple.
\[
v_1 \boxtimes_{\lambda} v_2 := \Big(\sum_{i,j} \lambda_{i,j}^{k} v_{1,i}v_{2,j} \Big)_{k \in [n]}
\]
\end{definition}

To ensure that the 3-tensors defined in Generation \ref{gen:lambda:tensor} can be iteratively used through homomorphic operations (see Definition \ref{def:FHE:algebraic-operations}), we will have to restrict them to a subset defined in Definition \ref{def:zero-divisor-ideal:sigma}

\begin{definition}[Zero-divisor ideal]\label{def:zero-divisor-ideal:sigma}
Let $\mathsf{C} = (p,q,\omega,u)$ be an arithmetic channel, let $\sigma$ be an $n$-repartition of $q$, and let $x = (x_1,\dots,x_n)$ be an element in the $\mathbb{Z}[X]$-module $\mathbb{Z}[X]^{(n)}$. We will denote as $\sigma\mathbb{Z}_q[X]_u$ the $\mathbb{Z}[X]$-submodule of $\mathbb{Z}_q[X]_u^{(n)}$ consisting of tuples $v = (v_1,v_2,\dots,v_n)$ such that for every $i \in [n]$, there exists an integer $r_i$ for which $\intbrackets{\mathsf{C}}(v_i)$ is of the form $q_{\sigma(i)}r_i$ in $\mathbb{Z}_q$. Similarly, we denote as $\sigma \mathcal{H}(x|\mathsf{C},\sigma)$ the set of $3$-tensors $\lambda = (\lambda_{i,j}^k)$ with coefficients in $\mathbb{Z}_q$ such that for every triple $(i,j,k)$ in $[n]$, there exists an integer $r_{i,j}^k$ such that $\lambda_{i,j}^k = q_{\sigma(k)} r_{i,j}^k$.
\end{definition}

The following remark refines the discussion of Generation \ref{gen:lambda:tensor} using arithmetic concepts. In particular, it ensures that the operation defined in Definition \ref{def:boxtimes:channel} is stable for the homomorphic framework defined later in Definition \ref{def:FHE:algebraic-operations}.

\begin{remark}[Zero-divisor ideal]\label{rem:zero-divisor-ideal:sigma}
Let $\mathsf{C} = (p,q,\omega,u)$ be an arithmetic channel, let $\sigma$ be an $n$-repartition of $q$, and let $x = (x_1,\dots,x_n)$ be an element in the $\mathbb{Z}[X]$-module $\mathbb{Z}[X]^{(n)}$.
It directly follows from Definition \ref{def:zero-divisor-ideal:sigma} and Definition \ref{def:boxtimes:channel} that if a 3-tensor $\lambda = (\lambda_{i,j}^k)$ belongs to $\sigma \mathcal{H}(x|\mathsf{C},\sigma)$, then every $n$-vector of the form $v_1 \boxtimes_{\lambda} v_2$ belongs to $\sigma \mathbb{Z}_q[X]_u$.

Note that generating an element $\lambda$ in $\sigma \mathcal{H}(x|\mathsf{C},\sigma)$ should be as difficult as using the methods presented in Generation \ref{gen:lambda:tensor} to generate an element in $\mathcal{H}(x|\mathsf{C},\sigma)$. Specifically, this holds if we assume that the greatest common divisor of the following sequence of $n$ integers is equal to $1$.
\[
q_{\sigma(1)} \iota_q(\intbrackets{\mathsf{C}}(x_1)),\,q_{\sigma(2)} \iota_q(\intbrackets{\mathsf{C}}(x_2)),\,\dots,\, q_{\sigma(n)}\iota_q(\intbrackets{\mathsf{C}}(x_n))
\]
In this case, we can find elements $\mu_k \in \mathbb{Z}$ such that the following B\'{e}zout identity holds.
\[
q_{\sigma(1)}\iota_q(\intbrackets{\mathsf{C}}(x_1)) \cdot \mu_1 + q_{\sigma(2)} \iota_q(\intbrackets{\mathsf{C}}(x_2)) \cdot \mu_2 + \dots +  q_{\sigma(n)} \iota_q(\intbrackets{\mathsf{C}}(x_n)) \cdot \mu_n = 1
\]
Taking $\ell_{i,j}$ as in Generation \ref{gen:lambda:tensor} and multiplying the previous equation by the integer $\iota_q(\intbrackets{\mathsf{C}}(x_ix_j)) - \ell_{i,j} \sigma[q]_{i,j}$ would then generate an element $\lambda = (\lambda_{i,j}^k)$ in $\sigma \mathcal{H}(x|\mathsf{C},\sigma)$ such that the following equation holds:
\[
\lambda_{i,j}^k = q_{\sigma(k)} \mu_k \cdot \Big(\intbrackets{\mathsf{C}}(x_i) \cdot \intbrackets{\mathsf{C}}(x_j) - \ell_{i,j} \sigma[q]_{i,j}\Big)
\]
Note that we can mitigate the formal construction of each element $\lambda_{i,j}^k$ by randomly sampling $\ell_{i,j} \in \mathbb{Z}_q$. However, we must note that the formula
$\lambda_{i,i}^i \equiv q_{\sigma(i)} \cdot \mu_i \cdot \intbrackets{\mathsf{C}}(x_i)^2\,(\mathsf{mod}\,q)$ always holds.
\end{remark}

\begin{proposition}\label{prop:boxtimes_lambda:product}
Let $\mathsf{C} = (p,q,\omega,u)$ be an arithmetic channel, let $\sigma$ be an $n$-repartition of $q$, and let $x = (x_1,\dots,x_n)$ be an element in the $\mathbb{Z}[X]$-module $\mathbb{Z}[X]^{(n)}$. For every element $\lambda = (\lambda_{i,j}^k)$ in $\mathcal{H}(x|\mathsf{C},\sigma)$ and every pair $(v_1,v_2)$ of elements in $\sigma\mathbb{Z}_q[X]_u$, there exists an element $e$ in $\mathbb{Z}_q[X]_u$ such that $\intbrackets{\mathsf{C}}(e) = 0$ and the following equation holds in $\mathbb{Z}_q[X]_u$.
\[
(v_1^Tx) \cdot (v_2^T x) = (v_1 \boxtimes_{\lambda} v_2)^T x + e
\]
\end{proposition}
\begin{proof}
The proof is a straightforward calculation using the different concepts introduced in Definition \ref{def:tensor:channel} and Definition \ref{def:boxtimes:channel}. Specifically, we have:
\begin{align*}
(v_1^Tx) \cdot (v_2^T x) & = (\sum_{i} v_{1,i} x_i ) \cdot (\sum_j v_{2,j} x_j) = \sum_{i,j} v_{1,i}  v_{2,j} x_ix_j\\
& = \sum_{i,j} v_{1,i}v_{2,j} \big(\sum_{k = 1}^n \lambda_{i,j}^k \cdot x_k + e_{i,j}\big)\\
& = \sum_{k = 1}^n \Big( \sum_{i,j} \lambda_{i,j}^k v_{1,i}v_{2,j}\Big) \cdot x_k + \sum_{i,j} v_{1,i}v_{2,j} e_{i,j}\\
& = (v_1 \boxtimes_{\lambda} v_2)^T x + \sum_{i,j} v_{1,i}v_{2,j} e_{i,j}
\end{align*}
If we take $e = \sum_{i,j} v_{1,i}v_{2,j} e_{i,j}$, then there exist two tuples $(k_{1,i})_{i \in [n]}$, $(k_{2,i})_{i \in [n]}$ of integers in $\mathbb{Z}_q$ and a matrix $(k_{3,i,j})_{i,j \in [n]}$ of integers in $\mathbb{Z}_q$ such that the following equations hold.
\begin{align*}
\intbrackets{\mathsf{C}}(e) & = \intbrackets{\mathsf{C}}\big(\sum_{i,j} v_{1,i}v_{2,j} e_{i,j}\big) \\
& = \sum_{i,j} \intbrackets{\mathsf{C}}(v_{1,i})\intbrackets{\mathsf{C}}(v_{2,j})\intbrackets{\mathsf{C}}(e_{i,j}) \\
& = \sum_{i,j} q_{\sigma(i)}k_{1,i}q_{\sigma(j)}k_{2,j}\sigma[q]_{i,j}k_{3,i,j} \\
& = \sum_{i,j} qk_{1,i}k_{2,j}k_{3,i,j}  = 0
\end{align*}
This proves the proposition.
\end{proof}

Note that the message $m$ in Definition \ref{def:encryption-space:general} (below) belongs to the set $\mathbb{Z}_q$, unlike the set $\mathbb{Z}_p$ used with ACES (see Application \ref{app:LWE:setup:Yoneda} or Remark \ref{rem:aces-givesriseto-encryption-space}). As discussed prior to Proposition \ref{prop:FHE:decryption}, this distinction influences how elements of encryption spaces are refreshed and, more broadly, decrypted.

\begin{definition}[Encryption space]\label{def:encryption-space:general}
Let $\mathsf{C} = (p,q,\omega,u)$ be an arithmetic channel, let $\sigma$ be an $n$-repartition of $q$ and let $x$ be an element in $\mathbb{Z}[X]^{(n)}$.
For every element $m \in \mathbb{Z}_q$ and every non-negative integer $k$, we define the $k$-th \emph{$\mathsf{C}$-encryption space} $\mathcal{S}^{x}_{\mathsf{C},k}(m|\sigma)$ at the element $m$, relative to the repartition $\sigma$, as the set containing the pairs
\[
(c,c') \in \sigma\mathbb{Z}_q[X]_u \times \mathbb{Z}_q[X]_u
\]
for which there exist $r \in \mathcal{E}(\mathsf{C})$ and $e \in \mathcal{I}_k(C)$ such that the equation $c' = r(m) + c^Tx + e$ holds in $\mathbb{Z}_q[X]_u$.
\end{definition}

Note that the owner of the secret key $x$ can always use the formula given in Definition \ref{def:encryption-space:general} to construct ciphertexts for given messages. On the other hand, any third party without knowledge of the secret key is constrained to using the ACES formalism described below.

\begin{remark}[ACES]\label{rem:aces-givesriseto-encryption-space}
Let $\mathsf{C} = (p,q,\omega,u)$ be an arithmetic channel and let $\sigma$ be an $n$-repartition of $q$. Consider an ACES encryption scheme, as defined in section \ref{sec:FHE:from-Yoneda} (see Generation \ref{gen:FHE:Yoneda}) with a private key $x \in \mathbb{Z}_q[X]_u^{(n)}$ and such that each row of the initializer matrix
\[
f_0  = \left(
\begin{array}{ccc}
f_{0,1,1} & \dots & f_{0,1,n}\\
\vdots & \dots & \vdots\\
f_{0,N,1} & \dots & f_{0,N,n}\\
\end{array}
\right)\in \mathbb{Z}_q[X]_u^{(N\times n)}
\]
is an element of $\sigma \mathbb{Z}_q[X]_u$. For every message $m \in \mathbb{Z}_p$ and every encryption $(c_1,c_2)$ of $m$ via this encryption scheme, the pair $(c_1,c_2)$ belongs to the $\mathsf{C}$-encryption space $\mathcal{S}^{x}_{\mathsf{C},Np}(m|\sigma)$. Indeed, as explained in Encryption \ref{enc:FHE:Yoneda}, we have the equations
\[
c_1 = f_0^Tb\quad\quad\quad c_2 = r(m) + b^T(f_0x + e) = r(m) + c_1^Tx + b^Te
\]
where $e \in \mathcal{I}_1(\mathsf{C})^{(N)}$ and the vector $b = (b_1,\dots,b_N)$ in $\mathbb{Z}_q[X]_u^{(N)}$ is such that $\intbrackets{C}(b_i) \in \{0,1,\dots,p\}$ for every $i \in [N]$. First, the expression $c_1 = f_0^Tb$ tells us that $c_1$ is a linear combination of elements in the $\mathbb{Z}[X]$-module $\sigma\mathbb{Z}_q[X]_u$ and is therefore in $\sigma\mathbb{Z}_q[X]_u$. Then, since we have the equation
\[
\intbrackets{C}(b^Te) = \sum_{i=1}^N \intbrackets{C}(b_i \cdot e_i) = \sum_{i=1}^N \intbrackets{C}(b_i) \cdot \intbrackets{C}(e_i)
\]
it follows from Proposition \ref{prop:chi:properties} that the element $\intbrackets{C}(b^Te)$ is in $\chi(Np)$, and hence $b^Te$ is in $\mathcal{I}_{Np}(\mathsf{C})$. This proves our earlier statement.
\end{remark}

Proposition \ref{prop:FHE:decryption}, presented below, consists of two independent statements that establish the correctness of ACES. The second statement directly asserts the correctness of the leveled encryption scheme structure associated with ACES. The first statement, on the other hand, enables the exploration of a novel method for ciphertext refreshing, which differs from traditional bootstrapping techniques. Notably, in the first statement, the message $m$ is considered in $\mathbb{Z}_q$, whereas in the second, it is taken in $\mathbb{Z}_p$.

\begin{proposition}[Decryption]\label{prop:FHE:decryption}
Let $\mathsf{C} = (p,q,\omega,u)$ be an arithmetic channel, let $\sigma$ be an $n$-repartition of $q$, and let $x$ be an element in $\mathbb{Z}[X]^{(n)}$. Consider the following decryption function, which was introduced in Generation \ref{gen:FHE:Yoneda}.
\[
\mathsf{D}:
\left(
\begin{array}{ccc}
\mathbb{Z}_q[X]_u \times \mathbb{Z}_q[X]_u & \to & \mathbb{Z}_q\\
(g,c) & \mapsto & \pi_p(\iota_q(\intbrackets{\mathsf{C}}(c - g)))
\end{array}
\right)
\]
For every element $m \in \mathbb{Z}_q$, let $\xi_p(m)$ denote the integer quotient $\lfloor \iota_q(m)/p \rfloor$, representing the division of $m$ by $p$. Using this notation, the following implication holds:
\[
(c, c') \in \mathcal{S}^{x}_{\mathsf{C},k}(m) \quad\quad\Rightarrow \quad\quad (c, c') \in \mathcal{S}^{x}_{\mathsf{C},k+\xi_p(m)}\big(\pi_p \circ \iota_q(m)\big).
\]
Also, for every $m \in \mathbb{Z}_p$ and every non-negative integer $k < (q+1)/p - 1$, the following implication is satisfied:
\[
(c, c') \in \mathcal{S}^{x}_{\mathsf{C},k}(m) \quad\quad\Rightarrow \quad\quad \mathsf{D}(c^T x, c') = m.
\]
\end{proposition}
\begin{proof}
The statement contains two properties to show. First, let $m \in \mathbb{Z}_q$ and let $(c,c') \in \mathcal{S}^{x}_{\mathsf{C},k}(m)$. According to the notations defined in the statement, the Euclidean division of $m$ by $p$ is given by the equation $m = p \cdot \xi_p(m) + \pi_p\circ\iota_q(m)$. This equation gives us the following relations.
\begin{align*}
c' &= r(m) + c^Tx+e\\
&= r(m) - p \cdot \xi_p(m) + c^Tx+e + p \cdot \xi_p(m)
\end{align*}
If we define the function $r':\mathbb{Z}_q \to \mathbb{Z}_q[X]_u$ with the mapping rule
\[
m' \mapsto r(m' + p \cdot \xi_p(m))(X) - p \cdot \xi_p(m)
\]
and we let $e'$ denote the polynomial $e'(X) := e(X) + p \cdot \xi_p(m)$, then we have the relation
\begin{equation}\label{eq:FHE:decryption:reformulation}
c' = r'(\pi_p\circ\iota_q(m)) + c^Tx+e'
\end{equation}
The fact that $\intbrackets{\mathsf{C}}:\mathbb{Z}_q[X]_u \to \mathbb{Z}_q$ is a ring homomorphism gives us the following equation.
\[
\intbrackets{\mathsf{C}}(r'(\pi_p\circ\iota_q(m))) = \intbrackets{\mathsf{C}}(r(m)) - \intbrackets{\mathsf{C}}(p \cdot \xi_p(m)) = m - p \cdot \xi_p(m) = \pi_p\circ\iota_q(m)
\]
For the same reasons, we have the following equation.
\[
\intbrackets{\mathsf{C}}(e') = \intbrackets{\mathsf{C}}(e)+\intbrackets{\mathsf{C}}(p \cdot \xi_p(m)) = p\cdot k + p \cdot \xi_p(m) = p\cdot (k +\xi_p(m))
\]
By Definition \ref{def:arith-channel:E:errors} and Definition \ref{def:vanishing-ideal}, this shows that $r' \in \mathcal{E}(\mathsf{C})$ and $e \in \mathcal{I}_{k +\xi_p(m)}(\mathsf{C})$. Equation (\ref{eq:FHE:decryption:reformulation}) then shows that the ciphertext $(c,c')$ belongs to the $\mathsf{C}$-encryption space $\mathcal{S}^{x}_{\mathsf{C},k+\xi_p(m)}\big(\pi_p\circ\iota_q(m)\big)$.

The second property of the statement pertains to the properties of the function $\iota_q:\mathbb{Z}_q \to \mathbb{Z}$. First, let $m$ be an element in the set $\mathbb{Z}_p$ and let $(c,c') \in \mathcal{S}^{x}_{\mathsf{C},k}(m)$. Recall that we have the following identity.
\[
\mathsf{D}(c^Tx,c') = \pi_p(\iota_q(\intbrackets{\mathsf{C}}(c' - c^Tx)))
\]
We can further expand this expression as follows:
\begin{align*}
\mathsf{D}(c^Tx,c') & = \pi_p(\iota_q(\intbrackets{\mathsf{C}}(r(m) + e))) &(\textrm{Definition \ref{def:encryption-space:general}})\\
& = \pi_p(\iota_q(\intbrackets{\mathsf{C}}(r(m)) + \intbrackets{\mathsf{C}}(e))) &(\textrm{Proposition \ref{prop:channel-homomorphism}})\\
& = \pi_p(\iota_q(m + \intbrackets{\mathsf{C}}(e)))&(\textrm{Definition \ref{def:arith-channel:E:errors}})
\end{align*}
Given that $m \in \mathbb{Z}_p$ and $\iota_q(\intbrackets{\mathsf{C}}(e)) \in \{0,p,2p,\dots,kp\}$ for $k < (q+1)/p-1$, we have the following inequality.
\[
\iota_q(m) + \iota_q(\intbrackets{\mathsf{C}}(e)) < p-1 + ((q+1)/p-1)p = q
\]
It follows from Proposition \ref{prop:maps-for-modulo-operations} that the following equations hold:
\[
\mathsf{D}(c^Tx,c') = \pi_p(\iota_q(m)) + \pi_p(\iota_q(\intbrackets{\mathsf{C}}(e))) = m
\]
Note that the rightmost equation comes from the equality $\pi_p(\iota_q(\intbrackets{\mathsf{C}}(e))) = 0$, since $e$ belongs to $\mathcal{I}_k(\mathsf{C})$.
\end{proof}

The following remark shows that the inequality used in the statement of Proposition \ref{prop:FHE:decryption} can be simplified.

\begin{remark}[Coprimes]
Let $p$ and $q$ be two positive integers. When $p$ and $q$ are coprime, the condition
\[
k < (q+1)/p-1
\]
used in the statement of Proposition \ref{prop:FHE:decryption} is equivalent to the condition $k < q/p-1$. Indeed, first note that the condition $k < (q+1)/p-1$ is equivalent to the condition $pk < q+1-p$, which is also equivalent to the inequality $p-1 < q-pk$. Since $q$ is coprime with $p$, we cannot have $q-pk = p$, hence the condition $p-1 < q-pk$ is equivalent to $p < q-pk$. We can then show that the latter condition is equivalent to the condition $k < q/p-1$
\end{remark}

We now define the algebraic operations for the homomorphic framework associated with ACES. Theorem \ref{theo:Yoneda:to:leveled-FHE} will give conditions for which these operations are stable with respect to encryption spaces.

\begin{definition}[Algebraic operations]\label{def:FHE:algebraic-operations}
Let $\mathsf{C} = (p,q,\omega,u)$ be an arithmetic channel, let $\sigma$ be an $n$-repartition of $q$, and let $x$ be an element in $\mathbb{Z}[X]^{(n)}$.
Also, let us consider an element $\lambda \in \sigma \mathcal{H}(x|\mathsf{C},\sigma)$. We define two binary operations $\oplus$ and $\otimes_{\lambda}$ of the form
\[
\big(\sigma\mathbb{Z}_q[X]_u \times \mathbb{Z}_q[X]_u\big) \times \big(\sigma\mathbb{Z}_q[X]_u \times \mathbb{Z}_q[X]_u\big) \to \sigma\mathbb{Z}_q[X]_u \times \mathbb{Z}_q[X]_u
\]
such that the following equations hold for every $(c_1,c'_1) \in \sigma\mathbb{Z}_q[X]_u \times \mathbb{Z}_q[X]_u$ and $(c_2,c'_2) \in \sigma\mathbb{Z}_q[X]_u \times \mathbb{Z}_q[X]_u$.
\[
\def\arraystretch{1.2}
\begin{array}{ll}
(c_1,c'_1) \oplus (c_2,c'_2) &= (c_1+c_2, c_1'+c_2')\\
(c_1,c'_1) \otimes_{\lambda} (c_2,c'_2) & = (c_2' \cdot c_1 +c_1'\cdot c_2 - c_1 \boxtimes_{\lambda} c_2, c_1'c_2')
\end{array}
\]
Note that the expression $c_2' \cdot c_1 +c_1'\cdot c_2$ should be seen as a linear combination of the vector $c_1$ and $c_2$ such that $c_2'$ and $c_1'$ are seen as scalar coefficients in $\mathbb{Z}[X]$.
The previous two equations fully determine the mapping rules associated with the operations $\oplus$ and $\otimes_{\lambda}$.
\end{definition}

As in Proposition \ref{prop:FHE:decryption}, Theorem \ref{theo:Yoneda:to:leveled-FHE} considers messages in $\mathbb{Z}_q$, enabling encryption spaces to be well-defined for the sum and product of these messages, both computed within $\mathbb{Z}_q$. In essence, this approach leverages the definition of encryption spaces (Definition \ref{def:encryption-space:general}) over $\mathbb{Z}_q$, ensuring consistency in operations like addition and multiplication.

\begin{theorem}[Leveled FHE]\label{theo:Yoneda:to:leveled-FHE}
Let $\mathsf{C} = (p,q,\omega,u)$ be an arithmetic channel, let $\sigma$ be an $n$-repartition of $q$, and let $x$ be an element in $\mathbb{Z}[X]^{(n)}$.
Also, let us consider an element $\lambda \in \sigma \mathcal{H}(x|\mathsf{C},\sigma)$. For every pair $(m_1,m_2)$ of elements in $\mathbb{Z}_q$, every pair $(k_1,k_2)$ of non-negative integers, and every pair of ciphertexts $(c_1,c'_1) \in \mathcal{S}^{x}_{\mathsf{C},k_1}(m_1|\sigma)$ and $(c_2,c'_2) \in \mathcal{S}^{x}_{\mathsf{C},k_2}(m_2|\sigma)$, the following implications hold:
\[
\begin{array}{lll}
k_1+k_2 < q/p &\Rightarrow (c_1,c'_1) \oplus (c_2,c'_2) &\in \mathcal{S}^{x}_{\mathsf{C},k_1+k_2}(m_1+m_2|\sigma)\\
(k_1+k_2+k_1k_2)p < q/p &\Rightarrow (c_1,c'_1) \otimes_{\lambda} (c_2,c'_2) &\in \mathcal{S}^{x}_{\mathsf{C},(k_1+k_2+k_1k_2)p}(m_1m_2|\sigma)
\end{array}
\]

\end{theorem}
\begin{proof}
The proof is a straightforward application of Definition \ref{def:FHE:algebraic-operations} and Proposition \ref{prop:level-ideal:vanishing-ideal}. Let us suppose that the inequality $k_1+k_2 < q/p$ holds and we have the following encoding for the two elements $(c_1,c'_1) \in \mathcal{S}^{x}_{\mathsf{C},k_1}(m_1|\sigma)$ and $(c_2,c'_2) \in \mathcal{S}^{x}_{\mathsf{C},k_2}(m_2|\sigma)$:
\[
(c_1,c_1') = (c_1, r_1(m_1) + c_1^Tx + e_1)
\quad
(c_2,c_2') = (c_1, r_2(m_2) + c_2^Tx + e_2)
\]
Let us denote $e_3 = e_1+e_2$. Since $k_1+k_2 < q/p$, with $e_1 \in \mathcal{I}_{k_1}(\mathsf{C})$ and $e_2 \in \mathcal{I}_{k_2}(\mathsf{C})$, Proposition \ref{prop:level-ideal:vanishing-ideal} implies that $e_3 \in \mathcal{I}_{k_1+k_2}(\mathsf{C})$. We can use the previous formulas to show that there exists $r \in \mathcal{E}(\mathsf{C})$ for which the following equations hold:
\begin{align*}
(c_1,c'_1) \oplus (c_2,c'_2) & = (c_1+c_2, r_1(m_1) + r_2(m_2) + (c_1+c_2)^Tx + e_1 + e_2)&(\textrm{linearity})\\
& = (c_1+c_2, r(m_1+m_2) + (c_1+c_2)^Tx + e_3)&(\textrm{Proposition \ref{prop:homomorphic-properties:E:errors}})
\end{align*}
According to Definition \ref{def:encryption-space:general}, the previous equations show that $(c_1,c'_1) \oplus (c_2,c'_2)$ is in $\mathcal{S}^{x}_{\mathsf{C},k_1+k_2}(m_1+m_2|\sigma)$.

To show the property satisfied by the operation $\otimes_{\lambda}$, we will proceed in two steps. First, suppose that $k_1k_2p < q/p$ and consider the element $e_3 = e_2r_1(m_1)+e_1r_2(m_2) + e_1e_2$. Since we have $e_1 \in \mathcal{I}_{k_1}(\mathsf{C})$ and $e_2 \in \mathcal{I}_{k_2}(\mathsf{C})$, and the images of the polynomials $r(m_1)$ and $r(m_2)$ via the function $\intbrackets{\mathsf{C}}$ land in $\{0,1,\dots,p\}$, Proposition \ref{prop:level-ideal:vanishing-ideal} implies that we have the relations $e_2r_1(m_1) \in \mathcal{I}_{k_2p}(\mathsf{C})$, $e_1r_2(m_2) \in\mathcal{I}_{k_1p}(\mathsf{C})$ and $e_1e_2 \in \mathcal{I}_{k_1k_2p}(\mathsf{C})$. It follows from the first item of Proposition \ref{prop:level-ideal:vanishing-ideal} that the relation $e_3 \in \mathcal{I}_{k_1p+k_2p+k_1k_2p}(\mathsf{C})$ holds. Then, we can use Proposition \ref{prop:homomorphic-properties:E:errors} and the expression of the element $e_3$ to show that there exists $r \in \mathcal{E}(\mathsf{C})$ for which the following equations hold:
\begin{align*}
(r_1(m_1) + e_1)(r_2(m_2) + e_2) & = r_1(m_1)r_2(m_2) + e_3 &\\
& = r(m_1m_2) + e_3&(\textrm{Proposition \ref{prop:homomorphic-properties:E:errors}})
\end{align*}
Alternatively, Proposition \ref{prop:boxtimes_lambda:product} gives us an element $e_4 \in \mathbb{Z}_q[X]_u$ for which the following identities hold.
\begin{align*}
(r_1(m_1) + e_1)(r_2(m_2) + e_2) & = (c'_1 - c_1^Tx)(c'_2 - c_2^Tx) &\\
& =  c'_1c'_2 - \big(c'_2 \cdot c_1 + c'_1 \cdot c_2\big)^Tx + (c_1^Tx) (c_2^Tx)&\\
& =  c'_1c'_2 - \big(c'_2 \cdot c_1 + c'_1 \cdot c_2\big)^Tx + \big(c_1\boxtimes_{\lambda} c_2\big)^T x + e_4 &(\textrm{Proposition \ref{prop:boxtimes_lambda:product}}) \\
& =  c'_1c'_2 - \big(c'_2 \cdot c_1 + c'_1 \cdot c_2- c_1\boxtimes_{\lambda} c_2\big)^Tx + e_4 &(\textrm{linearity})
\end{align*}
Our two previous developments give us the following equation:
\begin{equation}\label{eq:final:Yoneda:to:leveled-FHE}
c'_1c'_2 = r(m_1m_2) +  \big(c'_2 \cdot c_1 + c'_1 \cdot c_2- c_1\boxtimes_{\lambda} c_2\big)^Tx + e_3 + e_4
\end{equation}
According to Proposition \ref{prop:boxtimes_lambda:product}, the element $e_4$ is such that $\intbrackets{\mathsf{C}}(e_4) = 0$, which means that we have the identity
$\iota_q(\intbrackets{\mathsf{C}}(e_4)) = \iota_q(\intbrackets{\mathsf{C}}(e_3+e_4))$
and hence the following relation:
\[
e_3+e_4 \in \mathcal{I}_{k_1p+k_2p+k_1k_2p}(\mathsf{C})
\]
It follows from Definition \ref{def:encryption-space:general} that equation (\ref{eq:final:Yoneda:to:leveled-FHE}) is equivalent to the statement that $(c_1,c'_1) \otimes_{\lambda} (c_2,c'_2)$ is in $\mathcal{S}^{x}_{\mathsf{C},(k_1+k_2+k_1k_2)p}(m_1m_2|\sigma)$.
\end{proof}

In summary of this section, the assertions made in Theorem \ref{theo:Yoneda:to:leveled-FHE} and Proposition \ref{prop:FHE:decryption} affirm that the encryption scheme introduced in section \ref{sec:FHE:from-Yoneda} qualifies as a leveled fully homomorphic encryption scheme. This holds true under the condition that $\mathsf{Alice}$ publishes an element of the set $\sigma \mathcal{H}(x|\mathsf{C},\sigma)$, where $x$ represents the private key of the encryption scheme.

\subsection{Refreshable ciphertexts}\label{ssec:refreshable-ciphertexts}
This section introduces the concept of \emph{refreshable ciphertexts}, which will be instrumental in section \ref{ssec:proper-FHE} to demonstrate that ACES enables a fully homomorphic encryption scheme.
Throughout this section and section \ref{ssec:proper-FHE}, we adopt the following convention to streamline notation:

\begin{convention}[Notation]
Let $n$ be a positive integer. For every function $f: A \to B$ and every $n$-tuple $a = (a_1, a_2, \dots, a_n)$ of elements in $A$, denote by $f\tuplebrk{a}$ the corresponding $n$-tuple $(f(a_1), f(a_2), \dots, f(a_n))$ of elements in $B$.
\end{convention}

The following convention assigns two levels of meaning to the term \emph{pseudociphertext}. The first level refers to a generic terminology, primarily used in relation to an ACES secret key (see Definition \ref{def:refreshable-pseudociphertext}). The second specifically refers to a pseudociphertext associated with an ACES ciphertext (defined in Definition \ref{def:encryption-space:general}). Most of our results will only require the first level of meaning until we reach Theorem \ref{theo:test-refreshable-ciphertexts}.

\begin{convention}[Pseudociphertexts]\label{conv:pseudociphertexts}
For any positive integer $q$, we refer to a pair $(v,v')$, where $ v \in \mathbb{Z}_q^{(n)} $ and $ v' \in \mathbb{Z}_q $, as a \emph{$q$-pseudociphertext}. Let $\mathsf{C} = (p,q,\omega,u)$ be an arithmetic channel, $\sigma$ be an $n$-repartition of $q$, and $x$ be an element in $\mathbb{Z}[X]^{(n)}$. For any element $m$ in $\mathbb{Z}_q$ and any non-negative integer $k$, we define the \emph{underlying $q$-pseudociphertext} of a given ciphertext $(c,c') \in \mathcal{S}_{\mathsf{C},k}^x(m|\sigma)$ as the following $q$-pseudociphertext:
\[
\Big(\intbrackets{\mathsf{C}}\tuplebrk{-c}, \intbrackets{\mathsf{C}}(c')\Big).
\]
Here, the $n$-tuple $\intbrackets{\mathsf{C}}\tuplebrk{-c}$ is the tuple whose $i$-th coefficient is given by the following element in $\mathbb{Z}_q$.
\[
\intbrackets{\mathsf{C}}(-c_i) = -\intbrackets{\mathsf{C}}(c_i) = \pi_q(q - \iota_q(\intbrackets{\mathsf{C}}(c_i)))
\]
\end{convention}

The following definition is justified by a lack of homomorphism in the function $\iota_q:\mathbb{Z} \to \mathbb{Z}_q$, as suggested in Proposition \ref{prop:maps-for-modulo-operations} and illustrated in Example \ref{exa:maps-for-modulo-operations}. Specifically, the function $\iota_q:\mathbb{Z} \to \mathbb{Z}_q$ is only homomorphic up to a multiple of $q$ and the following definition targets cases in which this homomorphic property holds up to some multiple of $p\cdot q$ for some positive integer $p$.

\begin{definition}[Refreshable pseudociphertexts]\label{def:refreshable-pseudociphertext}
Let $\mathsf{C} = (p,q,\omega,u)$ be an arithmetic channel, let $n$ be a positive integer and let $x$ be an element in $\mathbb{Z}[X]^{(n)}$. For every non-negative integer $k$, we define the \emph{$k$-th space of $p$-refreshable $q$-pseudociphertexts} relative to $x$ as:
\[
\mathcal{F}_{p,q}^k(x) = \{(v,v') \in \mathbb{Z}_q^{(n)} \times \mathbb{Z}_q ~|~\iota_q(v') + \iota_q\tuplebrk{v}^T \big(\iota_q\circ\intbrackets{C}\tuplebrk{x}\big) = \iota_q\big(v'+v^T\intbrackets{C}\tuplebrk{x}\big) + kpq\}
\]
Every $q$-pseudociphertext in the set $\mathcal{F}_{p,q}^k(x)$ is said to be \emph{$p$-refreshable}.
\end{definition}

Proposition \ref{prop:refreshable-pseudociphertexts} (stated below) gives a context in which the function $\iota_q:\mathbb{Z} \to \mathbb{Z}_q$ can be seen as a ring homomorphism (up to composition) relative to the algebraic operations used in the decryption algorithm of ACES (this will be furthered discussed in Proposition \ref{prop:refreshable-ciphertexts}).

\begin{proposition}[Refreshable pseudociphertexts]\label{prop:refreshable-pseudociphertexts}
Let $\mathsf{C} = (p,q,\omega,u)$ be an arithmetic channel, let $n$ be a positive integer and let $x$ be an element in $\mathbb{Z}[X]^{(n)}$. For every non-negative integer $k$ and every element $(v,v') \in \mathcal{F}_{p,q}^k(x)$, the following equation holds in $\mathbb{Z}_p$.
\[
\pi_p \circ \iota_q(v') + \Big(\pi_p \circ \iota_q\tuplebrk{v}\Big)^T \Big(\pi_p \circ \iota_q \circ\intbrackets{C}\tuplebrk{x}\Big) = \pi_p \circ \iota_q\big(v'+v^T\intbrackets{C}\tuplebrk{x}\big)
\]
\end{proposition}
\begin{proof}
Since we have $(v,v') \in \mathcal{F}_{p,q}^k(x)$, the following relation holds.
\[
\iota_q(v') + \iota_q\tuplebrk{v}^T \big(\iota_q\circ\intbrackets{C}\tuplebrk{x}\big) = \iota_q(v'+v^T\intbrackets{C}\tuplebrk{x}) + kpq
\]
By applying the ring homomorphism $\pi_p:\mathbb{Z} \to \mathbb{Z}_p$ on the previous equation, we obtain the following relation.
\[
\pi_p(\iota_q(v')) + \pi_p\tuplebrk{\iota_q\tuplebrk{v}}^T \pi_p\tuplebrk{\iota_q\circ\intbrackets{C}\tuplebrk{x}} = \pi_p\big(\iota_q\big(v'+v^T\intbrackets{C}\tuplebrk{x}\big)\big) + 0
\]
It then follows from composing the function $\pi_p$ with $\iota_q$ that the equation of the statement holds.
\end{proof}

\begin{definition}[Refreshable ciphertexts]\label{def:refreshable-ciphertexts}
Let $\mathsf{C} = (p,q,\omega,u)$ be an arithmetic channel, $\sigma$ be an $n$-repartition of $q$, and $x$ be an element in $\mathbb{Z}[X]^{(n)}$. For every element $m$ in $\mathbb{Z}_q$ and every non-negative integer $k$, we will say that a ciphertext $(c,c') \in \mathcal{S}_{\mathsf{C},k}^x(m|\sigma)$ is \emph{$p$-refreshable} if so is its underlying $q$-pseudociphertext.
\end{definition}

The following proposition constitutes the main application of Definition \ref{def:refreshable-ciphertexts}. Specifically, it suggests that refreshable ciphertexts can be decrypted with reduced knowledge of the secret key. Importantly, this information is unlikely to be inferred from the published data of ACES due to the absence of homomorphic properties in $\iota_q$. In fact, this feature is one of the main components ensuring the security of ACES.

\begin{proposition}[Refreshable ciphertexts]\label{prop:refreshable-ciphertexts}
Let $\mathsf{C} = (p,q,\omega,u)$ be an arithmetic channel, $\sigma$ be an $n$-repartition of $q$, and $x$ be an element in $\mathbb{Z}[X]^{(n)}$. For every element $m$ in $\mathbb{Z}_p$, every non-negative integer $k < (q+1)/p-1$, and every $p$-refreshable ciphertext $(c,c') \in \mathcal{S}_{\mathsf{C},k}^x(m|\sigma)$, the following equation holds in $\mathbb{Z}_p$.
\[
\pi_p \circ \iota_q(\intbrackets{\mathsf{C}}(c')) + \Big(\pi_p \circ \iota_q\circ\intbrackets{\mathsf{C}}\tuplebrk{-c}\Big)^T \Big(\pi_p \circ \iota_q\circ\intbrackets{C}\tuplebrk{x}\Big) = m
\]
\end{proposition}
\begin{proof}
By Definition \ref{def:refreshable-ciphertexts}, the $q$-pseudociphertext $(\intbrackets{\mathsf{C}}\tuplebrk{-c}, \intbrackets{\mathsf{C}}(c'))$ is $p$-refreshable. It follows from Proposition \ref{prop:refreshable-pseudociphertexts} that the top equation displayed below holds. The other equations underneath it are deduced directly from the homomorphic properties of $\intbrackets{\mathsf{C}}$ and the definition of the decryption algorithm for ACES (see Generation \ref{gen:FHE:Yoneda}).
\begin{align*}
\pi_p \circ \iota_q(\intbrackets{\mathsf{C}}(c')) + \Big(\pi_p \circ \iota_q\circ\intbrackets{\mathsf{C}}\tuplebrk{-c}\Big)^T \Big(\pi_p \circ \iota_q \circ \intbrackets{C}\tuplebrk{x}\Big) & = \pi_p \circ \iota_q\Big(\intbrackets{\mathsf{C}}(c')+\intbrackets{\mathsf{C}}\tuplebrk{-c}^T\intbrackets{C}\tuplebrk{x}\Big)\\
& = \pi_p \circ \iota_q(\intbrackets{\mathsf{C}}(c'-c^Tx))\\
& = \mathsf{D}(c^Tx,c')
\end{align*}
Since $m \in \mathbb{Z}_p$ and we assumed that the inequality $k < (q+1)/p-1$ holds, Proposition  \ref{prop:FHE:decryption} implies that the equation $\mathsf{D}(c^Tx,c') = m$ holds. This shows the statement.
\end{proof}

In the remainder of this section, we demonstrate how it is possible to determine whether a given ciphertext is refreshable without knowledge of the secret key. To achieve this, we introduce \emph{locators} (Definition \ref{def:locators}) and \emph{directors} (Definition \ref{def:directors}), which define an affine space in which ciphertext components can be expressed. This decomposition may resemble Gentry's squashing technique \cite{Gen09}, which also aims to refresh ciphertexts through decomposition. However, a key distinction is that our method targets the ciphertext itself, whereas squashing decomposes secret-key information, exposing it to greater security risks.

\begin{definition}[Locators]\label{def:locators}
Let $\mathsf{C} = (p,q,\omega,u)$ be an arithmetic channel, let $n$ be a positive integer and let $x$ be an element in $\mathbb{Z}[X]^{(n)}$. For every non-negative integer $k$, we define the \emph{$k$-th space of $p$-locators} relative to $x$ as:
\[
\mathcal{L}_{p,q}^k(x) = \left\{\ell \in \mathbb{Z}_q^{(n)} ~\left|~ \sum_{i=1}^n \iota_q\circ \intbrackets{C}(x_i) - \left\lfloor \frac{1}{q} \cdot \iota_q\tuplebrk{\ell}^T \big(\iota_q\circ\intbrackets{C}\tuplebrk{x}\big) \right\rfloor = kp\right.\right\}
\]
Every $n$-tuple in the set $\mathcal{L}_{p,q}^k(x)$ is said to be a $p$-locator for $x$.
\end{definition}

While locators describe structured sets of elements that characterize possible values of a ciphertext component, directors represent shifts within this space. A locator combined with directors can form another locator, allowing structured decomposition.

\begin{definition}[Directors]\label{def:directors}
Let $\mathsf{C} = (p,q,\omega,u)$ be an arithmetic channel, let $n$ be a positive integer and let $x$ be an element in $\mathbb{Z}[X]^{(n)}$. For every non-negative integer $k$, we define the \emph{$k$-th space of $p$-directors} relative to $x$ as:
\[
\mathcal{D}_{p,q}^k(x) = \left\{\ell \in \mathbb{Z}_q^{(n)} ~\left|~ \left\lfloor \frac{1}{q} \cdot \iota_q\tuplebrk{\ell}^T \big(\iota_q\circ\intbrackets{C}\tuplebrk{x}\big) \right\rfloor = kp\right.\right\}
\]
Every $n$-tuple in the set $\mathcal{D}_{p,q}^k(x)$ is said to be a $p$-director for $x$.
\end{definition}

The following definition introduces the concept of a \emph{margin}, which quantifies how much flexibility there is in combining locators and directors within the affine structure.

\begin{definition}[Margin]\label{def:margin}
Let $n$ and $q$ be a positive integers and let $x$ be an element in $\mathbb{Z}[X]^{(n)}$. For every element $\ell \in \mathbb{Z}_q^{(n)}$, we define the \emph{margin} of $\ell$ relative to $x$ as the real number:
\[
\mathsf{marg}_x(\ell) = \frac{1}{q} \cdot \iota_q\tuplebrk{\ell}^T \big(\iota_q\circ\intbrackets{C}\tuplebrk{x}\big) - \left\lfloor \frac{1}{q} \cdot \iota_q\tuplebrk{\ell}^T \big(\iota_q\circ\intbrackets{C}\tuplebrk{x}\big) \right\rfloor
\]
Note that $\mathsf{marg}_x(\ell)$ belongs to the interval $[0,1)$ of non-negative real number less than $1$.
\end{definition}

The following proposition gives conditions in which it is possible to combine a locator with a collection of directors to obtain a new locator.

\begin{proposition}[Affine structure]\label{prop:affine-structure}
Let $\mathsf{C} = (p,q,\omega,u)$ be an arithmetic channel, let $n$ be a positive integer, and let $x$ be an element of $\mathbb{Z}[X]^{(n)}$. For every non-negative integer $r$, let $\square_i$ denote either addition or subtraction, and let $\widetilde{\square}_i$ denote its corresponding inverse operation. Define the following expressions:
\[
F(x_0,\dots,x_{r}) = x_0 {\,\square_1\,} x_1 {\,\square_2\,} \dots {\,\square_r\,} x_r,
\quad \textrm{and} \quad
\widetilde{F}(x_0,\dots,x_{r}) = x_0 {\,\widetilde{\square}_1\,} x_1 {\,\widetilde{\square}_2\,} \dots {\,\widetilde{\square}_r\,} x_r.
\]
For every pair $(k_0,k_1)$ of non-negative integers, every $p$-locator $\ell \in \mathcal{L}_{p,q}^{k_0}(x)$, and every collection of $p$-directors $\delta_1 \in \mathcal{D}_{p,q}^{k_1}(x), \dots, \delta_r \in \mathcal{D}_{p,q}^{k_r}(x)$, suppose that the expression $F(\iota_q\tuplebrk{\ell},\iota_q\tuplebrk{\delta_1},\dots,\iota_q\tuplebrk{\delta_r})$, computed in $\mathbb{Z}^{(n)}$, belongs to $\mathbb{Z}_q^{(n)}$ and that there exists a non-negative integer $k'$ such that
\[
F(\mathsf{marg}_x(\ell),\mathsf{marg}_x(\delta_1),\dots,\mathsf{marg}_x(\delta_r)) \in [pk',pk'+1)
\]
Then, the following hold:
\[
\def\arraystretch{2}
\left\{
\begin{array}{l}
F(\mathsf{marg}_x(\ell),\mathsf{marg}_x(\delta_1),\dots,\mathsf{marg}_x(\delta_r)) - pk' = \mathsf{marg}_x(F(\ell,\delta_1,\dots,\delta_r)),\\
F(\ell,\delta_1,\dots,\delta_r) \in \mathcal{L}_{p,q}^{\widetilde{F}(k_0,k_1,\dots,k_r)-k'}(x).
\end{array}
\right.
\]
\end{proposition}
\begin{proof}
We assume that the conditions of the proposition are satisfied. By Proposition \ref{prop:maps-for-modulo-operations}, if the expression $F(\iota_q\tuplebrk{\ell},\iota_q\tuplebrk{\delta_1},\dots,\iota_q\tuplebrk{\delta_r})$ belongs to $\mathbb{Z}_q^{(n)}$, then we have the identity:
\begin{equation}\label{eq:iota:affine}
F(\iota_q\tuplebrk{\ell},\iota_q\tuplebrk{\delta_1},\dots,\iota_q\tuplebrk{\delta_r}) = \iota_q\tuplebrk{F(\ell,\delta_1,\dots,\delta_r)}.
\end{equation}
Note that in equation (\ref{eq:iota:affine}), the term $F(\ell,\delta_1,\dots,\delta_r)$ is implicitly computed in $\mathbb{Z}_q^{(n)}$, since the domain of $\iota_q$ is $\mathbb{Z}_q$. We now show that the following relation holds:
\[
F(\ell,\delta_1,\dots,\delta_r) \in \mathcal{L}_{p,q}^{\widetilde{F}(k_0,k_1,\dots,k_r)-k'}(x).
\]
From equation (\ref{eq:iota:affine}) and Definition \ref{def:margin}, we obtain:
\begin{align*}
&F(\mathsf{marg}_x(\ell), \mathsf{marg}_x(\delta_1),\dots,\mathsf{marg}_x(\delta_r))   = \frac{1}{q} \cdot \big(\iota_q\tuplebrk{F(\ell,\delta_1,\dots,\delta_r)}\big)^T \big(\iota_q\circ\intbrackets{C}\tuplebrk{x}\big) - \dots\\
&  \quad\quad F\left(\left\lfloor \frac{1}{q} \cdot \iota_q\tuplebrk{\ell}^T \big(\iota_q\circ\intbrackets{C}\tuplebrk{x}\big) \right\rfloor, \left\lfloor \frac{1}{q} \cdot \iota_q\tuplebrk{\delta_1}^T \big(\iota_q\circ\intbrackets{C}\tuplebrk{x}\big) \right\rfloor, \dots, \left\lfloor \frac{1}{q} \cdot \iota_q\tuplebrk{\delta_r}^T \big(\iota_q\circ\intbrackets{C}\tuplebrk{x}\big) \right\rfloor \right).
\end{align*}
Using the definition of the margin of $F(\ell,\delta_1,\dots,\delta_r)$ along with the previous equation, we derive the identity:
\begin{align*}
&F(\mathsf{marg}_x(\ell), \mathsf{marg}_x(\delta_1),\dots,\mathsf{marg}_x(\delta_r)) - \mathsf{marg}_x(F(\ell,\delta_1,\dots,\delta_r))  = \dots \\
&  \quad\quad F\left(\left\lfloor \frac{1}{q} \cdot \iota_q\tuplebrk{\ell}^T \big(\iota_q\circ\intbrackets{C}\tuplebrk{x}\big) \right\rfloor, \left\lfloor \frac{1}{q} \cdot \iota_q\tuplebrk{\delta_1}^T \big(\iota_q\circ\intbrackets{C}\tuplebrk{x}\big) \right\rfloor, \dots, \left\lfloor \frac{1}{q} \cdot \iota_q\tuplebrk{\delta_r}^T \big(\iota_q\circ\intbrackets{C}\tuplebrk{x}\big) \right\rfloor \right) - \dots \\
& \quad\quad\quad\quad\quad\quad \left\lfloor \frac{1}{q} \cdot \iota_q\tuplebrk{F(\ell,\delta_1,\dots,\delta_r)}^T \big(\iota_q\circ\intbrackets{C}\tuplebrk{x}\big) \right\rfloor.
\end{align*}
On the left-hand side, we have a real number in the open interval $(pk'-1,pk'+1)$, while on the right-hand side, we have an integer. Since the only integer in this interval is $pk'$, we deduce that:
\[
F(\mathsf{marg}_x(\ell), \mathsf{marg}_x(\delta_1),\dots,\mathsf{marg}_x(\delta_r)) = \mathsf{marg}_x(F(\ell,\delta_1,\dots,\delta_r)) + pk'.
\]
Thus, the following decomposition holds:
\begin{align*}
& \left\lfloor \frac{1}{q} \cdot \iota_q\tuplebrk{F(\ell,\delta_1,\dots,\delta_r)}^T \big(\iota_q\circ\intbrackets{C}\tuplebrk{x}\big) \right\rfloor = \dots \\
& F\left(\left\lfloor \frac{1}{q} \cdot \iota_q\tuplebrk{\ell}^T \big(\iota_q\circ\intbrackets{C}\tuplebrk{x}\big) \right\rfloor, \left\lfloor \frac{1}{q} \cdot \iota_q\tuplebrk{\delta_1}^T \big(\iota_q\circ\intbrackets{C}\tuplebrk{x}\big) \right\rfloor, \dots, \left\lfloor \frac{1}{q} \cdot \iota_q\tuplebrk{\delta_r}^T \big(\iota_q\circ\intbrackets{C}\tuplebrk{x}\big) \right\rfloor \right) - pk'.
\end{align*}
Since $\ell \in \mathcal{L}_{p,q}^{k_0}(x)$ and $\delta_1 \in \mathcal{D}_{p,q}^{k_1}(x), \dots, \delta_r \in \mathcal{D}_{p,q}^{k_r}(x)$, it follows from the previous equation and Definitions \ref{def:locators} and \ref{def:directors} that:
\[
\sum_{i=1}^n \iota_q\circ \intbrackets{C}(x_i) - \left\lfloor \frac{1}{q} \cdot \iota_q\tuplebrk{F(\ell,\delta_1,\dots,\delta_r)}^T \big(\iota_q\circ\intbrackets{C}\tuplebrk{x}\big) \right\rfloor = (\widetilde{F}(k_0,k_1,\dots,k_r) -k') p.
\]
Thus, we conclude that $F(\ell,\delta_1,\dots,\delta_r) \in \mathcal{L}_{p,q}^{\widetilde{F}(k_0,k_1,\dots,k_r)-k'}(x)$, which completes the proof.
\end{proof}

\begin{remark}[Affine structure]\label{rem:affine_structure}
In Theorem \ref{theo:test-refreshable-ciphertexts} (see below), we will see that a ciphertext $(c, c')$, relative to some arithmetic channel $\mathsf{C}$, is refreshable if the vector $\intbrackets{C}\tuplebrk{c}$ is a locator. This characterization motivates the need for a test to determine whether a given $c$ is a locator. Our proposed approach is to attempt to decompose $c$ as
\[
c = \ell + \sum_{i = 1}^k \pm \delta_i,
\]
where $\ell$ is a locator and the terms $\delta_i$ are directors. If such a decomposition is found, then by Proposition \ref{prop:affine-structure}, we can conclude that $c$ is indeed a locator, provided that the conditions of the proposition are satisfied.

One way to find this decomposition is by solving a matrix equation $Ad = c$ in $\mathbb{Z}_q$, where the matrix $A$ encodes a database of preselected locators and directors. However, it is important to note that this test is not exhaustive: failure to find such a decomposition does not necessarily mean that $c$ is not a locator. The effectiveness of the test depends on the choice of stored vectors, which should ideally form a basis sufficient to represent the space of locators.
\end{remark}

We conclude this section with Theorem \ref{theo:test-refreshable-ciphertexts}, which provides a straightforward test for verifying whether ACES ciphertexts are refreshable.

\begin{theorem}\label{theo:test-refreshable-ciphertexts}
Let $\mathsf{C} = (p,q,\omega,u)$ be an arithmetic channel, let $n$ be a positive integer, and let $x$ be an element of $\mathbb{Z}[X]^{(n)}$. For every element $m$ in $\mathbb{Z}_p$, every non-negative integer $k$, and every ciphertext $(c, c') \in \mathcal{S}_{\mathsf{C},k}^x(m|\sigma)$, if there exists a non-negative integer $k_0$ such that $\ell = \intbrackets{C}\tuplebrk{c} \in \mathcal{L}_{p,q}^{k_0}(x)$, then the following implication holds:
\[
\frac{p \cdot (k+1)-1}{q} < 1 - \mathsf{marg}_x(\ell) \quad\Rightarrow \quad  (c, c')\textrm{ is $p$-refreshable}.
\]
\end{theorem}

\begin{proof}
We assume that the conditions $\ell = \intbrackets{C}\tuplebrk{c} \in \mathcal{L}_{p,q}^{k_0}(x)$ and $\frac{p(k+1)-1}{q} < 1 - \mathsf{marg}_x(\ell)$ hold. Let us denote $z = \intbrackets{C}(c')$. To prove the theorem, we need to show that the following equation holds:
\begin{equation}\label{eq:refreshable:euclidean:v_prime:0}
\iota_q(z) + \iota_q\tuplebrk{-\ell}^T \big(\iota_q \circ \intbrackets{C}\tuplebrk{x}\big) = \iota_q\big(z-\ell^T\intbrackets{C}\tuplebrk{x}\big) + k_0pq.
\end{equation}
Before proving this equation, observe that our assumption $(c, c') \in \mathcal{S}_{\mathsf{C},k}^x(m|\sigma)$, where $z = \intbrackets{C}(c')$ and $\ell = \intbrackets{C}\tuplebrk{c}$, gives the identity:
\begin{equation}\label{eq:refreshable:euclidean:v_prime:z_expression}
z = \pi_q\Big(m + \iota_q\tuplebrk{\ell}^T \big(\iota_q \circ \intbrackets{C}\tuplebrk{x}\big) + kp\Big).
\end{equation}
We will use expression (\ref{eq:refreshable:euclidean:v_prime:z_expression}) to establish both sides of equation (\ref{eq:refreshable:euclidean:v_prime:0}). We begin by reformulating the right-hand side.

First, note that the inequality $\frac{p(k+1)-1}{q} < 1 - \mathsf{marg}_x(\ell)$ simplifies to $pk+(p-1) < q$. This implies that for every $m \in \mathbb{Z}_p$, the sum $m+pk$ remains in $\mathbb{Z}_q$. By Definition \ref{def:maps-for-modulo-operations}, we obtain the following equations:
\begin{align*}
\iota_q\big(z-\ell^T\intbrackets{C}\tuplebrk{x}\big) & = \iota_q\Big(m + \ell^T  \intbrackets{C}\tuplebrk{x} + kp - \ell^T\intbrackets{C}\tuplebrk{x}\Big) & (\text{from equation (\ref{eq:refreshable:euclidean:v_prime:z_expression})})\\
& = \iota_q(m+kp)\\
& = m+k_1p & (\text{since } m+kp < q).
\end{align*}
This provides the desired reformulation for the right-hand side of equation (\ref{eq:refreshable:euclidean:v_prime:0}).

Next, we reformulate the left-hand side of equation (\ref{eq:refreshable:euclidean:v_prime:0}). Note that equation (\ref{eq:refreshable:euclidean:v_prime:z_expression}) expresses $z$ as an image of the quotient map $\pi_q:\mathbb{Z} \to \mathbb{Z}_q$. Specifically, this implies that $z$ is the remainder in the Euclidean division of $m + \iota_q\tuplebrk{\ell}^T \big(\iota_q \circ \intbrackets{C}\tuplebrk{x}\big) + kp$ by $q$. Therefore, we have:
\begin{equation}\label{eq:refreshable:euclidean:v_prime:1}
m + \iota_q\tuplebrk{\ell}^T \big(\iota_q \circ \intbrackets{C}\tuplebrk{x}\big) + kp = q\left \lfloor \frac{\iota_q\tuplebrk{\ell}^T \big(\iota_q \circ \intbrackets{C}\tuplebrk{x}\big) + m + kp}{q} \right\rfloor + \iota_q(z).
\end{equation}
Since we have the inequality $\frac{kp+m}{q} < 1-\mathsf{marg}_x(\ell)$, it follows from Definition \ref{def:margin} that we can simplify the floored term in equation (\ref{eq:refreshable:euclidean:v_prime:1}) as:
\begin{equation}\label{eq:refreshable:euclidean:v_prime:2}
\left \lfloor \frac{\iota_q\tuplebrk{\ell}^T \big(\iota_q\circ\intbrackets{C}\tuplebrk{x}\big) + m+kp}{q}\right\rfloor = \left \lfloor \frac{\iota_q\tuplebrk{\ell}^T \big(\iota_q\circ\intbrackets{C}\tuplebrk{x}\big)}{q}\right\rfloor.
\end{equation}
Also, observe that the identity $\iota_q\tuplebrk{\ell}^T + \iota_q\tuplebrk{-\ell}^T = (q,q,\dots,q)$ holds. Using equation (\ref{eq:refreshable:euclidean:v_prime:2}), we simplify equation (\ref{eq:refreshable:euclidean:v_prime:1}) as follows:
\[
m + q\sum_i \iota_q \circ \intbrackets{C}(x_i) + kp = q\left \lfloor \frac{\iota_q\tuplebrk{\ell}^T \big(\iota_q\circ\intbrackets{C}\tuplebrk{x}\big)}{q}\right\rfloor + \iota_q(z) + \iota_q\tuplebrk{-\ell}^T \big(\iota_q \circ \intbrackets{C}\tuplebrk{x}\big).
\]
This simplifies further to:
\begin{equation}\label{eq:refreshable:euclidean:v_prime:3}
\iota_q(z) + \iota_q\tuplebrk{-\ell}^T \big(\iota_q\circ\intbrackets{C}\tuplebrk{x}\big)  = m + kp + q\left(\sum_i \iota_q \circ \intbrackets{C}(x_i) -\left \lfloor \frac{\iota_q\tuplebrk{\ell}^T \big(\iota_q\circ\intbrackets{C}\tuplebrk{x}\big)}{q}\right\rfloor\right).
\end{equation}
Since $\ell$ is a $p$-locator in $\mathcal{L}_{p,q}^{k_0}(x)$ (by Definition \ref{def:locators}), the term inside the large parentheses in equation (\ref{eq:refreshable:euclidean:v_prime:3}) is equal to $k_0p$. Furthermore, since we showed that $\iota_q\big(z-\ell^T\intbrackets{C}\tuplebrk{x}\big) = m+kp$, equation (\ref{eq:refreshable:euclidean:v_prime:3}) is equivalent to equation (\ref{eq:refreshable:euclidean:v_prime:0}). This concludes the proof.
\end{proof}

\begin{remark}[Non-refreshable ciphertexts]
So far, we have not discussed cases where a ciphertext $(c, c')$ is not refreshable. If such a situation occurs, we may want to transform it into a refreshable ciphertext while preserving the encoded message. One way to achieve this is by applying neutral homomorphic operations to $(c, c')$ so that the resulting ciphertext still encrypts the same message. For example, let $\mathsf{C} = (p, q, \omega, u)$ be an arithmetic channel, let $\sigma$ be an $n$-repartition of $q$, and let $x$ be an element of $\mathbb{Z}[X]^{(n)}$. For every element $m \in \mathbb{Z}_p$ and every non-$p$-refreshable ciphertext $(c, c') \in \mathcal{S}_{\mathsf{C},k}^x(m|\sigma)$, one can attempt to find another ciphertext $(c_0, c'_0) \in \mathcal{S}_{\mathsf{C},k_0}^x(0|\sigma)$ such that the sum $(c_0, c'_0) \oplus (c, c') \in \mathcal{S}_{\mathsf{C},k_0+k}^x(m|\sigma)$ is $p$-refreshable.
\end{remark}

\subsection{From leveled FHE schemes to proper FHE schemes}\label{ssec:proper-FHE}
This section introduces a refresh operation designed to decrease the encryption level linked to a ciphertext. First, let us extend our set of algebraic operations defined in Definition \ref{def:FHE:algebraic-operations} with a scalar product.
\smallskip

The following definition uses the implicit fact that the addition operation $\oplus$ defined in Definition \ref{def:FHE:algebraic-operations} is associative and commutative.

\begin{definition}[Scalar product]\label{def:scalar-product-ciphertexts}
Let $\mathsf{C} = (p,q,\omega,u)$ be an arithmetic channel, let $\sigma$ be an $n$-repartition of $q$, and let $x$ be an element in $\mathbb{Z}[X]^{(n)}$. Also, let us consider an element $\lambda \in \sigma \mathcal{H}(x|\mathsf{C},\sigma)$. For every positive integer $n_0$, consider the following:
\begin{itemize}
\item[1)] an $n_0$-tuple $\gamma_1 = ((c_{1,i},c_{1,i}'))_{i \in [n_0]}$ of elements in $\sigma\mathbb{Z}_q[X]_u \times \mathbb{Z}_q[X]_u$;
\item[2)] an $n_0$-tuple $\gamma_2 = ((c_{2,i},c_{2,i}'))_{i \in [n_0]}$ of elements in $\sigma\mathbb{Z}_q[X]_u \times \mathbb{Z}_q[X]_u$.
\end{itemize}
We define the $\lambda$-product of $\gamma_1$ with $\gamma_2$ as the following element $\gamma_1 \odot_{\lambda} \gamma_2$ defined in $\sigma\mathbb{Z}_q[X]_u \times \mathbb{Z}_q[X]_u$:
\begin{align*}
\gamma_1 \odot_{\lambda} \gamma_2 & = \big((c_{1,1},c_{1,1}') \otimes_{\lambda} (c_{2,1},c_{2,1}')\big) \oplus \dots \oplus \big((c_{1,n_0},c_{1,n_0}') \otimes_{\lambda} (c_{2,n_0},c_{2,n_0}')\big)\\
& = \bigoplus_{i=1}^{n_0} (c_{1,i},c_{1,i}') \otimes_{\lambda} (c_{2,i},c_{2,i}')
\end{align*}
\end{definition}

The following definition introduces a sequence of ciphertexts, called \emph{refresher}, that encrypts transformed information related to the secret key. This transformation reduces the risk of revealing the secret key and eliminates the need for deep circuit evaluation by matching the simplicity of the base decryption algorithm. The refresher structure also diverges from conventional bootstrapping techniques, as it is expected to use the general ciphertext formula from Definition \ref{def:encryption-space:general}, rather than the public-key-dependent encryption used in traditional bootstrapping. For these reasons, our approach notably differs from conventional bootstrapping techniques.

\begin{definition}[Refresher]\label{def:refresher}
Let $\mathsf{C} = (p,q,\omega,u)$ be an arithmetic channel, let $\sigma$ be an $n$-repartition of $q$, and let $x$ be an element in $\mathbb{Z}[X]^{(n)}$. We define a \emph{refresher} for the pair $(\mathsf{C},x)$ as an $n$-sequence $(\rho_1,\rho_1'),(\rho_2,\rho_2'), \dots, (\rho_n,\rho_n')$ of ciphertexts and an $n$-sequence $\kappa_1,\kappa_2,\dots,\kappa_n$ of non-negative integers such that for every $i \in [n]$, the following relation holds:
\[
(\rho_i,\rho_i') \in \mathcal{S}_{\mathsf{C},\kappa_i}^{x}\big(\pi_q \circ \iota_p \circ \pi_p \circ \iota_q \circ \intbrackets{\mathsf{C}}(x_i)\big|\sigma\big)
\]
Such a refresher structure will be denoted as a pair $(\kappa,\varrho)$ where $\kappa$ denotes the $n$-vector $(\kappa_1,\kappa_2,\dots ,\kappa_n)$ and $\varrho$ denotes the $n$-vector $\big((\rho_1,\rho_1'),(\rho_2,\rho_2'), \dots, (\rho_n,\rho_n')\big)$.
\end{definition}

The main idea behind Theorem \ref{theo:Yoneda:to:proper-FHE} (stated below) is to extend the result of Proposition \ref{prop:refreshable-ciphertexts} to the type of information made available by a refresher (Definition \ref{def:refresher}).

\begin{theorem}[Proper FHE]\label{theo:Yoneda:to:proper-FHE}
Let $\mathsf{C} = (p,q,\omega,u)$ be an arithmetic channel, let $\sigma$ be an $n$-repartition of $q$, and let $x$ be an element in $\mathbb{Z}[X]^{(n)}$. Also, let us consider $\lambda \in \sigma \mathcal{H}(x|\mathsf{C},\sigma)$ and a refresher $(\kappa,\varrho)$ for $(\mathsf{C},x)$. For every $m \in \mathbb{Z}_p$, every non-negative integer $k < (q + 1) / p - 1$, and every $p$-refreshable ciphertext $(c, c') \in \mathcal{S}^{x}_{\mathsf{C}, k}(m | \sigma)$, let $\gamma = ((c_{1,i}, c_{1,i}'))_{i \in [n]}$ where
\[
(c_{1,i}, c_{1,i}') \in \mathcal{S}_{\mathsf{C}, k_{1,i}}^{x}\big(\pi_q \circ \iota_p \circ \pi_p \circ \iota_q \circ\intbrackets{\mathsf{C}}(-c_i) \big| \sigma \big) \textrm{ for some integer } k_{1,i},
\]
and take
\[
(c_2, c'_2) \in \mathcal{S}_{\mathsf{C}, k_{2}}^{x}\big(\pi_q \circ \iota_p \circ \pi_p \circ \iota_q \circ\intbrackets{\mathsf{C}}(c') \big| \sigma \big)  \textrm{ for some integer } k_2.
\]
If we let
\[
\left\{
\begin{array}{ll}
\kappa_{\ast} &= \displaystyle k_2 +\sum_{i=1}^n p \cdot (\kappa_i+k_{1,i} + \kappa_i \cdot k_{1,i})\\
\kappa^{\ast} &= \displaystyle \left\lfloor \frac{(p-1) + n(p-1)^2}{p} \right\rfloor
\end{array}
\right.
\]
and assume that $\kappa_{\ast} < q/p$, then the ciphertext $(c_2, c'_2) \oplus (\gamma \odot_{\lambda} \varrho)$ belongs to the set $\mathcal{S}_{\mathsf{C}, \kappa_{\ast} + \kappa^{\ast}}^{x}(m | \sigma)$.
\end{theorem}
\begin{proof}
For convenience, we will use the following notations:
\[
\def\arraystretch{1.2}
\begin{array}{l}
\Gamma_1 = \pi_q \circ \iota_p \circ \pi_p \circ \iota_q \circ\intbrackets{\mathsf{C}}:\mathbb{Z}_q[X]_u \to \mathbb{Z}_q\\
\Gamma_2 = \pi_p \circ \iota_q \circ\intbrackets{\mathsf{C}}:\mathbb{Z}_q[X]_u \to \mathbb{Z}_p\\
\end{array}
\]
Now, since, for every integer $i \in [n]$, the inequality
$
p \cdot (k_{1,i}+\kappa_i+ \kappa_i \cdot k_{1,i}) \leq \kappa_{\ast} < q/p
$
holds, Theorem \ref{theo:Yoneda:to:leveled-FHE} implies that the ciphertext $(c_{1,i},c_{1,i}') \otimes_{\lambda} (\rho_i,\rho_i')$ is in the $\mathsf{C}$-encryption space
\[
\mathcal{S}_{\mathsf{C},p \cdot (k_{1,i}+\kappa_i+ \kappa_i k_{1,i})}^{x}\Big(\Gamma_1(-c_i) \cdot \Gamma_1(x_i) |\sigma\Big).
\]
Similarly, since the inequality
$
\sum_{i=1}^n p \cdot (\kappa_i+k_{1,i} + \kappa_i k_{1,i}) \leq \kappa_{\ast} < q/p
$
holds, Theorem \ref{theo:Yoneda:to:leveled-FHE} and Definition \ref{def:scalar-product-ciphertexts} implies that the ciphertext $\gamma \odot_{\lambda} \rho$ is in the $\mathsf{C}$-encryption space
\[
\mathcal{S}_{\mathsf{C},(\sum_{i=1}^n p \cdot (\kappa_i+k_{1,i} + \kappa_ik_{1,i}))}^{x}\Big(\Gamma_1\tuplebrk{-c}^T \Gamma_1\tuplebrk{x} |\sigma\Big).
\]
Again, since the inequality
$
k_2 +\sum_{i=1}^n p \cdot (\kappa_i+k_{1,i} + \kappa_ik_{1,i}) = \kappa_{\ast} < q/p
$
holds, Theorem \ref{theo:Yoneda:to:leveled-FHE} implies that the ciphertext $(c_2,c'_2) \oplus \big(\gamma \odot_{\lambda} \varrho\big)$ belongs to the $\mathsf{C}$-encryption space
\[
\mathcal{S}_{\mathsf{C},\kappa_{\ast}}^{x}\Big(\Gamma_1(c') + \Gamma_1\tuplebrk{-c}^T \Gamma_1\tuplebrk{x} |\sigma\Big).
\]
It follows from the homomorphic properties of the ring homomorphism $\pi_q:\mathbb{Z} \to \mathbb{Z}_q$ that the following equation holds.
\[
\Gamma_1(c') + \Gamma_1\tuplebrk{-c}^T \cdot \Gamma_1\tuplebrk{x} = \pi_q\Big(\Gamma_2(c') + \Gamma_2\tuplebrk{-c}^T \cdot \Gamma_2\tuplebrk{x}\Big)
\]
As a result, applying the composite function $\pi_p \circ \iota_q:\mathbb{Z}_q \to \mathbb{Z}_q$ on the previous equation gives the following identity.
\[
\pi_p \circ \iota_q\Big(\Gamma_1(c') + \Gamma_1\tuplebrk{-c}^T \cdot \Gamma_1\tuplebrk{x} \Big) = \Gamma_2(c') + \Gamma_2\tuplebrk{-c}^T \cdot \Gamma_2\tuplebrk{x}
\]
Since $(c,c')$ is $p$-refreshable, and we took $m \in \mathbb{Z}_p$ and $k < (q+1)/p-1$, it follows from Proposition \ref{prop:refreshable-ciphertexts} that the following identity holds.
\[
\pi_p \circ \iota_q\Big(\Gamma_1(c') + \Gamma_1\tuplebrk{-c}^T \cdot \Gamma_1\tuplebrk{x} \Big) = m
\]
It follows from the first part of the statement of Proposition \ref{prop:FHE:decryption} that the ciphertext $(c_2,c'_2) \oplus \big(\gamma \odot_{\lambda} \varrho\big)$ belongs to the $\mathsf{C}$-encryption space
\[
\mathcal{S}_{\mathsf{C},\kappa_{\ast}+\xi_p(\Gamma_1(c') + \Gamma_1\tuplebrk{-c}^T \Gamma_1\tuplebrk{x})}^{x}\Big(m|\sigma\Big).
\]
Now, recall that the function $\Gamma_1:\mathbb{Z}_q[X]_u \to \mathbb{Z}$ factors through the function $\pi_q \circ \iota_p:\mathbb{Z}_p \to \mathbb{Z}_q$. As a result, the images of $\Gamma_1$ are less than or equal to $p-1$. We deduce from this that the following inequality holds.
\[
\Gamma_1(c') + \Gamma_1\tuplebrk{-c}^T \Gamma_1\tuplebrk{x} \leq (p-1) + n(p-1)^2
\]
This means that the integer quotient  $\xi_p(\Gamma_1(c') + \Gamma_1\tuplebrk{-c}^T \Gamma_1\tuplebrk{x})$ of the element $\Gamma_1(c') + \Gamma_1\tuplebrk{-c}^T \Gamma_1\tuplebrk{x}$ by $p$ is bounded from above by the quantity $\kappa^{\ast}$. In other words, the ciphertext $(c_2,c'_2) \oplus \big(\gamma \odot_{\lambda} \varrho\big)$ belongs to the $\mathsf{C}$-encryption space $\mathcal{S}_{\mathsf{C},\kappa_{\ast}+\kappa^{\ast}}^{x}(m|\sigma)$.
\end{proof}


To conclude, the combined statements of Theorem \ref{theo:Yoneda:to:leveled-FHE}, Proposition \ref{prop:FHE:decryption} and Theorem \ref{theo:Yoneda:to:proper-FHE} show that the encryption scheme defined in section \ref{sec:FHE:from-Yoneda} defines a proper fully homomorphic encryption schemes.

\subsection{Security} The security of the cryptosystem outlined in section \ref{sec:FHE:from-Yoneda} is contingent upon the computational complexity of the LWE problem (see \cite{RLWE_def}). Specifically, the challenge lies in the difficulty of deducing the secret key $x$, as defined in Generation \ref{gen:FHE:Yoneda}, solely from the knowledge of the public key $(f, f')$ provided in Publication \ref{pub:FHE:Yoneda}. Breaking an ACES ciphertext would also effectively require to solve the LWE problem in integers (through the homomorphism of Proposition \ref{prop:channel-homomorphism}). Here, we make the assumption that the publication of a $3$-tensor $\lambda$ from the set $\sigma \mathcal{H}(x|\mathsf{C},\sigma)$ (defined in Definition \ref{def:zero-divisor-ideal:sigma}) does not provide enough information about the components of $x$ to compromise the secure full homomorphic property of ACES.

\subsection{Implementation}
The cryptosystem described in section \ref{sec:FHE:from-Yoneda} has an existing Python implementation tailored for the parameter value $\omega = 1$. While this implementation features a user-friendly suite of functions, it is important to note that it is not optimized for real-world applications. Its design is geared towards facilitating research and experimentation. The code for this package can be accessed on GitHub through the following link: \url{https://github.com/remytuyeras/aces}.

\bibliographystyle{plain} 
\bibliography{crypto} 



\end{document}